\documentclass[journal]{IEEEtran}
\usepackage{bm} 
\usepackage{float}
\usepackage{graphicx}
\usepackage{support-caption}
\usepackage{subcaption}
\usepackage{amsmath}

\usepackage[linesnumbered, ruled, vlined]{algorithm2e}
\usepackage[noend]{algpseudocode}

\usepackage{filecontents}
\usepackage{comment}
\usepackage{enumerate}
\usepackage{amsfonts}
\usepackage{amssymb}
\usepackage{verbatim}
\usepackage{amsthm}
\usepackage{setspace}
\usepackage{mathrsfs}
\usepackage{tikz}
\usetikzlibrary{arrows,shapes}
\newcommand\figref{Figure~\ref}
\usepackage{graphicx}
\usepackage{subcaption}

\usepackage{etoolbox}

\usepackage{todonotes}

\definecolor{green}{rgb}{0,0,0}
\definecolor{orange}{rgb}{0,0,0}
\definecolor{yellow}{rgb}{0,0,0}
\definecolor{blue}{rgb}{0,0,0}
\definecolor{revision}{rgb}{0,0,0}
\definecolor{revision2}{rgb}{0,0,0}

\newtheorem{theorem}{Theorem}

\newtheorem{lemma}[theorem]{Lemma}
\newtheorem{proposition}{Proposition}
\newtheorem{definition}{Definition}
\newtheorem{assumption}{Assumption}

\makeatletter

\newcommand{\Rmnum}[1]{\expandafter\@slowromancap\romannumeral #1@}

\makeatother

\begin{document}
\title{Resilient Distributed Diffusion \textcolor{green}{in Networks with  Adversaries}}
\author{ \textcolor{yellow}{Jiani~Li, ~\IEEEmembership{Student Member,~IEEE}, Waseem Abbas, and Xenofon~Koutsoukos,~\IEEEmembership{Fellow,~IEEE}}
        % <-this % stops a space}
\thanks{J.~Li, W.~Abbas and X.~Koutsoukos are with the Department
of Electrical Engineering and Computer Science at Vanderbilt University, Nashville, TN, USA, (jiani.li@vanderbilt.edu,waseem.abbas@vanderbilt.edu, xenofon.koutsoukos@vanderbilt.edu)}% <-this % stops a space
\thanks{ This work was supported in
part by the National Science Foundation under Grant CNS-1238959,
by the Air Force Research Laboratory under Grant FA 8750-14-2-
0180, and by the National Institute of Standards and Technology
under Grant 70NANB18H198. A subset of the results appeared in preliminary form in \cite{li2018resilient}.}% <-this % stops a space
\markboth{Journal of \LaTeX\ Class Files,~Vol.~14, No.~8, August~2015}%
{Shell \MakeLowercase{\textit{et al.}}: Bare Demo of IEEEtran.cls for IEEE Journals}}
\maketitle
%============================ Abstract ==============================
\begin{abstract}
%Distributed diffusion is a powerful algorithm for multi-task state estimation which enables networked agents to interact with neighbors to process input data and diffuse information across the network. Compared to a centralized approach, distributed diffusion offers multiple advantages that include robustness to node and link failures. 
\textcolor{revision}{In this paper, we study resilient distributed diffusion for multi-task estimation in the presence of adversaries where networked agents must estimate distinct but correlated states of interest by processing streaming data.} \textcolor{green}{We show that in general diffusion strategies are not resilient to malicious agents that do not adhere to the diffusion-based information processing rules.} In particular, by exploiting the adaptive weights used for diffusing information, we develop \textcolor{revision}{time-dependent} attack models that drive normal agents to converge to states selected by the attacker. %\textcolor{revision}{We show that assuming the attacker has a strong knowledge of the system, the attack is guaranteed to be successful. But even without such knowledge, it is possible that the attacker still succeed by a good approximation of normal agents' state.} 
\textcolor{revision}{We show that an attacker that has complete knowledge of the system can always drive its targeted agents to its desired estimates. Moreover, an attacker that does not have complete knowledge of the system including streaming data of targeted agents or the parameters they use in diffusion algorithms, can still be successful in deploying an attack by approximating the needed information.} The attack models can be used for both stationary and non-stationary state estimation. \textcolor{revision}{In addition, we present and analyze a resilient distributed diffusion algorithm that is resilient to any data falsification attack in which the number of compromised agents in the local neighborhood of a normal agent is bounded. The proposed algorithm guarantees that all normal agents converge to their true target states if appropriate parameters are selected.}
%In addition, we develop a resilient distributed diffusion algorithm \textcolor{revision}{that can achieve resilience to any data falsification attacks to up to $F$ compromised nodes in the neighborhood of one normal node, where $F$ can be any integer up to one's neighborhood size. The proposed algorithm guarantees that any normal nodes will converge to their true targets as long as $F$ is selected properly.}
%under the assumption that the number of compromised nodes in the neighborhood of each normal node is bounded. 
\textcolor{green}{%The proposed algorithm ensures that the attacked nodes do not drive normal nodes to incorrect estimates. 
We also analyze trade-off between the resilience of distributed diffusion and its performance in terms of steady-state mean-square-deviation (MSD) from the correct estimates.} Finally, we evaluate the proposed attack models and resilient distributed diffusion algorithm using stationary and non-stationary multi-target localization.

\end{abstract}
%============================= Keywords =============================
\begin{IEEEkeywords}
Resilient diffusion, multi-task estimation, \textcolor{green}{network topology}, \textcolor{green}{adaptive systems}
\end{IEEEkeywords}
\IEEEpeerreviewmaketitle
% \fi

%==================== Section : (Modified) Introduction ========================
\section {Introduction}
% Motivation
Diffusion Least-Mean Squares (DLMS) is a powerful algorithm for distributed state estimation \cite{journals/spm/SayedTCZT13}. It enables networked agents to interact with neighbors to process streaming data and diffuse information across the network to perform the estimation tasks. Compared to a centralized approach, distributed diffusion offers multiple advantages \textcolor{green}{including} robustness to drifts in the statistical properties of the data, scalability, \textcolor{green}{reliance} on local data, and fast response among others. Applications of distributed diffusion include spectrum sensing in cognitive networks \cite{7086338}, target localization \cite{targetLocalization}, distributed clustering \cite{6232902}, and biologically inspired designs for mobile networks \cite{mobileAdaptiveNetworks}. 

Diffusion strategies are known to be robust to node and link failures as well as to high noise levels 
\textcolor{revision2}{\cite{6197748, 5948418, 7547395,6726166}}.
%\textcolor{revision}{\cite{6197748, 5948418, 7472545,6726166}}.
\textcolor{revision}{However, it is possible that a single adversarial agent that does not update its estimates according to the diffusion-based information processing rules, for instance by retaining a fixed value throughout, can fail other agents to converge to their true estimates. 
%\textcolor{revision}{However, a single adversary with a fixed state can fail the diffusion network to converge to its actual state. 
Resilience of diffusion-based distributed algorithms in the presence of such fixed-value Byzantine attacks has been studied in \cite{journals/spm/SayedTCZT13,6232902}}. A general approach to counteract such attacks is to allow agents to fuse information collected from other agents in local neighborhoods \textcolor{revision}{using adaptive weights instead of fixed ones.
By doing so, only neighbors estimating a similar state will be assigned large weights so as to eliminate the influence of a fixed-value Byzantine adversary.}
%This approach improves resilience of the overall system, however, adaptive weights also introduce other  vulnerabilities within the network.}

% Problem
In this paper, we consider distributed diffusion for multi-task estimation where networked agents must estimate distinct, but correlated states of interest by processing streaming data. 
\textcolor{revision}{Agents use adaptive weights when diffusing information with neighbors since adaptive weights have been successfully applied to  multi-task distributed estimation problems. However, we are interested in understanding if adaptive weights introduce vulnerabilities that can be exploited by Byzantine adversaries. %, for instance by employing a time-dependent attack. 
The first problem we consider is to analyze if it is possible for an attacker to compromise a node, and make other nodes in its neighborhood converge to a state selected by the attacker. Then, we consider a network attack and determine a minimum set of nodes to compromise to make all nodes within the network converge to attacker's desired state.}

\textcolor{revision}{We assume a \emph{strong attack} model, that is, the attacker has complete knowledge of the network topology, streaming data of targeted agents and their parameters used in the diffusion algorithm. A strong attacker can know the topology by monitoring the network, streaming data of agents by stealthily compromising their sensors/controllers and establishing backdoor channels, and diffusion parameters by doing reverse engineering. We note that having complete knowledge is a strong assumption, however, it is common to assume a strong attacker with complete knowledge of the system to examine the resilience of distributed networks \cite{DBLP:journals/jsac/LeBlancZKS13, mitra2016secure, DBLP:journals/pomacs/ChenSX17, DBLP:conf/nips/BlanchardMGS17, DBLP:conf/nips/LiWSV16}. In addition to this strong attack model, we also consider a \emph{weak attack} model in which the attacker has no knowledge of streaming data of targeted agents or their parameters. We show that such an attacker can also be successful in preventing normal agents from converging to true estimates by approximating their states.} %We note that the proposed resilient diffusion algorithm is resilient to both strong and weak attacks.}

\textcolor{revision}{As a result, we show that DLMS, which was considered to be resilient against Byzantine agents by itself (\cite{journals/spm/SayedTCZT13, 6232902,5948418}), is in fact, not resilient. A Byzantine agent sharing incorrect estimates whose values are not fixed and change over time (time-dependent Byzantine attack) can manipulate the normal agents to converge to incorrect estimates. } %We highlight that by the first time, our work dispels the myth that baseline diffusion multi-task LMS by itself is resilient to adversarial attack, as indicated by \cite{journals/spm/SayedTCZT13, 6232902,5948418}.}
\textcolor{revision}{
On the one hand, adaptive weights improve the resilience of diffusion algorithms to fixed-value Byzantine attacks, but on the other hand, introduce vulnerabilities that can be exploited by time-dependent attacks.  We analyze this issue in detail and propose a resilient diffusion algorithm that ensures that normal agents converge to true final estimates} \textcolor{revision2}{in the presence of any data falsification attack.}
The main contributions of the paper are summarized below:
%\begin{enumerate}
%\item 

\noindent $\bullet$ By exploiting the adaptive weights, we develop attack models that drive normal agents to converge to states selected by an attacker. The attack models can be used to deceive a specific node or the entire network and are applicable to both stationary and non-stationary state estimation.
\textcolor{revision}{
Although the attack models are based on a strong knowledge of the system, we also show that the attack can succeed without such knowledge}.% to demonstrate the practicality of the proposed attack.}
%This for the first time dispels the myth that baseline diffusion multi-task LMS by itself is resilient to adversarial attack.}

%\item 
\noindent $\bullet$ \textcolor{revision}{We propose a resilient distributed diffusion algorithm parameterized by a positive integer $F$. We show that if there are at most $F$ compromised agents in the neighborhood of a normal agent}, then the algorithm guarantees that normal agents converge to their actual \textcolor{revision2}{goal} states under any data falsification attack. If the parameter $F$ selected by the normal agents is large,  the resilient distributed diffusion algorithm degenerates to non-cooperative estimation. Thus, we also analyze trade-off between the resilience of distributed diffusion and its performance degradation in terms of \textcolor{revision2}{the steady-state} MSD.

% \item We develop a resilient distributed diffusion algorithm \textcolor{revision}{that can achieve resilience to up to $F$ compromised nodes in one normal agent's neighborhood.}
% % under the assumption that the number of compromised nodes in the neighborhood of each normal node is bounded by $F$.
% \textcolor{green}{By selecting an appropriate $F$, the proposed algorithm guarantees normal agents will converge to their actual targets under any data falsification attacks.
% %ensures that compromised nodes are not able to drive normal nodes to incorrect estimates. 
% If the parameter $F$ selected by the normal agents is large,  the resilient distributed diffusion algorithm degenerates to non-cooperative estimation. Thus, we also analyze trade-off between the resilience of distributed diffusion and its performance degradation in terms of MSD.}

%\item
\noindent $\bullet$ We evaluate the proposed attack models for \textcolor{revision}{both strong and weak attacks} and the resilient \textcolor{revision2}{distributed diffusion} algorithm using both stationary and non-stationary multi-target localization. The simulation results are consistent with our theoretical analysis and show that the approach provides resilience to attacks while incurring performance degradation which depends on the assumption about the number of compromised agents. 
%\end{enumerate}

%Organization
The rest of the paper is organized as follows: Section \ref{preliminaries} briefly introduces distributed diffusion. Section \ref{sec:problem} presents the attack and resilient distributed diffusion problems. \textcolor{revision}{Sections \ref{the section of attack model} and \ref{the section of network attack model} discuss single node attack and network attack models respectively. Section \ref{sec:resilient_diffusion} presents and analyzes the resilient distributed diffusion algorithm. Section \ref{sec:evaluation} provides simulation results evaluating our approach\textcolor{revision2}{es with}  multi-target localization. Section \ref{sec:weak attack} discusses and evaluates the attack model that does not require complete knowledge of the system. Section \ref{sec:related_work} gives a brief overview of the related work and Section \ref{sec:con} concludes the paper.}

%==================== Section : Preliminaries ========================
\section{Preliminaries}\label{preliminaries}
We use normal and boldface fonts to denote deterministic and random variables respectively.
The superscript $(\cdot)^{*}$ denotes complex conjugation for scalars and complex-conjugate transposition for matrices, $\mathbb{E}\{\cdot\}$ denotes expectation,  and $\|\cdot\|$ denotes the Euclidean norm of a vector.

Consider a network of $N$ (static) agents\footnote{We use the terms agent and node interchangeably.}, \textcolor{green}{in which an undirected edge (or a link) between two agents indicates that they share information and are neighbors of each other. The neighborhood of an agent $k$, denoted by $\mathcal{N}_k$ is the set of neighbors of $k$, including the agent $k$ itself.} At each iteration $i$, agent $k$ has access to a scalar measurement $\bm{d}_{k}(i)$ and a regression vector $\bm{u}_{k,i}$ of \textcolor{yellow}{size $M$} with zero-mean and uniform covariance matrix $R_{u, k} \triangleq \mathbb{E}\{\bm{u}_{k,i}^* \bm{u}_{k,i}\} > 0$, which are related via a linear model of the following form:
\begin{equation*}\label{eq:1}
\bm{d}_{k}(i) = \bm{u}_{k,i} w_k^0 + \bm{v}_{k}(i).
\end{equation*}
where $ \bm{v}_{k}(i)$ represents a zero-mean i.i.d. additive noise 
with variance $\sigma^2_{v,k}$ and
$w_k^{0}$ denotes the unknown $M\times 1$ state vector of agent $k$. 

The objective of each agent is to estimate $w_k^{0}$ from (streaming) data
$\{\bm{d}_{k}(i),\bm{u}_{k,i}\} $ $(k=1,2,...,N, i \geq 0)$. 
The \textcolor{revision2}{objective state} can be static or dynamic and we represent \textcolor{revision2}{it} as $w_{k}^0$ or $\bm{w}^0_{k,i}$ respectively. 
For simplicity, we use $w_k^0$ to denote the objective state in both the static and dynamic cases. 

The state $w_k^0$ can be computed as the the unique minimizer of the following cost function:
\begin{equation}\label{eq: cost function}
J_k(w) \triangleq \mathbb{E} \{ \|\bm{d}_{k}(i)- \bm{u}_{k,i}w\|^2\}.
\end{equation}
An elegant adaptive solution for determining $w_k^0$ is the least-mean-squares (LMS) filter
\cite{journals/spm/SayedTCZT13}, 
where each agent $k$ computes successive estimators of $w_k^0$ without cooperation 
(noncooperative LMS) as follows:
\begin{equation*}
\bm{w}_{k,i} = \bm{w}_{k,i-1} + \mu_k \bm{u}_{k,i}^*[\bm{d}_{k}(i)-\bm{u}_{k,i}\bm{w}_{k,i-1}],
\end{equation*}
where $\mu_k>0$ is the step size (can be identical or distinct across agents).

Compared to noncooperative LMS, diffusion strategies introduce an aggregation step that incorporates \textcolor{green}{information gathered from the neighboring agents into the adaptation mechanism.} One powerful diffusion scheme is adapt-then-combine (ATC) \cite{journals/spm/SayedTCZT13} which optimizes the solution in a distributed and adaptive way using the following update:
\begin{equation}
\label{eq:adapt}
{\bm{\psi}}_{k,i} = \bm{w}_{k,i-1} + \mu_k \bm{u}_{k,i}^*[\bm{d}_{k}(i)-\bm{u}_{k,i}\bm{w}_{k,i-1}] \hspace{0.073in} \text{(adaptation)}
\end{equation}
\begin{equation}
\label{eq:combine}
 \bm{w}_{k,i} = \sum_{l \in \mathcal{N}_k} a_{lk}(i) \bm{\psi}_{l,i}\;, \hspace{1in} \text{(combination)}
\end{equation}
where $a_{lk}(i)$ represents the weight assigned to agent $l$ from agent $k$  
that is used to scale the data it receives from $l$, 
and the weights satisfy the following constraints:
\begin{equation}\label{eq: weight constraints}
 a_{lk}(i)\geq0,  \qquad \sum_{l \in \mathcal{N}_k} a_{lk}(i) = 1, \qquad a_{lk}(i)=0 \text{ if } l\not\in \mathcal{N}_k.
\end{equation}
\textcolor{revision}{Here the intermediate state ${\bm{\psi}}_{k,i}$ (obtained by the adaptation step) is shared among neighboring agents and a combination of neighbors' intermediate states contribute to the current estimate $ \bm{w}_{k,i}$ of agent $k$.}

In the case where agents estimate a common state $w^0$ (i.e., $w_k^0$ is same for every $k$), several \textcolor{revision}{fixed} combination rules can be adopted such as Laplacian, Metropolis, averaging, and maximum-degree \cite{DBLP:journals/corr/abs-1205-4220}. In the case of multiple tasks, agents are pursuing distinct but correlated objectives $w_k^0$. In this case, the combination rules mentioned above are not applicable because they simply combine the estimation of all neighbors without distinguishing if the neighbors are pursuing the same objective. An agent estimating a different state will prevent its neighbors from estimating the state of interest.

Diffusion LMS (DLMS) has been extended for multi-task networks in \cite{6232902}
using the following adaptive weights:
\begin{equation}\label{eq: adaptive relative-variance combination rule}
a_{lk}(i)=
\begin{cases}
\frac {\gamma_{lk}^{-2}(i)} {\sum_{m \in \mathcal{N}_k}\gamma_{mk}^{-2}(i) }, & l \in \mathcal{N}_k\\
0, & \text{otherwise}.
\end{cases}
\end{equation}
where $\gamma_{lk}^2(i) = (1-\nu_k)\gamma_{lk}^2(i-1)+\nu_k \| \bm{\psi}_{l,i}-\bm{w}_{k,i-1}\| ^2$ and $\nu_k$ is a positive step size known as the forgetting factor.
This update enables agents to continuously learn about the neighbors agents should cooperate with. During the estimation task, agents pursuing different objectives will continuously assign smaller weights to each other according to \eqref{eq: adaptive relative-variance combination rule}. Once the weights become negligible, communication links between agents do not contribute to the estimation task. Consequently, as the estimation proceeds, only the agents estimating the same state cooperate.

\begin{comment}
ATC DLMS with adaptive weights (DLMSAW) for multitask networks can be summarized 
in the following algorithm \cite{6232902}.
\renewcommand{\algorithmicrequire}{\textbf{Set}}
    \begin{algorithm}\small
        \begin{algorithmic}[1] 
             \Require $\gamma_{lk}^2(-1)=0$ for all $k=1,2,...,N$ and $l \in \mathcal{N}_k$
             \ForAll {$k=1,2,...,N,  i \geq 0$}
             \State $e_{k}(i) = \bm{d}_{k}(i)-\bm{u}_{k,i}\bm{w}_{k,i-1}$
             \State $\bm{\psi}_{k,i}=\bm{w}_{k,i-1}+\mu_k \bm{u}_{k,i}^* e_{k}(i)$
             \State $\gamma_{lk}^2(i) = (1-\nu_k)\gamma_{lk}^2(i-1)+\nu_k \| \bm{\psi}_{l,i}-\bm{w}_{k,i-1}\| ^2, l \in \mathcal{N}_k$
             \State ${a}_{lk}(i) = \frac {\gamma_{lk}^{-2}(i)} {\sum_{m \in \mathcal{N}_k}\gamma_{mk}^{-2}(i) }, l \in \mathcal{N}_k$
             \State $\bm{w}_{k,i} = \sum_{l \in \mathcal{N}_k} {a}_{lk}(i) \bm{\psi}_{l,i}$
            \EndFor
        \end{algorithmic}
        \caption{\small{ATC DLMSAW algorithm for multitask networks}} 
    \end{algorithm}
    
\end{comment}
    
DLMS with adaptive weights (DLMSAW) outperforms the noncooperative LMS as measured by the 
steady-state mean-square-deviation performance (MSD) \cite{journals/spm/SayedTCZT13}. 
For sufficiently small step-sizes, the network performance of noncooperative LMS is 
defined as the average \textcolor{revision2}{steady-state} MSD level \textcolor{revision2}{among agents}:  % \cite{journals/spm/SayedTCZT13}:
\begin{equation*}
\text{MSD}_{\text{ncop}} \triangleq \lim_{i \rightarrow \infty} \frac{1}{N} \sum_{k=1}^N \mathbb{E} \| \tilde{\bm{w}}_{k,i}\|^2 \approx \frac{\mu M}{2} \cdot (\frac{1}{N} \sum_{k=1}^N \sigma_{v,k}^2),
\end{equation*}
where $\tilde{\bm{w}}_{k,i} \triangleq w_k^0 - \bm{w}_{k,i}$  \textcolor{yellow}{and $M$ is the size of regression vector $\bm{u}_{k,i}$}.
The network MSD performance of the diffusion network (as well as the MSD performance of a normal agent in the diffusion network) can be approximated by % \cite{journals/spm/SayedTCZT13}:
\begin{equation}\label{eq:diffusion MSD}
\text{MSD}_{\text{k}} \approx \text{MSD}_{\text{diff}} \approx \frac{\mu M}{2} \cdot \frac{1}{N} \cdot (\frac{1}{N} \sum_{k=1}^N \sigma_{v,k}^2).
\end{equation}
In \cite{journals/spm/SayedTCZT13}, it is shown that $\text{MSD}_{\text{diff}} = \frac{1}{N} \text{MSD}_{\text{ncop}}$, which demonstrates an $N$-fold improvement of MSD performance.
%\todo[inline]{\checkmark Please specify what is $M$ in the above equations.}

%==================== Section : Problem Formulation ========================
\section{Problem formulation}\label{sec:problem}
Diffusion strategies have been shown to be robust to node and link failures as well as to nodes or links with high noise levels %\cite{5948418, 7472545}
\textcolor{revision2}{\cite{5948418, 7547395}}. 
In this paper, we are interested in understanding if the adaptive weights 
introduce vulnerabilities in the case a subset of nodes within the network is compromised by a cyber attack. \textcolor{green}{In this direction, first we analyze if it is possible for an attacker who has compromised a node $k$ to make nodes in $\mathcal{N}_k$ converge to a state selected by the attacker. Second,} we consider a network attack model \textcolor{green}{in which we determine a} minimum set of nodes to compromise to make the entire network converge to states selected by the attacker. Finally, we formulate the resilient distributed diffusion problem that guarantees that normal agents are not driven to the attackers' desired states, and continue the normal operation \textcolor{revision2}{with the cooperation among neighbors} possibly with a degraded performance.

\subsection{Single Node Attack Model}
We consider false data injection attacks \textcolor{revision}{deployed by a \emph{strong} attacker that has complete knowledge of the system. In particular, we assume the following for the strong attack.}
\textcolor{revision}{
\begin{assumption}
A strong attacker knows the topology of the network, the streaming data of targeted agents and the diffusion algorithm parameters they use, such as $\mu_k$.
\end{assumption}
}
 %)  \footnote{ \textcolor{revision}{To dispel the concern of not practical, a weak attack is presented in section \Rmnum{8} assuming the knowledge of communication message only,  as well as the evaluation results.}}.

\textcolor{revision}{To examine the resilience of distributed networks, it is common to assume a strong attack with full knowledge of the system, for instance, Byzantine attackers having a complete knowledge of the system are considered in \cite{DBLP:journals/jsac/LeBlancZKS13, mitra2016secure, DBLP:journals/pomacs/ChenSX17, DBLP:conf/nips/BlanchardMGS17, DBLP:conf/nips/LiWSV16}. However, we also consider a weak attack model in Section \Rmnum{8} in which an attacker has no knowledge of agents' parameters and has no access to their streaming data.}
% \textcolor{revision}{
% Note that it is a very strong assumption of attacker's knowledge, because the original diffusion algorithm assumes only the sharing of ${\bm{\psi}}_{k,i}$.
% However, when it comes to examine the resilience of distributed networks, it is common to assume a strong attack with full knowledge of the system. For instance, Byzantine attackers having a complete knowledge of the system are considered in \cite{DBLP:journals/pomacs/ChenSX17, DBLP:conf/nips/BlanchardMGS17, DBLP:conf/nips/LiWSV16}.}
Compromised nodes are assumed to be Byzantine in the sense that they can send arbitrary messages to their neighbors, and can also send different messages to different neighbors.

The objective of the attacker is to drive the normal nodes to converge to
a specific state. 
We assume a compromised node $a$ wants agent $k$ to converge to state 
\begin{equation*}
w_{k,i}^a=
\begin{cases}
w_k^a, &\text{for stationary estimation}\\
w_k^a + \theta_{k,i}^a, &\text{for non-stationary estimation}.
\end{cases}
\end{equation*}
%where $w_k^a$ and $\theta_{k,i}^a$ are the stationary and the dynamic model, respectively, selected by the attacker.
This is equivalent to minimizing the objective function of the following form:

\begin{equation}\label{eq: objective function}
\min_{\bm{w}_{k,i}} \lim_{i \rightarrow \infty} \mathcal{G}(\bm{w}_{k,i}), \qquad \bm{w}^a_{k,i} \in \textcolor{green}{\mathcal{D}_{w,k}},
\end{equation}
where
\begin{equation*}
    \mathcal{G}(\bm{w}_{k,i}) = \| \bm{w}_{k,i} - w^a_{k,i} \|^2,
\end{equation*}
and \textcolor{green}{$\mathcal{D}_{w,k}$} is the domain of state $\bm{w}_{k,i}$.

Another objective of the attacker can be to delay the convergence time of the normal agents. \textcolor{green}{We observe that} if the compromised node can make its neighbors to converge to a selected state, it can keep changing this state before normal neighbors converge. By doing so, normal neighbors \textcolor{green}{of the attacked node} will never converge to a fixed state. Thus, the attacker can achieve its goal to prolong the convergence time of normal neighbors. For that reason, we focus on the attack model based on objective \eqref{eq: objective function}. 

\subsection{Network Attack Model}
%Determining which nodes to compromise is another problem. 
If the attacker has a specific target node that she wants to attack and make it converge to a specific state, the attacker can compromise any neighbor of this node to achieve the objective. In the case the attacker wants to compromise the entire network and drive the multi-task estimation to specific states, she needs to \textcolor{green}{determine a minimum set of nodes to compromise such that every normal node in the network can be driven to an incorrect estimate. Computing such a minimum set directly depends on the underlying structure, and can be formulated as \emph{minimum dominating set problem} in graphs as discussed in Section \ref{the section of network attack model}.} 

%-----------------
%\subsection{Anomaly Detection}
%After realizing the vulnerabilities of the system, normal agents should take steps to counteract possible attacks. We consider each normal agent is equipped with a detector, and the detector's objective is to decide whether a sequence of output corresponds to normal behavior or an attack. 
%
%\textcolor{yellow}{Can we elaborate it a bit more, may be by adding 2-3 lines.}

\subsection{Resilient Distributed Diffusion}\label{attack detection}
Distributed diffusion is said to be \emph{resilient} if
\begin{equation}\label{eq: resilient distributed diffusion}
\lim_{i \rightarrow \infty} \bm{w}_{k,i} = w_k^0.
\end{equation}
for all normal agents $k$ in the network which ensures that all the noncompromised nodes converge to the true state. 
% \textcolor{revision2}{
% We know that without cooperation, agents can achieve resilience, yet with a possibly high steady-state MSD. 
% In other words, noncooperative LMS is a solution for \eqref{eq: resilient distributed diffusion}.
% However, we consider the problem of modifying DLMSAW to achieve resilience in the presence of cooperation, which possibly incurring a performance degradation as measured by the steady-state MSD level, yet the steady-state MSD level for it must outperform noncooperative LMS.
% }

\textcolor{revision2}{We note that if agents do not cooperate or interact with each other at all, such as in the non-cooperative diffusion, then adversary cannot impact agents' estimates. So, non-cooperative diffusion is resilient in this sense. At the same time, agents are also unable to utilize the information from other agents aiming to achieve the similar objective. Consequently, the steady-state MSD as result of non-cooperative diffusion can be quite large. Here, our objective is to design a resilient diffusion algorithm that guarantees convergence to the true estimates in the presence of adversary and also results in smaller MSD (as compared to the non-cooperative diffusion) by leveraging cooperation and information exchange between agents.}
We assume that in the neighborhood of a normal node, there could be at most $F$ compromised nodes~\cite{DBLP:journals/jsac/LeBlancZKS13}. 
Assuming bounds on the number of adversaries is typical for \textcolor{green}{the resiliency analysis of distributed algorithms,} 
\textcolor{revision2}{and our resilient algorithm is also based on this assumption.}

%========================== Section: Single Node Attack ============================
\section{Single Node Attack Design}\label{the section of attack model}

\textcolor{green}{We design a \textcolor{revision}{strong} attack in which the attacker drives the targeted node $k$ to converge to a wrong estimate $w_{k,i}^a$ by making $k $ follow a desired trajectory defined using stochastic gradient descent. The attacker's goal is to ensure that $k$, which implements adaptive-then-combine LMS, actually updates its estimates according to the stochastic gradient descent defined by the attacker. Thus, the main task is to determine conditions under which adaptive-then-combine LMS of $k$ % follows the attacker's desired approach that 
guarantees the convergence of $k$'s estimate to $w_{k,i}^a$. }

\textcolor{green}{
We summarize the conditions below and then analyze them in detail in the rest of the section.}
Firstly, an attacker  needs  to  know  the  estimate  of node $k$ in  the  previous  iteration.  \textit{Lemma  1} shows  that  an attacker  can obtain the estimate given node $k$'s  streaming  data and parameters.
Secondly, \textcolor{green}{Node $k$} should  not  assign  any  weight  to  the  messages from its  non-attacked  neighbors. %but all  the  weight  to  the  attacker's  value.  
\textit{Lemma  2} \textcolor{green}{ensures this objective}.
Finally, the magnitude of the stochastic gradient descent update should be sufficiently small.
Details are given in \textit{Proposition  1}.
%\item The  value  of $r^a_k$ (variable  used  in  the  gradient  descent  designed  by  the  attacker) should  be  sufficiently  small. See \textit{Proposition  1}.

%Next, we explain the above statement one by one.
\textcolor{revision}{\subsection{Gradient-based Attack Design}}
\textcolor{revision}{Here, we present an attack based on gradient-descent updates, and in the next subsection, provide conditions under which the attack is successful.}
\textcolor{green}{For stationary estimation, 
the following gradient-descent update with a sufficient small step size $\mu_k^a$ at the $i^{th}$ iteration is sufficient to achieve the objective in \eqref{eq: objective function}}:
%For stationary estimation, following gradient descent approach, the desired iterative state of $k$ by the attacker to reach the objective \eqref{eq: objective function} is
\begin{equation}\label{eq: attacker's goal}
\begin{aligned}
    \bm{w}_{k,i} &= \bm{w}_{k,i-1} - \mu_k^a \nabla_{\bm{w}} \mathcal{G}(\bm{w}_{k,i-1})\\
    &= \bm{w}_{k,i-1} - r_k^a (\bm{w}_{k,i-1} - w^a_{k,i}),
    \end{aligned}
\end{equation}
where $r_k^a = 2 \mu_k^a$ is a non-negative step size \textcolor{revision}{(that can also be time-varying).} For non-stationary estimation, the form is slightly different and it is described by\footnote{See proof of Proposition 1 in the Appendix.}
%which we describe below (more details can be found in \textit{Proposition 1}):
\begin{equation} \label{eq: attacker model}
    \bm{w}_{k,i} = \bm{w}_{k,i-1} - r_k^a (\bm{w}_{k,i-1} - x_i),
\end{equation}
where
\begin{equation*}
x_i=
\begin{cases}
w_k^a, &\text{for stationary estimation}\\
w_k^a + \theta_{k,i-1}^a + \frac{\Delta \theta_{k,i-1}^a}{r_{k,i}^a}, &\text{for non-stationary estimation}
\end{cases}
\end{equation*}
with $\Delta \theta_{k,i}^a = \theta_{k,i+1}^a - \theta_{k,i}^a$. 
And the diffusion estimate of $k$ is 
\begin{equation*}
\bm{w}_{k,i} = \sum_{l \in \mathcal{N}_k} {a}_{lk}(i) \bm{\psi}_{l,i}  = \sum_{l \in \mathcal{N}_k \backslash a} {a}_{lk}(i) \bm{\psi}_{l,i} + {a}_{ak}(i) \bm{\psi}_{a,i}.
\end{equation*}
It is sufficient to achieve the attack objective \eqref{eq: objective function} if the attacker could make the estimate of $k$ follow the gradient-descent trajectory, i.e.,
\begin{equation}\label{eq: assign state to desired}
        \sum_{l \in \mathcal{N}_k \backslash a} {a}_{lk}(i) \bm{\psi}_{l,i} + {a}_{ak}(i) \bm{\psi}_{a,i} = \bm{w}_{k,i-1} - r_k^a (\bm{w}_{k,i-1} - w^a_{k,i}).
\end{equation}
Since $\bm{\psi}_{l,i} =  \bm{w}_{l,i-1} + \mu_l \bm{u}_{l,i}^*[\bm{d}_{l}(i)-\bm{u}_{l,i}\bm{w}_{l,i-1}]$ is a random variable that is not controlled by the attacker, the attacker should eliminate the influence of $\bm{\psi}_{l,i}$ for $l \in \mathcal{N}_k, l \neq a$.
Sufficient conditions to hold \eqref{eq: assign state to desired}, and thus to achieve the attack objective are as follows:
\begin{equation}\label{eq:attack objective message}
\bm{\psi}_{a,i} = \bm{w}_{k,i-1} - r_k^a (\bm{w}_{k,i-1} - x_i).
\end{equation}
and
\begin{equation}\label{eq: weight condition}
    {a}_{lk}(i) \rightarrow 0,  \forall l \in \mathcal{N}_k, l \neq a, \textbf{ }\text{ and }
    a_{ak}(i) \rightarrow 1,
% \begin{aligned}
%     &{a}_{lk}(i) \rightarrow 0, \qquad \forall l \in \mathcal{N}_k, l \neq a,
%     \\
%     &a_{ak}(i) \rightarrow 1,
% \end{aligned}
\end{equation}
%Note that this is a sufficient condition \textcolor{green}{for the attacker to achieve its objective} \textcolor{yellow}{with the constraints listed in \textit{Lemma 2} and \textit{Lemma 3}}.
That is, the attacker uses the exchanging message $\bm{\psi}_{k,i}$ as indicated in \eqref{eq:attack objective message} and the targeted node $k$ updates its estimate based only on $\bm{\psi}_{k,i}$.
$\bm{\psi}_{k,i}$ is computed given the knowledge of $\bm{w}_{k,i-1}$, that can be obtained by the attacker given \textit{Lemma 1}.

\begingroup
\makeatletter
\apptocmd{\lemma}{\unless\ifx\protect\@unexpandable@protect\protect\footnote{The proofs can be found in the Appendix.}\fi}{}{}
\makeatother

\begin{lemma}\label{lemma: attacker can deduce w_k,i-1}
If a compromised node $a$ has a knowledge of node $k$'s streaming data $\{\bm{d}_k(i), \bm{u}_{k,i}\}$ and 
the parameter $\mu_k$, then it can compute % the exact value of
$\bm{w}_{k,i-1}$. 
\end{lemma}
\begin{comment}
\begin{proof}
The message received by $a$ from $k \in \mathcal{N}_a$ is $\bm{\psi}_{k,i}$. Agent $a$ can compute $\bm{w}_{k,i-1}$ from $\bm{\psi}_{k,i}$ using
\begin{equation*}
\bm{w}_{k,i-1} = \bm{\psi}_{k,i} - \mu_k \bm{u}_{k,i}^* (\bm{d}_k(i) - \bm{u}_{k,i} \bm{w}_{k,i-1})
\end{equation*}
from which it can compute $\bm{w}_{k,i-1}$ as:
\begin{equation*}
\bm{w}_{k,i-1} = \frac{\bm{\psi}_{k,i} - \mu_k \bm{u}_{k,i}^* \bm{d}_k(i)}{1 - \mu_k \bm{u}_{k,i}^* \bm{u}_{k,i}}
\end{equation*}
Given the knowledge of $\mu_k$, $\bm{d}_k(i)$, and $\bm{u}_{k,i}$, 
the value $\bm{w}_{k,i-1}$ can be computed exactly.
\end{proof}
\end{comment}
Next, we see that \textcolor{revision2}{by} carefully designing $\bm{\psi}_{a,i}$ as explained in \textit{Lemma 2}, conditions in \eqref{eq: weight condition} are satisfied.
\endgroup

\begingroup
\makeatletter
\apptocmd{\lemma}{\unless\ifx\protect\@unexpandable@protect\protect\fi}{}{}
\makeatother

\begin{lemma}
\label{lem:lemma_2}
If the attacker sends the message $\bm{\psi}_{a,i}$ satisfying
$\| \bm{\psi}_{a,i}-\bm{w}_{k,i-1}\| \ll \| \bm{\psi}_{l,i}-\bm{w}_{k,i-1}\|$\textcolor{revision2}{,}
$\forall l \in \mathcal{N}_k, l \neq a, \forall i$, then \eqref{eq: weight condition} will be true.
\end{lemma}
\endgroup

%-------- New Subsection -------
\textcolor{revision}{\subsection{Sufficient Conditions and Convergence Analysis}}
\textcolor{revision}{Here, using results \textcolor{revision2}{from} the previous subsection, we present conditions that guarantee a successful attack. A direct consequence of \textit{Lemma \ref{lem:lemma_2}} is that we could replace the condition in \eqref{eq: weight condition} by $\| \bm{\psi}_{a,i}-\bm{w}_{k,i-1}\| \ll \| \bm{\psi}_{l,i}-\bm{w}_{k,i-1}\|, \forall l \in \mathcal{N}_k, l \neq a, \forall i$. At the same time, from \eqref{eq:attack objective message}, we get 
$$
\| \bm{\psi}_{a,i}-\bm{w}_{k,i-1}\| =  \|  r_k^a (\bm{w}_{k,i-1} - x_i) \|.
$$}
Therefore, a sufficient condition to achieve the attack objective can be rewritten as
\begin{equation}\label{eq:attack objective message new}
\begin{aligned}
&\bm{\psi}_{a,i} = \bm{w}_{k,i-1} - r_k^a (\bm{w}_{k,i-1} - x_i), \\
s.t. \ & \|  r_k^a (\bm{w}_{k,i-1} - x_i) \| \ll \| \bm{\psi}_{l,i}-\bm{w}_{k,i-1}\|.
\end{aligned}
\end{equation}
Thus, the attacker has to select \textcolor{green}{a sufficiently} small value of $r_k^a$ to make \eqref{eq:attack objective message new} true. Note that even though $r_k^a = 0$ is sufficient for \eqref{eq:attack objective message new}, it renders the gradient of \eqref{eq: attacker's goal} zero and as a result no progress is made towards convergence to $w^a_{k,i}$. 
Also note that to use \eqref{eq:attack objective message new}, it is assumed that the communication message $\bm{\psi}_{l,i}$ from every $l \in \mathcal{N}_k$ is known by the attacker, which can be achieved by \textcolor{green}{intercepting the message}. In practice, a sufficiently small value of $r_k^a$ guarantees that the condition holds. 
The attacker can select a small $r_k^a$ and observe if the attack succeeds; if not, decrease $r_k^a$ to find an appropriate value.
It is also worth noting that for a fixed value of  $r_{k}^a$, \eqref{eq:attack objective message new} may not hold for some iteration $i$ because of the randomness of variables. Yet we can always set $r_{k}^a = 0$ for such iterations $i$ (no progress at the current point). However, in practice, the attack succeeds by using a small fixed value of $r_{k}^a > 0$ since estimation is robust to infrequent small values of $\| \bm{\psi}_{l,i} - \bm{w}_{k,i-1}\|$ caused by randomness given the smoothing property of the \textcolor{revision2}{adaptive} weight.

Next, we argue that \eqref{eq:attack objective message new} is sufficient to achieve the attack objective.
We summarize the above discussion in \textit{Proposition 1} and include a detailed proof in Appendix A.

\begingroup
\makeatletter
\apptocmd{\proposition}{\unless\ifx\protect\@unexpandable@protect\protect\fi}{}{}
\makeatother

\begin{proposition}\label{proposition: constraint of r_{k,i}^a}
%{\rm{(Sufficient condition)}} 
If $r_{k}^a > 0$ is selected such that 
$\forall l \in \mathcal{N}_k \cap l \neq a$,
$\forall i \geq i_a, $ $ \| r_{k}^a (\bm{w}_{k,i-1} - x_i)\| \ll \| \bm{\psi}_{l,i} - \bm{w}_{k,i-1}\|$,  
then the compromised node $a$ can realize the objective \eqref{eq: objective function} by using $\bm{\psi}_{a,i}$ described in \eqref{eq:attack objective message} as the communication message with $k$.
\end{proposition}
\endgroup

%\textcolor{revision}{\subsection{Attack Convergence Analysis}}
\textcolor{revision}{Next, we discuss the convergence time of attack.} Note that as $i\rightarrow \infty$,
\begin{equation*}
\lim_{ i \rightarrow \infty} (1 - r_k^a)^{i} = 0.\footnote{Refer to equation \eqref{eq: convergence condition} in the Appendix..}
\end{equation*}
In practice, when the left side of the above equation is smaller than a certain small value $\epsilon$, that is, 
\begin{equation*}
    (1 - r_k^a)^{i_c^a(\epsilon)} \leq \epsilon,
\end{equation*}
we consider that the \textcolor{green}{convergence to the desired state is achieved. Moreover, time required to reach the desired state is denoted by $i^a_c(\epsilon)$, and is computed as}
\begin{equation}\label{eq: attack convergence time}
    {i^a_c(\epsilon)} = \frac{\log \epsilon}{\log (1 - r_k^a)}.
\end{equation}

\textcolor{revision}{
It is also worth mentioning that it is not necessary to start the attack at the beginning of the diffusion task in order to guarantee the convergence of the attack. In other words, the attack can start at any time even after the diffusion algorithm has converged to its correct target as long as the condition in \textit{Proposition 1} is satisfied.
%Proofs can be found in Appendix A.
}

%========================== Section: Network Attack =================
\section{Network Attack Design}\label{the section of network attack model}
\textcolor{green}{In this section, we consider the case when multiple nodes are compromised using the attack model presented above. Our objective is to determine the minimum set of nodes to compromise in order to attack the entire network. 
For this, we show: (1) It is not necessary for the attacker to compromise multiple compromised nodes in order to attack a single node and (2) it is not possible for a compromised node to influence nodes, that is, make such nodes not converge to the desired states, that are not its immediate neighbors.
Therefore, the minimum set to compromise is simply a \emph{minimum dominating set} of the network, which we explain later in the section. }
%However, if the answers to either of the above problems is positive, then the network attack problem becomes harder. We begin with addressing the first problem.}
\textcolor{revision}{\subsection{Impact of Compromised Nodes on Normal Nodes}}
\textcolor{revision}{In this subsection, first we discuss the impact of multiple compromised nodes attacking a single normal node, and then analyze the impact of a compromised node \textcolor{revision2}{can make} beyond its immediate neighbors. }

\begingroup
\makeatletter
\apptocmd{\lemma}{\unless\ifx\protect\@unexpandable@protect\protect\fi}{}{}
\makeatother

\begin{lemma}
%\sout{With the \textcolor{yellow}{same message} as proposed in the last section,}
\textcolor{yellow}{If the compromised nodes send identical message as proposed in \eqref{eq:attack objective message}, then} multiple compromised nodes attacking one normal node is equivalent to one compromised node attacking the normal node.
\end{lemma}
%\todo[inline]{\checkmark What message. Can we refer to the particular equation?}
\begin{comment}
\begin{proof}
We use $\mathcal{A}$ to denote the set of compromised nodes targeting at the same normal node $k$.
The proposed attack strategy results in the following condition holding as proved in \textit{Lemma 2}:
\begin{equation*}
\begin{split}
    &\frac{a_{lk}(i)}{a_{ak}(i)} \rightarrow 0, \qquad  l \in \mathcal{N}_{k} \backslash \mathcal{A}, a \in \mathcal{A}  \\
    &(i \geq i_a + n, \text{ subject to } (1-\nu_k)^{n+1} = 0)
\end{split}
\end{equation*}
Given that $\sum_{l \in \mathcal{N}_k}a_{lk} = 1$, we have
\begin{equation*}
    a_{lk}(i) = 0, a_{ak}(i) = \frac{1}{|\mathcal{A}|}, \qquad  l \in \mathcal{N}_{k} \backslash \mathcal{A}, a \in \mathcal{A}, 
\end{equation*}
\textcolor{green}{where $|\mathcal{A}|$ denotes the number of nodes in $\mathcal{A}$}. Since every compromised node $a \in \mathcal{A}$ sends the same message and is assigned the same weight that sums up to 1, it is equivalent to only one compromised node attacking the target node and being assigned a weight of 1.
Therefore, there is no need for multiple compromised nodes attacking a single normal node.
\end{proof}
\end{comment}

\endgroup

%\textcolor{revision}{\subsection{Indirect Impact of Compromised Nodes}}

The next problem to consider is if a compromised node could indirectly impact its neighbors' neighbors that \textcolor{green}{at the same time} are not the neighbors of the  attacker $a$. To illustrate this, we consider an attacker node $a$, a normal node $l$, and a large clique\footnote{Every node is connected to every other node in a clique.} of normal nodes $\mathcal{C}$ such that each node in a clique is connected to both $a$ and $l$, and there is no edge between nodes $a$ and $l$.% as shown in \figref{fig: network}. \textcolor{green}{Note that there is no edge between $a$ and $l$.} 

% \definecolor{normalNode}{rgb}{0,0,0}
% \definecolor{clique}{rgb}{0,0,0}
% \definecolor{attackerNode}{rgb}{1,0,0}
% \begin{figure}[H]
% \centering
% \begin{tikzpicture}
%     [normalNode/.style={circle,   draw=black, line width=0.75pt, fill=normalNode!20,  minimum size=3mm},
%     attackerNode/.style={circle,draw=black,fill=attackerNode!30, line width=0.75pt, minimum size=3mm},
%     clique/.style={ellipse, draw = black, line width=0.75pt, fill=clique!20, minimum height=20mm, minimum width=5mm}
%     ]

%     %\fill[fill=gray!20] (0,0)--(2,1)--(2,-1);
%     %\fill[fill=gray!20] (4,0)--(2,1)--(2,-1);
    
%     \draw[color=black,line width=0.75pt] (0,0) -- (2,-1);
%     \draw[color=black,line width=0.75pt] (0,0) -- (2,1);
%     \draw[color=black,line width=0.75pt] (4,0) -- (2,-1);
%     \draw[color=black,line width=0.75pt] (4,0) -- (2,1);
    
%     \node[attackerNode] (a) at (0, 0){\textcolor{black}{}};
%     \node[normalNode] (l) at (4, 0){\textcolor{black}{}};
%     \node[clique] (c) at (2,0){\textcolor{black}{}}; 
    
%     \node[] at (0,0) 
%     {$a$};
%     \node[] at (4,0) 
%     {$l$};
%     \node[] at (2,0) 
%     {$\mathcal{C}$};
    
% \end{tikzpicture}
%     \caption{Illustration network}
%     \label{fig: network}
% \end{figure}

Using the proposed attack model, $a$ is able to drive every node in the clique to converge to its selected state. We are interested in finding if the normal node $l$, that is connected to the clique, is also affected by the attack.
The state of $l$ is obtained by
\begin{equation}\label{eq: network weight of l}
\begin{aligned}
    \bm{w}_{l,i} &= \sum_{k \in \mathcal{C}} a_{kl}(i) \bm{\psi}_{k,i}  + a_{ll}(i) \bm{\psi}_{l,i} \\
    &= \sum_{k \in \mathcal{C}} a_{kl}(i)( \bm{w}_{k,i-1} + \mu_k \bm{u}_{k,i}^*[\bm{d}_{k}(i)-\bm{u}_{k,i}\bm{w}_{k,i-1}] )  \\
    & \quad + a_{ll}(i) ( \bm{w}_{l,i-1} + \mu_l \bm{u}_{l,i}^*[\bm{d}_{l}(i)-\bm{u}_{l,i}\bm{w}_{l,i-1}] ).
    \end{aligned}
\end{equation}

\begin{comment}
as Fig. \ref{fig: compromised node affect not beyond neighborhood} illustrates. 
\begin{figure}[H]
    \centering
        \includegraphics[width=0.15\textwidth, trim=0cm 1.5cm 0cm 1.5cm]{figure_new/fig2.pdf}
        \caption{Illustrating example}
        \captionsetup{width=0.4\textwidth}
\label{fig: compromised node affect not beyond neighborhood}
\end{figure}

Without loss of generality, we set $\nu_k = 1$ and we use $\bm{R}_1$ and $\bm{R}_2$ 
to denote the two random variables $\mu_k \bm{u}_{k,i}^* e_{k}(i)$ and $\mu_l \bm{u}_{l,i}^* e_{l}(i)$. 
Then, for $i > i_a$, the weight assigned to node $k$ by node $l$ is given by
\begin{equation}\label{eq: weight lk}
\begin{aligned}
a_{kl}(i) &= \frac{\|\bm{\psi}_{k,i} - \bm{w}_{l,i-1}\|^{-2}}{\|\bm{\psi}_{k,i}- \bm{w}_{l,i-1}\|^{-2} +\|\bm{\psi}_{l,i} - \bm{w}_{l,i-1}\|^{-2} } \\
& = \frac{\|\bm{w}_{k,i-1} + \bm{R}_1 - \bm{w}_{l,i-1}\|^{-2}}{\|\bm{w}_{k,i-1} + \bm{R}_1 - \bm{w}_{l,i-1}\|^{-2} + \|\bm{R}_2\|^{-2}} \\
& = \frac{\|\bm{R}_2\|^{2}}{\|\bm{R}_2\|^{2} + \|\bm{w}_{k,i-1} + \bm{R}_1 - \bm{w}_{l,i-1}\|^{2}}.
\end{aligned}
\end{equation}
\end{comment}

We use $\bm{R}_{k,i}$ to denote the random variable $\mu_k \bm{u}_{k,i}^*[\bm{d}_{k}(i)-\bm{u}_{k,i}\bm{w}_{k,i-1}]$ for $k$ in the clique and $\bm{R}_{l,i}$ to denote $ \mu_l \bm{u}_{l,i}^*[\bm{d}_{l}(i)-\bm{u}_{l,i}\bm{w}_{l,i-1}]$ for normal node $l$. 
Suppose the compromised node $a$ could affect nodes beyond its neighborhood, 
from some point $i$, $\bm{w}_{k,i}$ converges to $w_k^a$ and $\bm{w}_{l,i}$ converges to $w_l^a$ (assume both $w_k^a \neq w_k^0$ and $w_l^a \neq w_l^0$). 

Thus, \eqref{eq: network weight of l} turns into:
\begin{equation}\label{eq:wla}
\begin{aligned}
   w_l^a &= \sum_{k \in \mathcal{C}} a_{kl}(i)( w_k^a + \bm{R}_{k,i} )  + (1 - \sum_{k \in \mathcal{C}} a_{kl}(i)) ( w_l^a + \bm{R}_{l,i} ) \\
   &= \sum_{k \in \mathcal{C}} a_{kl}(i) (w_k^a - w_l^a + \bm{R}_{k,i} - \bm{R}_{l,i}) + w_l^a + \bm{R}_{l,i}.
\end{aligned}
\end{equation}
\textcolor{green}{After inserting constants and random variables, \eqref{eq:wla} can be written as } 
%Put the constants and random variables at different sides of the equation, then \eqref{eq:wla} can be written as:
\begin{equation}\label{eq: false equation}
   \sum_{k \in \mathcal{C}} a_{kl}(i) ( w_l^a - w_k^a ) =  \sum_{k \in \mathcal{C}} a_{kl}(i)  \bm{R}_{k,i} + (1 - \sum_{k \in \mathcal{C}} a_{kl}(i) ) \bm{R}_{l,i}.
\end{equation}
Here, $(w_k^a - w_l^a)$ is a constant and $a_{lk}(i)$ changes slowly and can be considered as a constant that does not change within a small period of time. Then, \eqref{eq: false equation} implies a constant equals to a random variable, which does not hold except that both sides equal to zero. 
%Thus, proof by contradiction, we show that $l$ will not be converging to $w_l^a$, that is, $l$'s convergence state is not going to be tampered by $a$. 
For the left side, that is when $\sum_{k \in \mathcal{C}} a_{kl}(i) \rightarrow 0$ or  $( w_l^a - w_k^a ) \rightarrow 0$. Consider, when $( w_l^a - w_k^a ) \rightarrow 0$, that is, $ w_l^a \rightarrow w_k^a$. 
In such cases, 
\begin{equation*}
\begin{split}
    \bm{R}_{l,i}  &= \mu_l \bm{u}_{l,i}^* [\bm{d}_{l}(i)-\bm{u}_{l,i}\bm{w}_{l,i-1}] \\
    &= \mu_l \bm{u}_{l,i}^* [  \bm{u}_{l,i} w_l^0 + \bm{v}_{l}(i) -\bm{u}_{l,i} {w}_{l}^a]\\ 
    &= \mu_l \bm{u}_{l,i}^* [  \bm{u}_{l,i} (w_l^0 - {w}_{l}^a) + \bm{v}_{l}(i) ] \neq \bm{0}.
\end{split}
\end{equation*}
So is $\bm{R}_{k,i}$.
Therefore, equation \eqref{eq: false equation} does not hold under the condition $( w_l^a - w_k^a ) \rightarrow 0$.

The other possible solution for equation \eqref{eq: false equation} is when $\sum_{k \in \mathcal{C}} a_{kl}(i) \rightarrow 0$. This means $l$ does not assign any weight to $k \in \mathcal{C}$ and operates by itself. In such cases, equation \eqref{eq: false equation} holds when the right side of the equation is zero. Since $\sum_{k \in \mathcal{C}} a_{kl}(i) \rightarrow 0$, the right side turns into $\bm{R}_{l,i}$. We know when $l$ converges to its true objective state $w_l^0$, $\bm{R}_{l,i}$ is zero, i.e.,
\begin{equation*}
\begin{split}
    \bm{R}_{l,i}  &= \mu_l \bm{u}_{l,i}^* [\bm{d}_{l}(i)-\bm{u}_{l,i}\bm{w}_{l,i-1}] \\
    &= \mu_l \bm{u}_{l,i}^* [  \bm{u}_{l,i} w_l^0 + \bm{v}_{l}(i) -\bm{u}_{l,i} {w}_{l}^0] \\
    &= \mu_l \bm{u}_{l,i}^* \bm{v}_{l}(i) \rightarrow \bm{0}.
\end{split}
\end{equation*}
Thus, equation \eqref{eq: false equation} holds under two conditions: \textcolor{green}{First,} $\sum_{k \in \mathcal{C}} a_{kl}(i) \rightarrow 0$, that is, $l$ does not give any weight to $k \in \mathcal{C}$. \textcolor{green}{Second,} $\bm{R}_{l,i} \rightarrow 0$, that is, $l$ converges to its true objective state $w_l^0$. 

We note that the above two conditions indicate 
that
%\sout{are in fact \textcolor{yellow}{necessary} for $l$ to converge to its actual goal state. In other words, } 
$l$ %\sout{is sure to} 
converges to its original goal state and will not assign any weight to its \textcolor{green}{compromised neighbors under the above conditions.}
%The two conditions are  not contradictory. Therefore, these two conditions are the necessary conditions at the convergence state of $l$. Meaning $l$ is sure to converge to its original goal state and will not assign any weight to its neighbors that are attacked by the attacker.
Based on this discussion, we have \textit{Lemma 4}.
\begin{comment}
If $l$ converges to the true state, i.e., $\bm{w}_{l,i} = w_l^0$, then 
\begin{equation*}
\bm{d}_{l}(i+1)-\bm{u}_{l,i+1}\bm{w}_{l,i}  = \bm{u}_{l,i+1} w_l^0 + \bm{v}_{l}(i+1) -\bm{u}_{l,i+1}w_l^0 = \bm{v}_{l}(i+1)
\end{equation*}
Thus, 
\begin{equation*}
 \mathbb{E} \{ \bm{R}_{l,i} \}  = \mathbb{E} \{\mu_l \bm{u}_{l,i}^* [\bm{d}_{l}(i)-\bm{u}_{l,i}\bm{w}_{l,i-1}]\}= \mathbb{E} \{\mu_l \bm{u}_{l,i}^*\bm{v}_{l}(i)\}
\end{equation*}
Since $\bm{u}_{l,i}$ and $\bm{v}_{l}(i)$ are independent of each other, $\mathbb{E} \{\mu_l \bm{u}_{l,i}^*\bm{v}_{l}(i)\} = \mathbb{E} \{\mu_l \bm{u}_{l,i}^*\} \mathbb{E} \{\bm{v}_{l}(i)\}$. And because the noise is a zero-mean random variable, $\mathbb{E} \{ \bm{R}_{l,i} \} =  \mathbb{E} \{\mu_l \bm{u}_{l,i}^*\} \mathbb{E} \{\bm{v}_{l}(i)\} = 0$.
Yet if $l$ does not converge to its true objective state, $\mathbb{E} \{ \bm{R}_{l,i} \} \neq 0$. Besides, 
\end{comment}

\begin{lemma}
The attacker cannot change the convergence state of the nodes that are not its immediate neighbors.
\end{lemma}

\textcolor{revision}{Next, we see how many compromised nodes are needed to attack the entire network.}

\textcolor{revision}{\subsection{Minimum Set of Compromised Nodes to Attack the Entire Network}}

Since it is not  necessary  to use  more  than  one  compromised  nodes  to  attack  one  single normal  agent, and a compromised node cannot affect nodes beyond its neighborhood, finding a minimum set of nodes to compromise in order to attack the entire network is equivalent to finding a minimum dominating set of the network as defined below \cite{DBLP:journals/dam/HedetniemiLP86}. 
\begin{definition}
{\rm{(Dominating set)}} A dominating set of a graph $G = (V, E)$ is a subset $D$ of $V$ such that every vertex not in $D$ is adjacent to at least one member of $D$. 

\end{definition}
\begin{definition}
{\rm{(Minimum dominating set)}} A minimum dominating set of a graph is a dominating set of the smallest size.
\end{definition}
%\textcolor{green}{An example of a minimum dominating set is shown in \figref{fig:dom}.
\textcolor{green}{It should be noted that finding a minimum dominating set of a network is an NP-complete problem but approximate solutions using greedy approaches work well in practice (for instance, see \cite{DBLP:journals/dam/HedetniemiLP86}).} 
\begin{comment}
\begin{figure}[h!]
\centering
\includegraphics[scale=0.8]{Dom}
\caption{(a) Dominating and (b) minimum dominating set examples.}
%\caption{(a) A dominating set consisting of three (colored) nodes. (b) A minimum dominating set containing two (colored) nodes.}
%\captionsetup{width=0.4\textwidth}
\label{fig:dom}
\end{figure}
\end{comment}
With the above discussion, we state the following:
\begin{proposition}
The  compromised  nodes  need  to  form  a  dominating  set if the attacker  wants  every  node  in  the  network  to  converge  to its desired state.
\end{proposition}
Based on the above discussion, we observe that the above condition is both necessary and sufficient.

%============ Section: Resilient Distributed Diffusion ==============
\section{Resilient Distributed
Diffusion}\label{sec:resilient_diffusion}

\textcolor{revision}{In this section, we propose a resilient diffusion algorithm that guarantees convergence of normal nodes to their actual \textcolor{revision2}{states} if the number of compromised nodes in the neighborhood of a normal node is bounded. The proposed algorithm takes a \textcolor{revision2}{non-negative} integer $F$ as an input parameter. If the number of compromised nodes in the neighborhood of a normal node is at most $F$, then the algorithm is resilient to any such attack. It is obvious that selecting a large $F$ value achieves a higher level of resilience, while selecting $F=0$ means that the algorithm is not resilient to any attack. However, there exists a trade-off between the resilience and \textcolor{revision2}{the steady-state} MSD performance of the algorithm, which we will analyze in detail.} Since the proposed algorithm is adapted from the known DLMSAW, we call it a \emph{Resilient Diffusion Least Mean Square with Adaptive Weights (R-DLMSAW)}. \textcolor{revision2}{
We also note that in contrast to the connectivity requirements needed by resilient concensus problems \cite{DBLP:journals/jsac/LeBlancZKS13},
since in resilient diffusion, connectivity does not affect convergence, but only the estimation performance measured by the steady-state MSD.
%We also note that the resilience of the proposed algorithm does not require the underlying network topology to meet specific connectivity or robustness conditions, which are needed in the case of resilient consensus algorithms \cite{DBLP:journals/jsac/LeBlancZKS13}.
}

\textcolor{revision}{Since our algorithm can achieve resilience to up to $F$ compromised nodes,}
\textcolor{green}{
we assume that there can be at most $F$ compromised 
%(attacked) 
nodes in the neighborhood of any node, which is also referred to as the $F$-local model in \cite{DBLP:journals/jsac/LeBlancZKS13}. Specifically, we define:}

\begin{definition}\label{F-local definition}
{\rm{($F$-local model)}} A node satisfies the $F$-local model if there is at most $F$ compromised nodes in its neighborhood.
\end{definition}
\begin{definition}\label{F-local network}
{\rm{($F$-local network)}} A network is considered to
satisfy the $F$-local model if every node in the network has at most $F$ compromised nodes in its neighborhood.
\end{definition}

\textcolor{green}{%In general, normal nodes can select different values of $F$. 
While the paper focuses on the $F$-local model, scenarios involving bounds on the total number of compromised nodes within the network ($F$-total model \cite{DBLP:journals/jsac/LeBlancZKS13}) can also be analyzed using a similar approach. Next, we describe our resilient diffusion algorithm.}

\subsection{Resilient Diffusion Algorithm (R-DLMSAW)}

\textcolor{green}{In the context of distributed consensus, it is shown in \cite{DBLP:journals/jsac/LeBlancZKS13} \textcolor{revision2}{that for  Mean-Subsequence-Reduced (MSR) algorithms,} that during the state update phase, a node discards the values of neighbors that are too far off from the node's own value,  resilience against attacks can be achieved, that is, distributed consensus in the presence of compromised nodes ($F$-local and $F$-total models) is guaranteed. In distributed diffusion, we recall that a node updates its estimate by taking a weighted average of the estimates of all of its neighbors \eqref{eq:combine}. For resilient diffusion, we utilize a similar idea as in \cite{DBLP:journals/jsac/LeBlancZKS13}, that is instead of considering the estimates of all neighbors during the state update phase, only consider values from a subset of neighbors sharing close estimates. %that are selected based on their \emph{cost contribution} explained below. A node discards messages from all neighbors whose  cost contribution is too high. 
We show that this strategy guarantees convergence of normal nodes to true estimates. Before outlining the resilient distributed diffusion algorithm, we first explain the notion of \textcolor{revision2}{the} cost of a node.}

%Given the $F$-local model, a node $k$ has at most $F$ compromised neighbors. It is been studied in \cite{DBLP:journals/jsac/LeBlancZKS13} that given the bounds of the compromised nodes, the resilience of the state update can be achieved by  not considering the values of neighbors that are either too high or too low from the current estimate of the node. We take this idea but alternatively, we discard message from neighbors that has a large cost contribution $c_{lk}(i)$.
%-----------------------------------------------------

Following \eqref{eq:combine}, normal agent $k$ follows diffusion dynamics given by
\begin{equation*}
    \bm{w}_{k,i} = \sum_{l \in \mathcal{N}_k} a_{lk}(i) \bm{\psi}_{l,i}.
\end{equation*}

Thus, the cost function in \eqref{eq: cost function} in the $i^{th}$ iteration can be written as:
\begin{equation*}
    \begin{split}
J_k(\bm{w}_{k,i}) &= J_k(\sum_{l \in \mathcal{N}_k} a_{lk}(i) \bm{\psi}_{l,i}) \\
&= \mathbb{E} \{\|\bm{d}_k(i) - \bm{u}_{k,i} (\sum_{l \in \mathcal{N}_k} a_{lk}(i) \bm{\psi}_{l,i})\|^2\}.
    \end{split}
\end{equation*}
Since $\sum_{l \in \mathcal{N}_k} a_{lk}(i) = 1$, we have 
$$\bm{d}_k(i) =  \sum_{l \in \mathcal{N}_k} a_{lk}(i) \bm{d}_k(i).$$
Thus,
\begin{equation}\label{eq: cost contribution}
    \begin{split}
J_k(\bm{w}_{k,i}) &= \mathbb{E} \{ \sum_{l \in \mathcal{N}_k} a_{lk}(i) \bm{d}_k(i) -  \sum_{l \in \mathcal{N}_k} a_{lk}(i) \bm{u}_{k,i}  \bm{\psi}_{l,i} \|^2\}\\
&= \mathbb{E} \{ \|\sum_{l \in \mathcal{N}_k} a_{lk}(i) (\bm{d}_k(i) - \bm{u}_{k,i} \bm{\psi}_{l,i})\|^2\}\\
&=  \sum_{l \in \mathcal{N}_k} a_{lk}^2(i) \mathbb{E} \{ \| \bm{d}_k(i) - \bm{u}_{k,i} \bm{\psi}_{l,i}\|^2\} \\
&= \sum_{l \in \mathcal{N}_k} a_{lk}^2(i) J_k(\bm{\psi}_{l,i}) \\
 &= \textcolor{blue}{ \frac { \sum_{l \in \mathcal{N}_k} \gamma_{lk}^{-4}(i) J_k(\bm{\psi}_{l,i})} {[\sum_{m \in \mathcal{N}_k}\gamma_{mk}^{-2}(i)]^2} }.
    \end{split}
\end{equation}
%---------------------
%\todo[inline]{\checkmark Note that I have removed the subscript $g$ from $i$ as it had no significance at all, and was confusing.}
\begin{comment}
\textcolor{green}{From above, we note that the cost of $k$'s current state $J_k(\bm{w}_{k,i})$ is related to its neighbors' assigned weights $a_{lk}(i)$ and cost of their intermediate estimates $\bm{\psi}_{l,i}$. 
Further, we note from \eqref{eq: adaptive relative-variance combination rule} that the weight assigned to each neighbor $l$ is  proportional to $\gamma^{-2}_{lk}(i)$, that is, $a_{lk}(i) \propto \gamma^{-2}_{lk}(i)$. Consequently,} 
$$
a_{lk}^2(i) J_k(\bm{\psi}_{l,i})  \propto  \frac{J_k(\bm{\psi}_{l,i})}{\gamma_{lk}^4(i)}.
$$ 
\textcolor{green}{We then define the \emph{cost contribution} of $l$ to its neighbor $k$'s cost $J_k(\bm{w}_{k,i})$ as}
\begin{equation*}
    c_{lk}(i) \triangleq   \frac{J_k(\bm{\psi}_{l,i})}{\gamma_{lk}^4(i)}.
\end{equation*}
This value is an indication of the cost incurred to $k$ due to cooperation with $l$.
%The larger $c_{lk}(i)$ is, the more cost of $k$ is contributed by $l$.
\end{comment}

\textcolor{blue}{The goal of $k$ is to minimize its cost at every iteration, i.e., to minimize $J_k(\bm{w}_{k,i})$ by discarding $F$ neighbors' message. Therefore, the removal set $\mathcal{R}_{k}(i)$ of size $F$ should be selected by}
\textcolor{blue}{
\begin{equation*}
\begin{split}
        \mathcal{R}_{k}(i) &= \arg \min J_k(\bm{w}_{k,i})
        \\
        &= \arg \min \frac { \sum_{l \in \mathcal{N}_k \backslash \mathcal{R}_k(i)} \gamma_{lk}^{-4}(i) J_k(\bm{\psi}_{l,i})} {[\sum_{m \in \mathcal{N}_k \backslash \mathcal{R}_k(i)}\gamma_{mk}^{-2}(i)]^2}.
\end{split}
\end{equation*}
}

\textcolor{revision}{We note that the algorithm presented here is a generalization of the algorithm in \cite{li2018resilient} which is resilient to a specific type of Byzantine attack and has a lower computational cost. % complexity of $O({|\mathcal{N}_k|}^2)$. 
In contrast, the algorithm proposed in this work is resilient to any Byzantine attack, but has a higher computational cost. % of $O( {|\mathcal{N}_k|\choose F} )$. 
Thus, there is a trade off between the computation complexity of the algorithm and the scope of attacks to which the algorithm is resilient.}

% \textcolor{revision}{
% Note that in this work, we have generalized our previous work  \cite{li2018resilient} that achieves resilience to a single type of Byzantine attack as proposed in this paper, which has a computation complexity of $O(N_k^2)$. In contrast, the resilient algorithm proposed in this work handles the resilience problem to any type of Byzantine attacks, yet with a high computation complexity of $O( {|\mathcal{N}_k|\choose F} )$.
% Thus, there is a tradeoff between the scope of attacks we can achieve resilience and computation complexity we require.}

To compute the cost $J_k(\bm{\psi}_{l,i}) = \mathbb{E} \|\bm{d}_k(i) - \bm{u}_{k,i} \bm{\psi}_{l,i}\|^2$, agent $k$ has to store all the streaming data. 
Alternatively, we can approximate $J_k(\bm{\psi}_{l,i})$ using a moving average based on
the previous iterations.

Next, we outline the basic idea of the proposed resilient distributed diffusion algorithm below, and present the details of \emph{R-DLMSAW} in Algorithm 1.

%The basic idea of the  method is described in the following and we summarize the idea in \textit{Algorithm 1}, which we refer to as R-DLMSAW. Note that for $F = 0$, DLMSAW and R-DLMSAW are essentially same.
%\\
%\\
%\begin{minipage}[h]{0.46\textwidth}
%\fbox{
%\parbox{\textwidth}{
\begin{enumerate}
\item If $F \geq |\mathcal{N}_k|$, agent $k$ updates its current state $\bm{w}_{k,i}$ using only its own $\bm{\psi}_{k,i}$, which degenerates distributed diffusion to non-cooperative LMS.

\item \textcolor{blue}{If $F < |\mathcal{N}_k|$,  at each iteration $i$, agent $k$ computes $\binom{|\mathcal{N}_k|}{F}$ possible removal sets, and selects the one by removing which $J_k(\bm{\psi}_{l,i})$ is minimized.}
%computes $c_{lk}(i)$ for $l \in \mathcal{N}_k \text{ and } l \neq k$, sorts the results, and computes the set of  nodes $\mathcal{R}_k(i)$ consisting of $l$ with the $F$ largest $c_{lk}(i)$. 
Then, the agent updates its current weight $a_{lk}(i)$ and state $\bm{w}_{k,i}$ without using information from nodes in $\mathcal{R}_k(i)$.
\end{enumerate}

We note that for $F = 0$, DLMSAW and R-DLMSAW are essentially identical.

\begin{algorithm}
%\SetAlgoLined
\small
\DontPrintSemicolon
\KwIn{ $ \gamma_{lk}^2(-1)=0$ , maintain $n \times 1$ matrix $D_{k,i} = \bm{0}_{n \times 1}$ and $n \times M$ matrix $U_{k,i} = \bm{0}_{n \times M}$ for all $k=1,2,...,N$, and $l \in \mathcal{N}_k$ }
             \For{$k=1,2,...,N,  i \geq 0$}
             {
             $e_{k}(i) = \bm{d}_{k}(i) -\bm{u}_{k,i}\bm{w}_{k,i-1}$\;
             $\bm{\psi}_{k,i}=\bm{w}_{k,i-1}+\mu_k \bm{u}_{k,i}^* e_{k}(i)$\;
             \If{$ F \geq |\mathcal{N}_k|$} {
               $\bm{w}_{k,i} = \bm{\psi}_{k,i}$
             }\Else{
             $\gamma_{lk}^2(i) = (1-\nu_k)\gamma_{lk}^2(i-1)+\nu_k \| \bm{\psi}_{l,i}-\bm{w}_{k,i-1}\| ^2$\;
              Update $D_{k,i}$ and $U_{k,i}$ by adding $\bm{d}_k(i)$ and $\bm{u}_{k,i}$ and removing $\bm{d}_k(i-n)$ and $\bm{u}_{k,i-n}$\;
             $J_k(\bm{\psi}_{l,i}) = \mathbb{E}\|D_{k,i} - U_{k,i} \bm{\psi}_{l,i}\|^2$\;
             Compute all possible discarded set $\mathcal{R}_k(i)^1$, $\mathcal{R}_k(i)^2$, $\ldots$, $\mathcal{R}_k(i)^{\binom{|\mathcal{N}_k|}{F}}$ \;
             $J_{\min} = \infty$\;
             \For{$j = 1, 2, \ldots, \binom{|\mathcal{N}_k|}{F}$}{
              $J = \frac { \sum_{l \in \mathcal{N}_k \backslash \mathcal{R}_k(i)^j} \gamma_{lk}^{-4}(i) J_k(\bm{\psi}_{l,i})} {[\sum_{m \in \mathcal{N}_k \backslash \mathcal{R}_k(i)^j}\gamma_{mk}^{-2}(i)]^2} $\;
              \If{$J < J_{\min}$}
            { $\mathcal{R}_k(i) = \mathcal{R}_k(i)^j$\;
            $J_{\min} = J$\;}
            
              }
             
            ${a}_{lk}(i) = \frac {\gamma_{lk}^{-2}(i)} {\sum_{m \in \mathcal{N}_k \backslash \mathcal{R}_k(i)}\gamma_{mk}^{-2}(i) }, l \in \mathcal{N}_k \backslash \mathcal{R}_k(i)$\;
             $\bm{w}_{k,i} = \sum_{l \in N_{k} \backslash \mathcal{R}_k(i)} {a}_{lk}(i) \bm{\psi}_{l,i}$\;
           }
            \textbf{return} $ \bm{w}_{k,i} $ 
}
 \caption{\small{Resilient distributed diffusion under $F$-local bounds (R-DLMSAW)}}
\label{Algorithm 2}
\end{algorithm}

\begin{proposition}
\textcolor{blue}{If the network is a $F$-local network, then R-DLMSAW is resilient to any message falsification attack}% aiming at driving it to converge to a selected state
.
\end{proposition}
\begin{proof}
Given the $F$-local model, \textcolor{green}{we assume that there are $n\le F$ compromised nodes in the neighborhood of a normal node $k$.} In the case of $F \geq |\mathcal{N}_k|$, $k$ updates its state without using information from neighbors. Next, consider the case when $F < |\mathcal{N}_k|$. 
\textcolor{blue}{To deploy the attack, 
the attacker must try to make the message it sends to \textcolor{revision2}{the} normal nodes \textcolor{revision2}{not being discarded by the normal nodes}.
This can only be achieved if the cost of keeping the attacker's message is smaller than keeping some normal agents' message (discarding the attacker's message). 
Therefore, any attack message not being discarded actually results in a cost smaller than the normal case. Therefore, R-DLMSAW is resilient to any message falsification attack.
From the attacker's perspective, since its goal is to maximize cost $J_k(\bm{w}_{k,i})$, 
the optimal strategy for
the attacker is not to make this cost even smaller. As a result, 
the information from the attacker will be discarded.}

%The algorithm removes the $F$ largest cost contributions. 
%Based on the discussion in Section \ref{the section of attack model}, a necessary condition for the attack \textcolor{green}{to be effective} is given in \eqref{eq: weight condition}, that is, $ a_{lk}(i)  \rightarrow 0, a_{ak}(i) \rightarrow 1$, \textcolor{yellow}{for $l \in \mathcal{N}_k$ and $l \neq a$}. \textcolor{green}{As a result, in the case of an effective attack $c_{lk}(i)\rightarrow 0$.} Moreover, since we assume the attacker's desired state is different from the normal agent's original objective state, $J_k(\bm\psi_{a,i}) \gg 0$, thus $c_{ak}(i) \gg 0$, \textcolor{green}{which implies $c_{ak}(i) \gg c_{lk}(i)$.} Consequently, in iteration $i$, any compromised node $a$ \textcolor{green}{in the neighborhood of $k$, that is driving $k$ towards $w^a_{k,i}$} must be in the set of nodes with $n$ largest cost contributions. Since $n \leq F$, a compromised node must be in $\mathcal{R}_k(i)$, and therefore, its message is discarded. 
Thus,
%Consequently, in each iteration $i$, any compromised node $a \in \{a_1, a_2, \ldots\, a_n \}$ that drives $k$ towards $w^a_{k,i}$ must be in the set of nodes with $n$ largest cost contributions. Since $n \leq F$, a compromised node must be in $\mathcal{R}_k(i)$, and therefore, its message is discarded. Thus,
\begin{equation*}
    \bm{w}_{k, i} = \sum_{l \in \mathcal{N}_k \backslash \mathcal{R}_k(i)} a_{lk}(i) \bm{\psi}_{l,i},
\end{equation*}
meaning the algorithm performs the diffusion adaptation \textcolor{green}{step} as if there were no compromised node \textcolor{green}{in its neighborhood.} 
Note that messages from normal neighbors may \textcolor{green}{also} be discarded since $F$ may be greater than the number of compromised neighbors. 
However, the distributed diffusion algorithm is robust to node and link failures \textcolor{revision}{\cite{5948418}}, 
and it converges to the true state despite the links to some or all of its neighbors fail. 
%Therefore, we can take those discarded messages from well-behaved neighbors as link failures and the algorithm should work fine as well. 
Finally, the algorithm will converge and equation \eqref{eq: resilient distributed diffusion} holds, 
showing the resilience of R-DLMSAW.
\end{proof}

\begin{comment}
\subsection{Attacks against Resilient Distributed Diffusion}
If the number of compromised nodes satisfies the $F$-local model, 
R-DLMSAW is resilient to message falsification byzantine attacks aiming at 
driving normal nodes converge to a selected state. An important question is if there are attacks
against resilient distributed diffusion. 
The attacker could try to make the messages it sends to normal nodes not being discarded 
but affecting the convergence of normal agents. 
This must be achieved by selecting $c_{ak}(i)$ not to be one of the $F$ largest values
and thus be smaller than the value of some normal neighbor of $k$. 
In this case, $J_k(\bm{w}_{k,i})$ is even smaller than when this value is discarded
but the attacker's goal is to maximize $J_k(\bm{w}_{k,i})$. Thus, the optimal strategy for
the attacker is not to contribute cost less than a normal neighbor of $k$, and as a result, 
the information from a compromised node will be discarded.
\end{comment}

\subsection{Trade-off Between Resilience and MSD Performance}
\textcolor{green}{An important aspect of R-DLMSAW is the selection of parameter $F$ by each normal node. On the one hand, selection of a large $F$ degrades the performance of the diffusion algorithm as measured by the \textcolor{revision2}{steady-state} MSD, but on the other hand, a smaller $F$ might result in an algorithm that is not resilient against attacks.} In the following, we summarize the trade-off between \textcolor{revision2}{the steady-state}  MSD performance and resilience.

\textcolor{green}{It is rather obvious that if a normal node selects $F$ smaller than the number of compromised nodes in its neighborhood, then the messages from the compromised nodes might not be discarded entirely during the state update phase of R-DLMSAW. As a result, the algorithm might not be resilient against the attack, and the normal node might eventually converge to the attacker's desired state. However, if $F$ is selected too large,  then in the worst case, normal agents discard all the  information  from  their neighbors. The algorithm becomes a non-cooperative diffusion algorithm and incurs an $N$-fold MSD performance deterioration. Thus, the performance of R-DLMSAW lies somewhere in-between the cooperative diffusion and non-cooperative diffusion depending on the choice of $F$ selected.}

%\begin{proposition}
%When the number of compromised nodes in the neighborhood of the agent running \textit{Algorithm 1} is less than $F$, the MSD performance of the agent's network may incur deterioration.
%And the greater the difference between $F$ and number of compromised nodes in its neighborhood, the more likely a greater MSD deterioration being  produced.
%By a large $F$, the MSD level may incur degradation; Yet a small $F$ may render the algorithm not resilient.
%increases as $F$ increases beyond the number of the compromised nodes in its neighborhood.
%\end{proposition}
%\begin{proof}
%First, it is obvious that
%R-DLMSAW cannot ensure resilience if $F$ is selected smaller than the number 
%of compromised nodes in one normal agent's neighborhood. In such cases, message from 
%compromised nodes may not be entirely removed. 
%%To ensure an absolute resilience, $F$ should be selected as large as possible. 
%However, 
%if $F$ is selected too large,  then in the worst case, agents  discard  all  the  information  from  their neighbors. The algorithm  will degenerate  to noncooperative diffusion algorithm and incur an $N$-fold MSD performance deterioration.
%By selecting an appropriate $F$ value, R-DLMSAW lies somewhere in between the cooperative
%diffusion and non-cooperative diffusion depending on the choice of F picked.
%As we increase $F$, agents cut down more channels with their neighbors and the network is becoming more sparse.
%And as a result, the MSD level of the network is likely to increase.

Consider a connected network  with $N$ normal agents running R-DLMSAW.  Let $\sigma_{v,k}^2 = \{\sigma_{v,1}^2, \ldots, \sigma_{v,N}^2\}$ be the noise variance. \textcolor{green}{Suppose by selecting some $F$ the network is resilient, but is no longer a connected graph and is decomposed into $n$ connected sub-networks, each of which is denoted by $S_j$ where $j\in\{1,\cdots,n\}$.} Using \eqref{eq:diffusion MSD}, the \textcolor{revision2}{steady-state} MSD for each sub-network is 
\begin{equation*}
\text{MSD}_{S_j} \approx \frac{\mu M}{2} \cdot \frac{1}{({|S_j|})^2}  \sum_{k \in S_j} \sigma_{v,k}^2,
\end{equation*}
\textcolor{green}{where $|S_j|$ is the number of nodes in $j^{th}$ sub-network.}
The \textcolor{revision2}{steady-state} MSD for the overall network (consisting of sub-networks) after running R-DLMSAW is the weighted average of the \textcolor{revision2}{steady-state} MSD of the sub-networks, that is 
\begin{equation*}
    \text{MSD}_{\text{after}} = \frac{1}{N} \sum_{j = 1}^n \text{MSD}_{S_j} \cdot {|S_j|}   \approx \frac{\mu M}{2N} \cdot \sum_{j = 1}^n \frac{1}{{|S_j|}}  \sum_{k \in S_j} \sigma_{v,k}^2.
\end{equation*}
At the same time, the \textcolor{revision2}{steady-state} MSD for the (original) connected network before running R-DLMSAW is
\begin{equation*}
\text{MSD}_{\text{before}} \approx \frac{\mu M}{2} \cdot \frac{1}{N^2}  \sum_{k = 1}^{N} \sigma_{v,k}^2 \approx \frac{\mu M}{2N} \cdot \sum_{j = 1}^n \frac{1}{N}  \sum_{k \in S_j} \sigma_{v,k}^2. 
\end{equation*}
The difference between the two is
\begin{equation*}
    \text{MSD}_{\text{after}} - \text{MSD}_{\text{before}}  = 
    \frac{\mu M}{2N} \cdot \sum_{j = 1}^n (\frac{1}{{|S_j|}} - \frac{1}{N})  \sum_{k \in S_j} \sigma_{v,k}^2 .
\end{equation*}
We know that ${|s_j|} \leq N$. Therefore, $\frac{1}{{|S_j|}} - \frac{1}{N} \geq 0$, meaning the \textcolor{revision2}{steady-state} MSD of the network after running R-DLMSAW is worse than the \textcolor{revision2}{steady-state} MSD of the original network, and as the network is decomposed into more sub-networks, $\sum_{j = 1}^n (\frac{1}{{|S_j|}} - \frac{1}{N})$ and $\text{MSD}_{\text{after}}$ becomes larger.

\textcolor{green}{Therefore, it is crucial to select an appropriate $F$, that is a value with which the algorithm is resilient against compromised nodes and at the same time useful links between nodes are preserved. To this end, a simple way to select $F$ is to first estimate $\bm{w}_{\text{ncop},k,i}$ by a non-cooperative diffusion and compute $J_k(\bm{w}_{\text{ncop},k,i})$. Then, starting with a small $F$, for instance $F=0$, perform cooperative diffusion and compute $J_k(\bm{w}_{\text{coop},k,i})$. If $J_k(\bm{w}_{\text{coop},k,i}) > J_k(\bm{w}_{\text{ncop},k,i})$, it means that a compromised node is able to effect the estimation, and therefore increase $F$ by $1$. We keep repeating this as long as $J_k(\bm{w}_{\text{coop},k,i}) > J_k(\bm{w}_{\text{ncop},k,i})$ is true. }

%Therefore, it is essential for R-DLMSAW that $F$ is selected as small as possible but meanwhile to be resilient. 
%To this end, a simple way is to select a small $F$ at the beginning of the task, for instance, $F = 0$. Through operation, maintain a noncooperative estimation $\bm{w}_{\text{ncop},k,i}$, compare $J_k(\bm{w}_{\text{ncop},k,i})$ and $J_k(\bm{w}_{\text{coop},k,i})$ periodically, and increase $F$ by $1$ if $J_k(\bm{w}_{\text{coop},k,i})$ is strictly greater than $J_k(\bm{w}_{\text{ncop},k,i})$, which indicates the value of $J_k(\bm{w}_{\text{coop},k,i})$ is compromised by the attacker.

%====================== Section: Evaluation =========================
\section{Evaluation}\label{sec:evaluation}
In this section, we evaluate three algorithms, \emph{non-cooperative diffusion}, \emph{DLMSAW}, and \emph{R-DLMSAW}; and compare their performance  for \emph{no-attack} and \emph{attack} scenarios. We evaluate the proposed attack model and resilient  algorithms using the application of  multi-target localization 
\textcolor{revision2}{\cite{DBLP:journals/corr/abs-1205-4220, 6845334}}
%\textcolor{revision}{\cite{DBLP:journals/corr/abs-1205-4220, DBLP:conf/icassp/ChenRS14}} 
for both  stationary and non-stationary targets. %Based on that, we evaluate the proposed resilient algorithm.

\textcolor{revision}{We consider a network of $N=100$ agents, in which each agent's objective is to estimate the unknown location of its target of interest by the noisy observations of both the distance and the direction vector towards the target. These agents and targets are distributed in a plane. The location of agent $k$ is denoted by the two-dimensional vector $p_k = [x_k, y_k]^\top$, and similarly the location of target is represented by the vector $w_k^0 = [x_k^0, y_k^0]^\top$. Figure \ref{fig: target localization} illustrates how an agent estimates the location of the target.}

%We consider a network with $N=100$ agents. Figure \ref{fig: initial stationary network topology} shows the network topology before operating the diffusion algorithms. 
%\textcolor{revision}{For each agent, the objective is to estimate the unknown location of its target of interest by the noisy observations of both their distance and direction vector towards the target.}
%\textcolor{revision}{Agents are spread over the cartesian plane and the location of every agent $k$ is denoted
%by the $2\times1$ vector $p_k = [x_k, y_k]^\top$.}
 % (colored nodes in blue and green) 
%\textcolor{revision}{
%Targets in the same space are represented by the $2\times1$ vector $w_k^0 = [x_k^0, y_k^0]^\top$.
%The illustration of how one agent estimates the location of the target is given in \figref{fig: target localization}.}

\definecolor{agent}{rgb}{0,0,1}
\definecolor{target}{rgb}{0,1,0}
\begin{figure}[H]
\centering
\begin{tikzpicture}
    [agent/.style={circle,   draw=black, line width=0.75pt, fill=agent!30,  minimum size=4mm},
    target/.style={diamond,draw=black,fill=target!30, line width=0.75pt, minimum size=2mm}
    ]

    %\fill[fill=gray!20] (0,0)--(2,1)--(2,-1);
    %\fill[fill=gray!20] (4,0)--(2,1)--(2,-1);
    
    \draw[dotted, color=black,line width=1pt] (0,0) -- (3,2);
    \draw[-latex, color=red,line width=2pt] (0,0) -- (3/3,2/3);
    \draw[dotted, color=black,line width=1pt] (0,0) -- (3.5,0);
    
    \node[agent] (n) at (0, 0){\textcolor{black}{}};
    \node[target] (t) at (3,2){\textcolor{black}{}}; 
    
    \node[] at (4, 2.5) 
    {target};
    \node[] at (-1.2,0.2) 
    {agent $k$};
    \node[] at (1.4, 1.4) 
    {$r_k^0$};
    \node[] at (1.1, 0.3) 
    {$u_k^0$};
    \node[] at (-1.2, -0.3) 
    {$[x_k, y_k]^\top$};
    \node[] at (4.1, 2) 
    {$[x_k^0, y_k^0]^\top$};

\end{tikzpicture}
    \caption{Illustration of target localization.}
    \label{fig: target localization}
\end{figure}
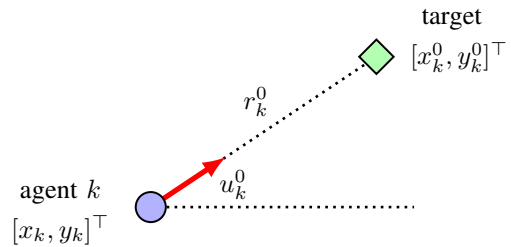

\textcolor{revision}{In Figure \ref{fig: target localization}, the distance between agent $k$ and the target is denoted by $r_k^0 = \|w_k^0 - p_k\|$, and the unit direction vector from agent $k$ to the target is $u_k^0 = \frac{(w_k^0 - p_k)^\top}{\|w_k^0 - p_k\|}$. Therefore, the relationship holds such that $r_k^0 = u_k^0 (w_k^0 - p_k)$. Since agents have only noisy observations $\{ \bm{r}_k(i), \bm{u}_{k,i} \}$ of the distance and direction vector at every iteration $i$, we get the following:
\begin{equation*}
    \bm{r}_k(i) = \bm{u}_{k,i} (w_k^0 - p_k) + \bm{v}_{k}(i).
\end{equation*}
If we use the adjusted signal $\bm{d}_k(i)$, such that
\begin{equation*}
    \bm{d}_k(i) = \bm{r}_k(i) + \bm{u}_{k,i} p_k,
\end{equation*}
then we derive the following linear model for variables $\{\bm{d}_k(i), \bm{u}_{k,i}\}$ in order to estimate the target $w_k^0$:
\begin{equation*}
    \bm{d}_k(i) = \bm{u}_{k,i} w_k^0  + \bm{v}_{k}(i).
\end{equation*}
%This is the linear model diffusion LMS can be used to approach.
As a result, agents can rely on DLMSAW algorithm for the multi-target localization problem. Figure \ref{fig: initial stationary network topology} shows the network topology before the application of diffusion algorithm\textcolor{revision2}{s}. For better readability, we only illustrate the network topology of agents without showing targets.}

For stationary target localization, the location of the two stationary targets are given by
\begin{equation*}
w_{k}^0=
\begin{cases}
[0.1, 0.1]^\top, & \text{for } k \text{ depicted in blue}\\
[0.9, 0.9]^\top, & \text{for } k \text{ depicted in green}
\end{cases}
\end{equation*}
~
Non-stationary targets are given by
\begin{spacing}{0.4}
\begin{equation*}
\bm{w}_{k,i}^0 =
\begin{cases}
\begin{bmatrix}
    0.1 + 0.1 \cos(2\pi\omega i) \\
    0.1 + 0.1 \sin(2\pi\omega i)
\end{bmatrix}
, \text{for } k \text{ depicted in blue} \\
\\
\begin{bmatrix}
    0.9 + 0.1 \cos(2\pi\omega i) \\
    0.9 + 0.1 \sin(2\pi\omega i)
\end{bmatrix}
, \text{for } k \text{ depicted in green}
\end{cases}
\end{equation*}
\\
\end{spacing}
\noindent where $\omega = \frac{1}{2000}$.

Regression data is white Gaussian with diagonal covariance matrices 
$R_{u,k} = \sigma_{u,k}^2 I_M$ with $M = 2$, $\sigma_{u,k}^2 \in [0.8, 1.2]$
and noise variance $\sigma_k^2 \in [0.15, 0.2]$. The step size of $\mu_k = 0.01$ and the forgetting factor $\nu_k = 0.01$ are set uniformly across the network. 
%And the initial state $w_{k,-1}$ are randomly generated. The total simulation iterations is 4000. 
\textcolor{revision}{Note that we adopt a signal-to-noise ratio (SNR) of $5-10$ dB in our setup. However, the same results are generated if we choose low SNR values.}

\textcolor{revision}{\subsection{Strong Attacks}}
\textcolor{revision}{
%For simplicity, we refer to the attack with strong knowledge as strong attack in comparison to the attack without strong knowledge which is referred to as weak attack (will be discussed and evaluated in section \ref{sec:weak attack}).
\textcolor{revision}{We consider the strong attack model discussed in Sections \ref{the section of attack model} and \ref{the section of network attack model}}. The attacker aims at making the normal agents estimate a specific location selected by the attacker. In this evaluation, we select the attacker's targeted location to be $w_k^a = [0.5, 0.5]^\top$, and}
the attack parameters are selected uniformly across the compromised agents as $r_{k}^a = 0.002$. \textcolor{revision}{For non-stationary estimation, we select} $\theta_{k,i}^a = [0.1 \cos(2\pi\omega_a i),$ $0.1 \sin(2\pi\omega_a i)]^\top$, $\Delta \theta_{k,i}^a = [-0.2 \pi \omega_a \sin(2\pi\omega_a i),$ $0.2 \pi \omega_a \cos(2\pi\omega_a i)]^\top$, where $\omega_a = \frac{1}{2000}$. 
Figure \ref{fig:without attack, after simulation} shows the network topology at the end of the simulation using DLMSAW with no attack for \textcolor{yellow}{both stationary and non-stationary tasks}. 
%\todo[inline]{I think it shows only for one.
%\textcolor{green}{The resulted networks are the same for stationary and non-stationary.}}
If the weights between agents $k$ and $l$ are such that $a_{lk}(i) < 0.01$ and $a_{kl}(i) < 0.01$, \textcolor{green}{then we remove the link between such nodes \textcolor{revision2}{from the network}. We observe that only the links between agents estimating the same target are kept, that is green nodes are connected with green nodes only, and blue nodes are connected with only blue ones, thus, illustrating the robustness of DLMSAW in multi-task networks.} 
\textcolor{revision}{\figref{fig: state convergence 1} and \figref{fig: state convergence 2} shows the estimation dynamics by DLMSAW for the targets' locations $\bm{w}_{k,i}(1)$ and $\bm{w}_{k,i}(2)$ for every agent $k$ and iteration $i$ from 0 to $5000$ under no attack. Here $\bm{w}_{k,i}(1)$ and $\bm{w}_{k,i}(2)$ are the first and second \textcolor{revision2}{element} of the estimate respectively, that is $\bm{w}_{k,i} = [ \bm{w}_{k,i}(1), \bm{w}_{k,i}(2)]^\top$. It is shown that the two groups of nodes converge to their goal state.} 

\textcolor{revision}{
\figref{fig: initial topology with attacker} shows the initial network topology with compromised nodes.}
There are four compromised nodes (red nodes with yellow centres) in the network.
%and attack parameters are selected uniformly across the compromised agents as $w_k^a = [0.5, 0.5]^\top$ and $r_{k}^a = 0.002$. 
Figure \ref{fig:attacked after simulation} shows the network topology at the end of DLMSAW in the case of \textcolor{revision}{a strong attack}. 
\textcolor{green}{All red nodes are the normal agents converging to $w_k^a$. We observe that neighbors of a compromised node communicate only with the compromised node, and not with any other node in the network. As a result, compromised nodes successfully drive all of their neighbors to desired states $w_k^a$ as discussed in Section \ref{the section of network attack model}.}
\textcolor{revision}{\figref{fig: state convergence 3} and \figref{fig: state convergence 4} shows the estimation dynamics by DLMSAW for the targets' location $\bm{w}_{k,i}(1)$ and $\bm{w}_{k,i}(2)$ for every agent $k$ and iteration from 0 to $5000$ under attack.
The attacked nodes in the figure refer to the immediate neighbors of the compromised nodes. 
%, where $\bm{w}_{k,i} = [ \bm{w}_{k,i}(1), \bm{w}_{k,i}(2)]^\top$. 
It is shown that all the immediate neighbors of compromised nodes are driven to converge to $w_k^a$ whereas all the other normal nodes converge to their original goal states.}

%------------ Begin Figure -------------------------
\begin{figure*}
    \setlength{\abovecaptionskip}{0.1cm}  
    \begin{subfigure}[t]{0.23\textwidth}
    \centering
        \includegraphics[width=0.9\textwidth, trim=1.7cm 1.7cm 1.7cm 1.7cm]{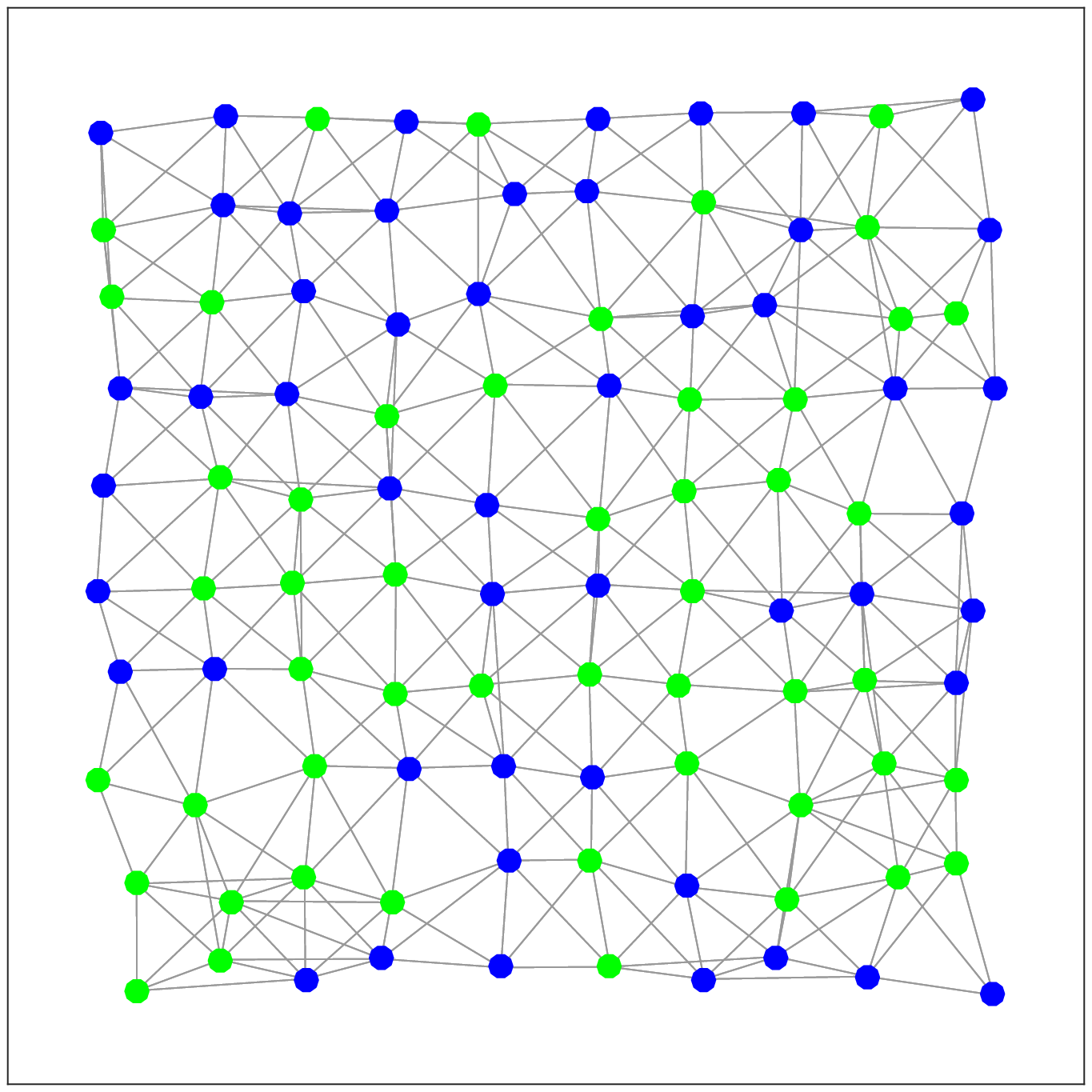}
        \caption{Initial network topology (no compromised nodes)}
        \label{fig: initial stationary network topology}
    \end{subfigure}
    ~ 
    \begin{subfigure}[t]{0.23\textwidth}
    \centering
        \includegraphics[width=0.9\textwidth, trim=1.7cm 1.7cm 1.7cm 1.7cm]{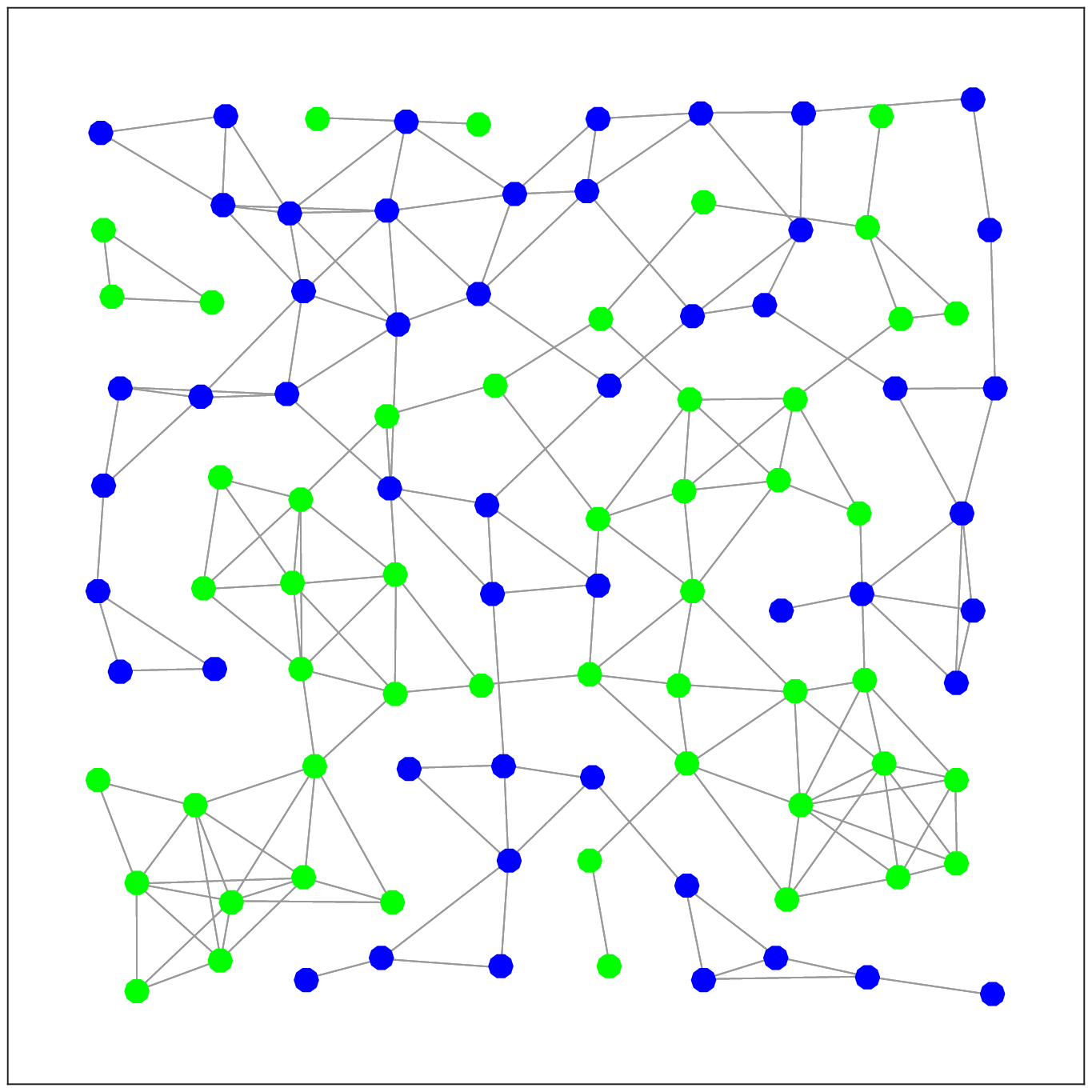}
        \caption{At the end of DLMSAW with no attack}
        \label{fig:without attack, after simulation}
    \end{subfigure}
    ~
     \begin{subfigure}[t]{0.23\textwidth}
    \centering
        \includegraphics[width=0.9\textwidth, trim=1.7cm 1.7cm 1.7cm 1.7cm]{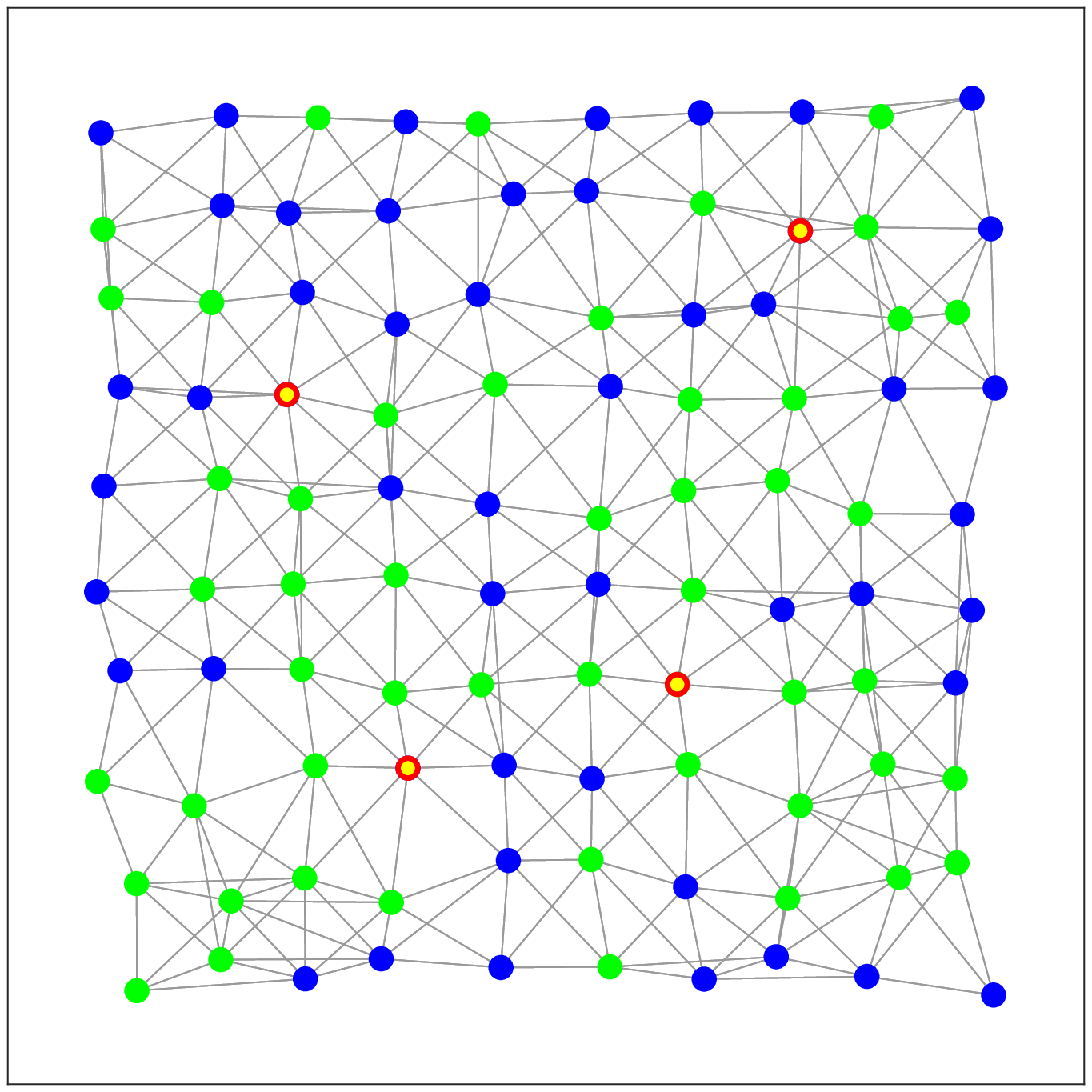}
        \caption{Initial network topology (with compromised nodes)}
        \label{fig: initial topology with attacker}
    \end{subfigure}
~
    \begin{subfigure}[t]{0.23\textwidth}
    \centering
        \includegraphics[width=0.9\textwidth, trim=1.7cm 1.7cm 1.7cm 1.7cm]{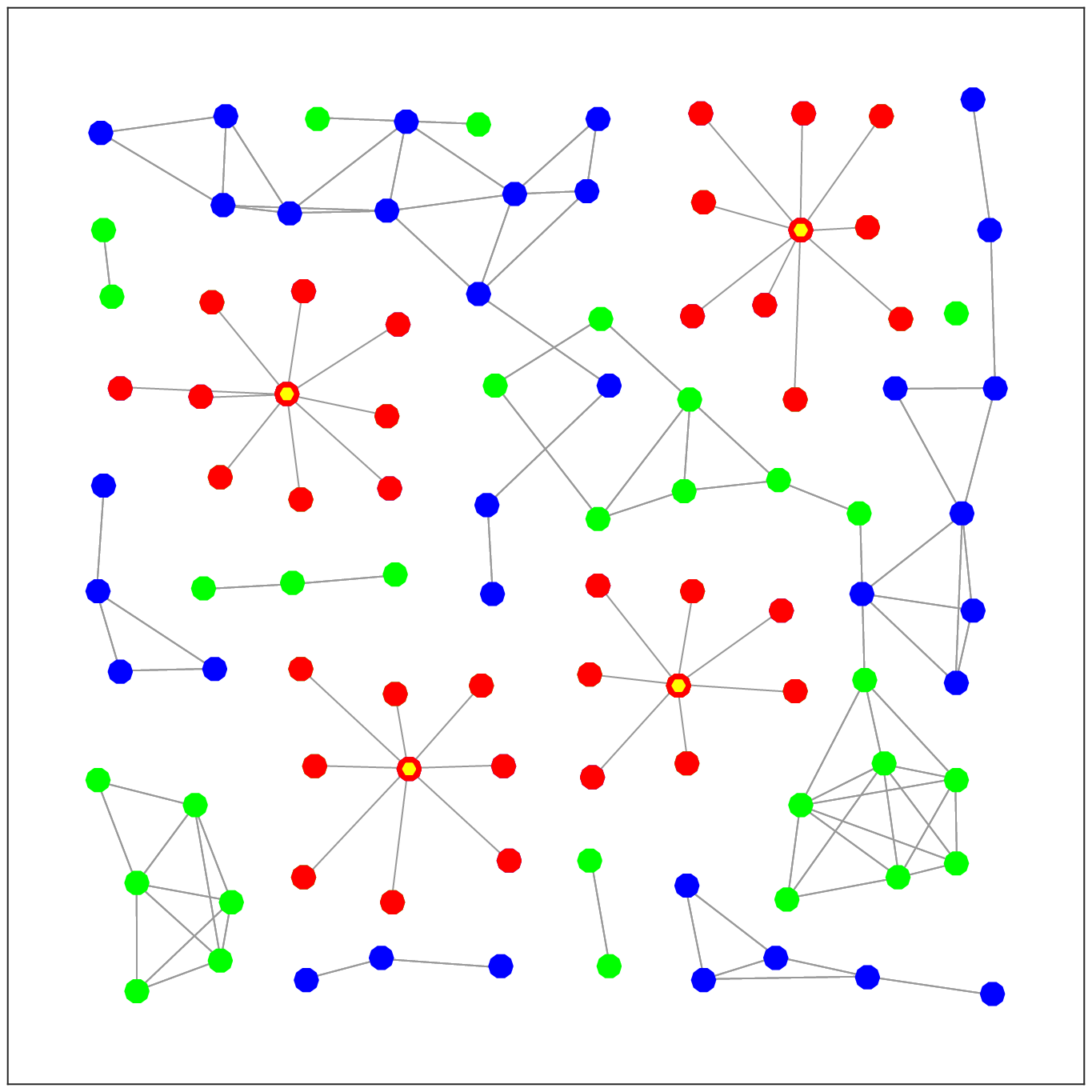}
        \caption{At the end of DLMSAW under strong attack}
        \label{fig:attacked after simulation}
    \end{subfigure}
    \caption{\textcolor{green}{Network topologies in the case of DLMSAW algorithm.}}\label{fig:topology}
\end{figure*}
%------------ End Figure -------------------------

%------------ Begin Figure -------------------------
\begin{figure*}
    \centering
    \setlength{\abovecaptionskip}{0.1cm}  
    \begin{subfigure}[t]{0.235\textwidth}
    \centering
        \includegraphics[width=0.9\textwidth, trim=1.5cm 1.5cm 1.5cm 1.5cm]{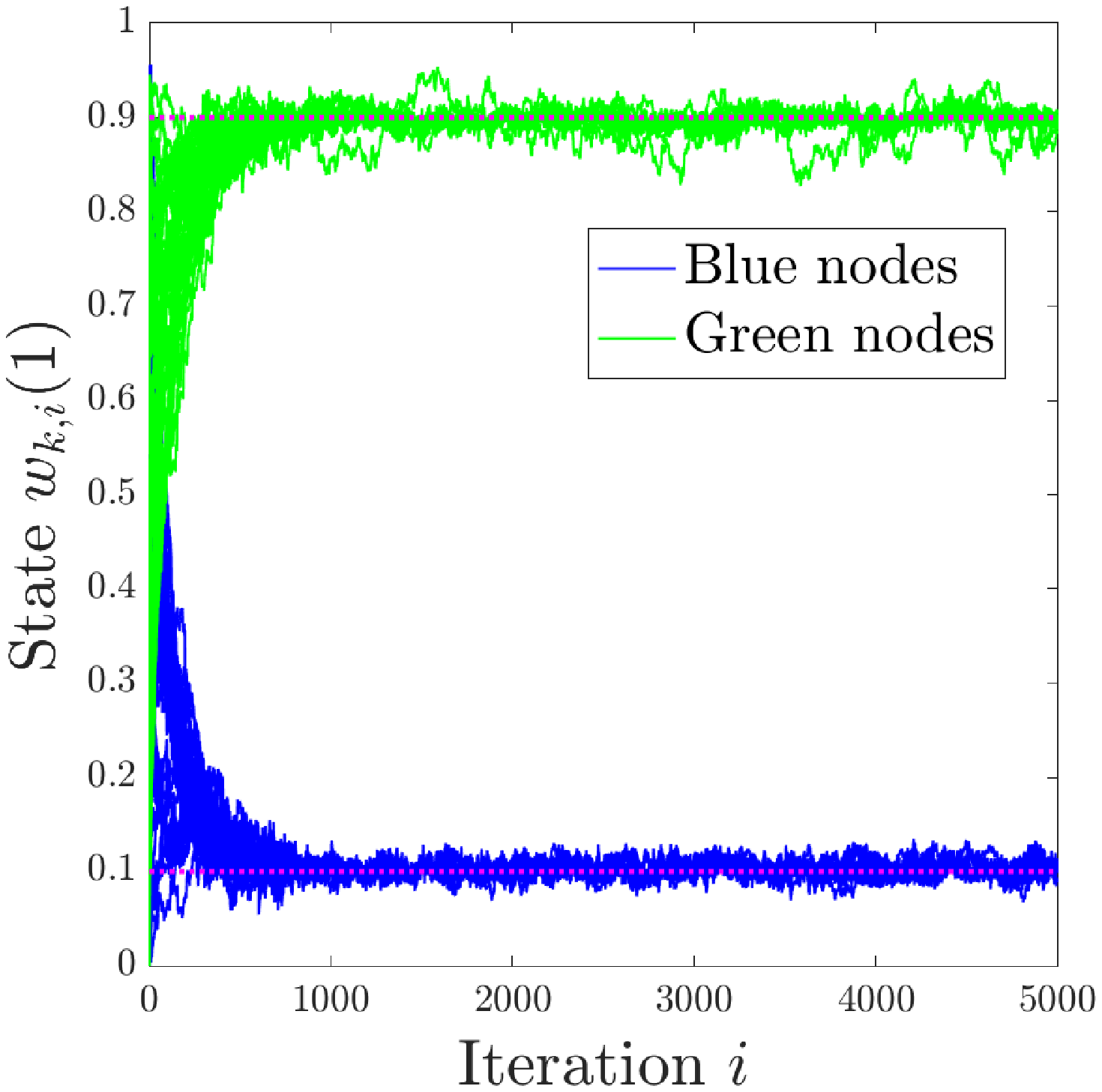}
        \vspace{0.3cm}
        \caption{$\bm{w}_{k,i}(1)$ (under no attack)}
        \label{fig: state convergence 1}
    \end{subfigure}
    ~
    \begin{subfigure}[t]{0.235\textwidth}
    \centering
        \includegraphics[width=0.9\textwidth,  trim=1.5cm 1.5cm 1.5cm 1.5cm]{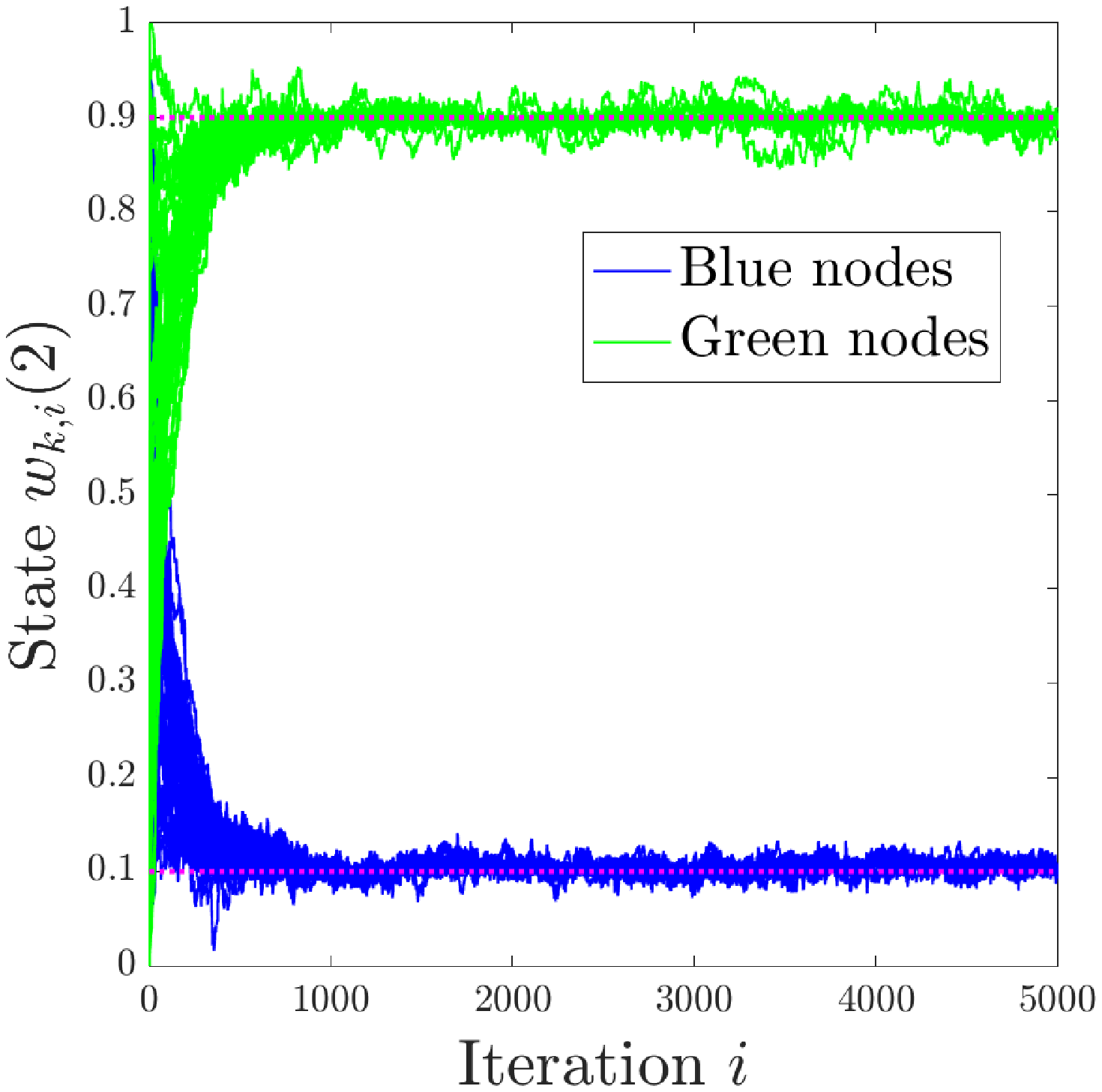}
         \vspace{0.3cm}
        \caption{$\bm{w}_{k,i}(2)$ (under no attack)}
        \label{fig: state convergence 2}
    \end{subfigure}
    ~
    \begin{subfigure}[t]{0.235\textwidth}
    \centering
        \includegraphics[width=0.9\textwidth, trim=1.5cm 1.5cm 1.5cm 1.5cm]{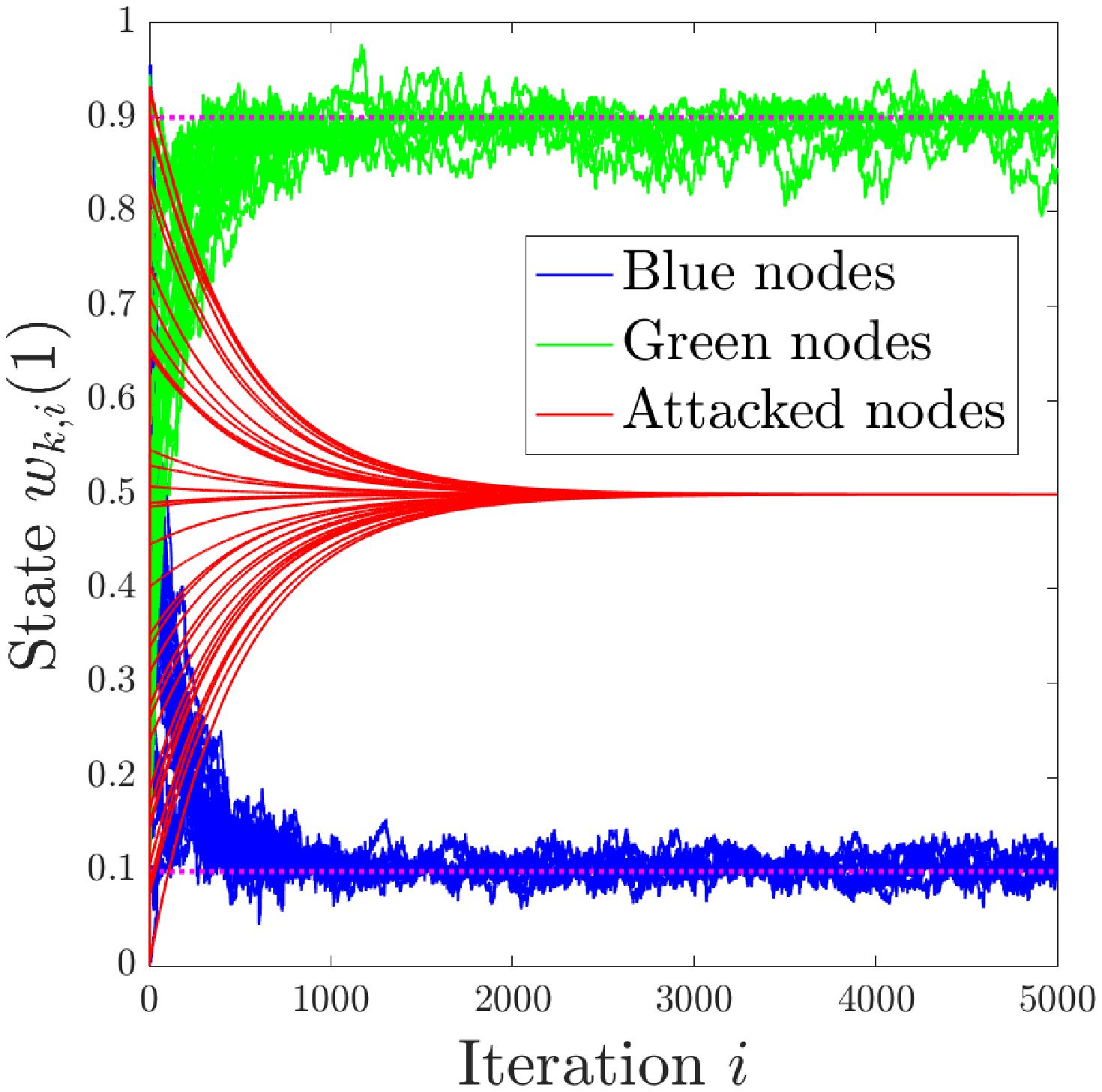}
         \vspace{0.3cm}
        \caption{$\bm{w}_{k,i}(1)$ (under strong attack)}
        \label{fig: state convergence 3}
    \end{subfigure}
    ~
        \begin{subfigure}[t]{0.235\textwidth}
    \centering
        \includegraphics[width=0.9\textwidth, trim=1.5cm 1.5cm 1.5cm 1.5cm]{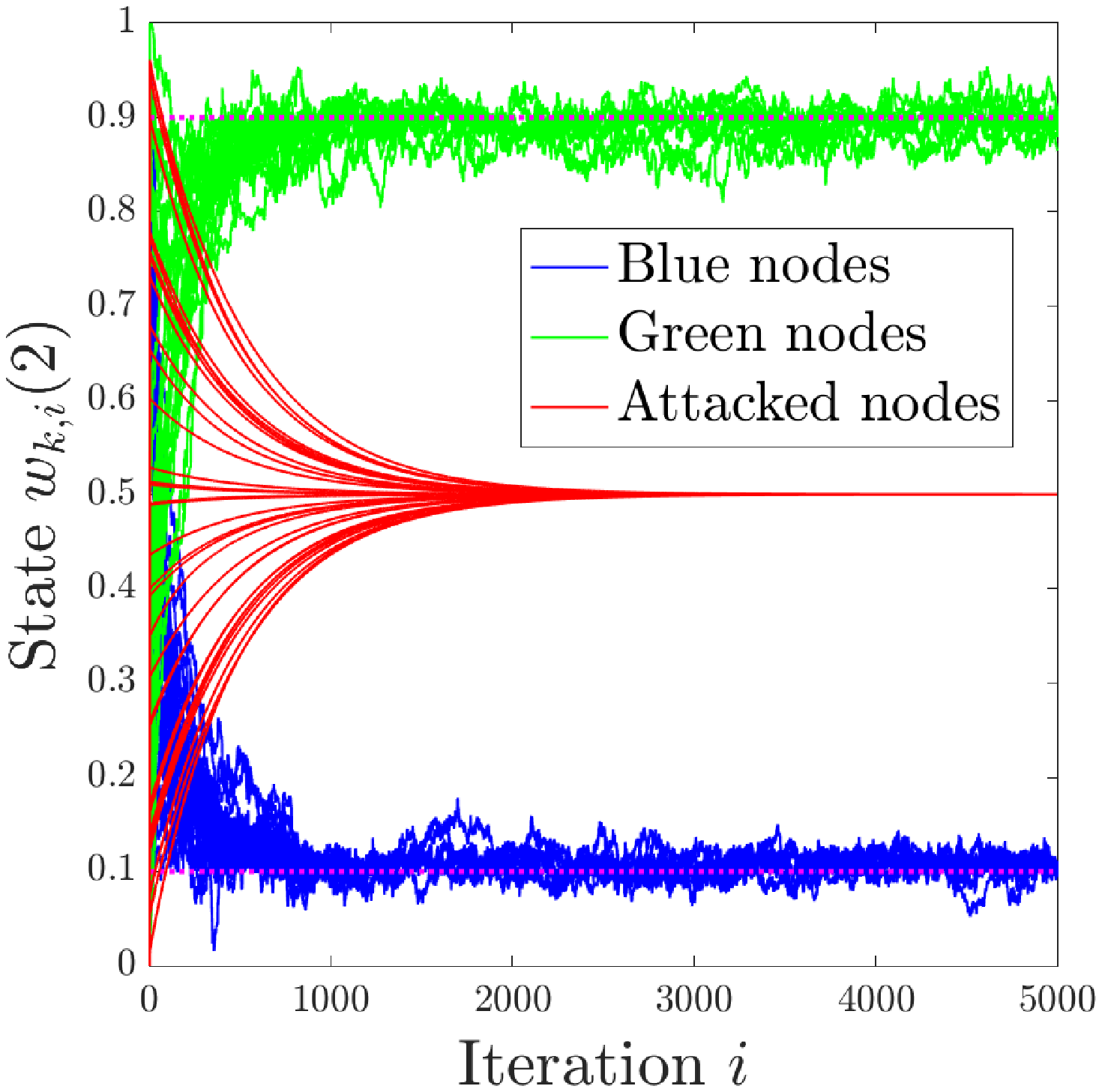}
        \vspace{0.3cm}
        \caption{$\bm{w}_{k,i}(2)$ (under  strong attack)}
        \label{fig: state convergence 4}
    \end{subfigure}
    \caption{\textcolor{green}{Estimation dynamics for stationary target localization by DLMSAW.}}\label{fig:state convergence}
\end{figure*}
%------------ End Figure -------------------------

Figure \ref{fig: stationary_convergence_trends} shows the convergence of nodes under attack (stationary targets). We note at around 3000 iterations, the difference between the average state of nodes under attack and the attacker's desired state $w_k^a$ becomes almost zero. This observation is also consistent with the result in \eqref{eq: attack convergence time}, as for $i=3000$ and $r_k^a = 0.002$, the value of $\epsilon$ turns out to be $0.0025$, which is indeed quite small and indicates the convergence of node's estimate to $w_k^a$. 

\textcolor{yellow}{Figure \ref{fig: non-stationary_convergence_trends} shows the average state dynamics of nodes under attack for non-stationary targets. \textcolor{yellow}{Since states are changing over time, we illustrate the dynamics of average states' changing with respect to the dynamics of attacker's selected state, instead of a convergence plot like \ref{fig: stationary_convergence_trends}.} Here, the $X$-coordinate denotes the first element of the estimation vector, i.e., $\bm{w}_{k,i}(1)$, and $Y$-coordinate denotes the second, i.e., $\bm{w}_{k,i}(2)$. At iteration $0$, the average state $\overline{w}_{k,i}$ of the nodes under attack is different than the attacker's desired state $w_{k,i}^a$. As the attack proceeds, $\overline{w}_{k,i}$ gradually converges towards $w_{k,i}^a$}, \textcolor{yellow}{which shows the effectiveness of attack for non-stationary state estimation.} %\sout{It seems in the figure that after around 500 iteration the convergence is reached, yet it is because the magnitude is pretty coarse. } 

\textcolor{green}{Figure \ref{fig:MSD} shows the \textcolor{revision2}{steady-state} MSD performance of DLMSAW and non-cooperative LMS. We observe that under no attack, cooperation indeed improves the \textcolor{revision2}{steady-state} MSD performance of DLMSAW. However, in the case of an attack, the \textcolor{revision2}{steady-state} MSD level of DLMSAW is quite high, whereas, the \textcolor{revision2}{steady-state} MSD level of non-cooperative LMS is barely affected by the attack.} 

%\todo[inline]{I really dont understand Figure 5b, and the above explanation in red. Could you please re-write it. By the way do we need this? Is it really necessary?
%\textcolor{green}{Please find in the above.}}

% ------------- Figure Begins ---------------
\begin{figure*}
 \centering
\begin{minipage}[c]{0.5\linewidth}
    \centering
    \vspace{0cm} 
\setlength{\abovecaptionskip}{0.1cm}  
    \begin{subfigure}[t]{0.45\textwidth}
       \centering
\includegraphics[width=0.9\textwidth, trim=1.5cm 1.5cm 1.5cm 1.5cm]{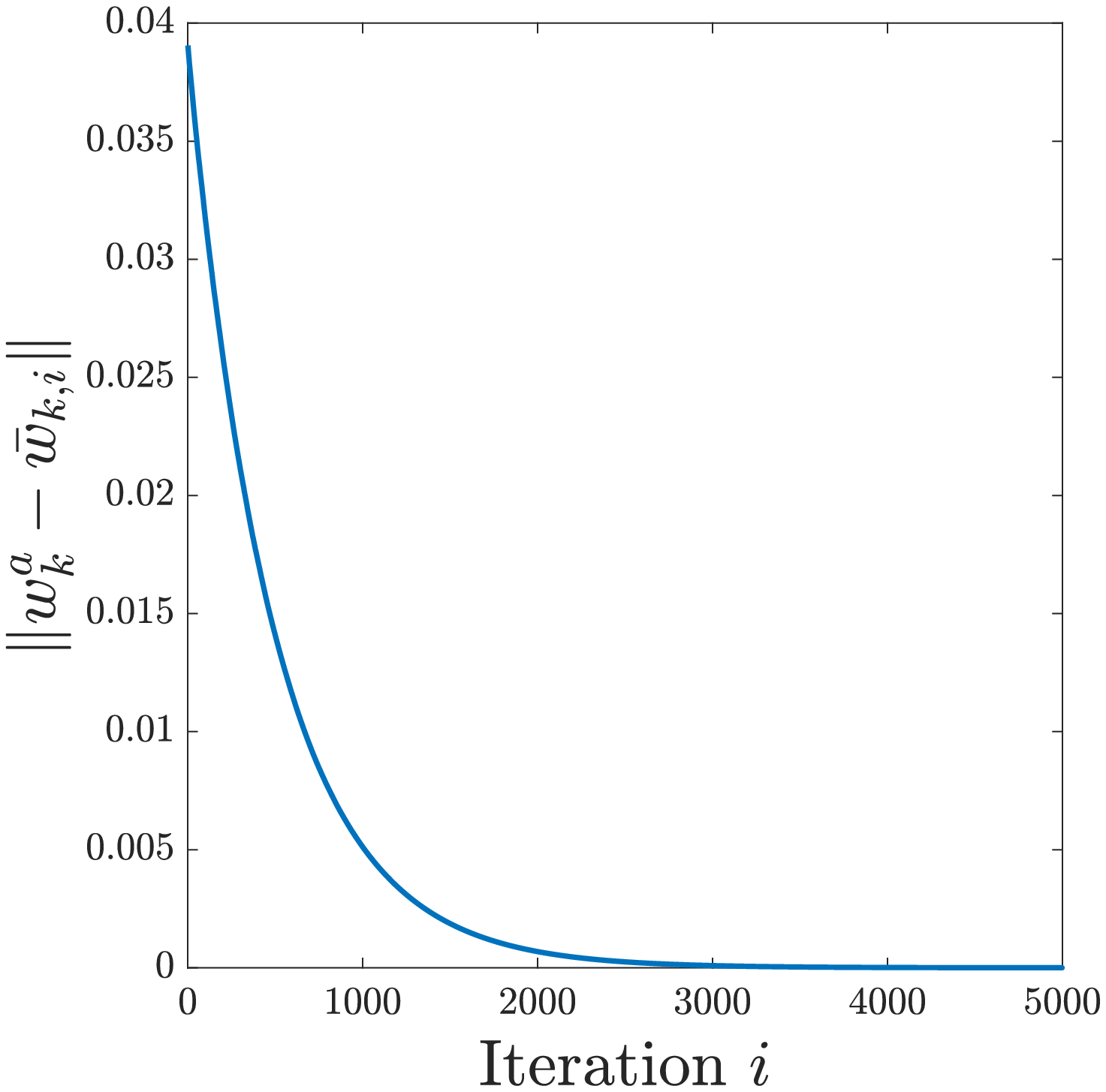}
\vspace{0.3cm}
\caption{Stationary targets}\label{fig: stationary_convergence_trends}
    \end{subfigure}
    ~ %add desired spacing between images, e. g. ~, \quad, \qquad, \hfill etc. 
    \begin{subfigure}[t]{0.45\textwidth}
            \centering
        \includegraphics[width=0.9\textwidth, trim=1.5cm 1.5cm 1.5cm 1.5cm]{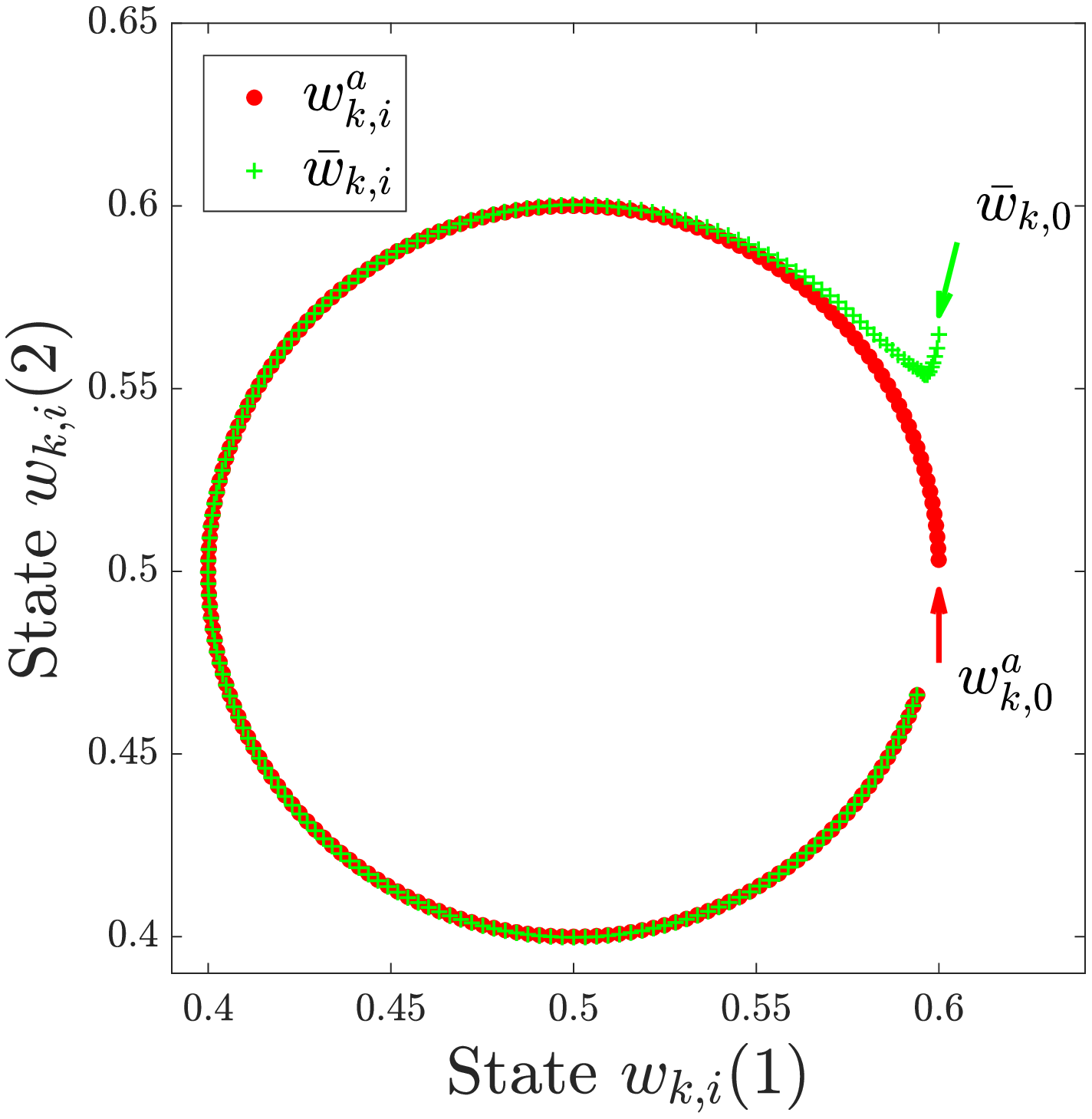}
        \vspace{0.3cm}
\caption{Non-stationary targets}\label{fig: non-stationary_convergence_trends}
    \end{subfigure}
    
    \caption{Average state dynamics of compromised nodes’ neighbors \\(under strong attack).}\label{fig:convergence dynamics}
     \end{minipage}%
\begin{minipage}[c]{0.5\linewidth}
    \centering
    \vspace{0cm} 
\setlength{\abovecaptionskip}{0.1cm}  
    \begin{subfigure}[t]{0.45\textwidth}
        \centering
        \includegraphics[width=0.9\textwidth, trim=1.5cm 1.5cm 1.5cm 1.5cm]{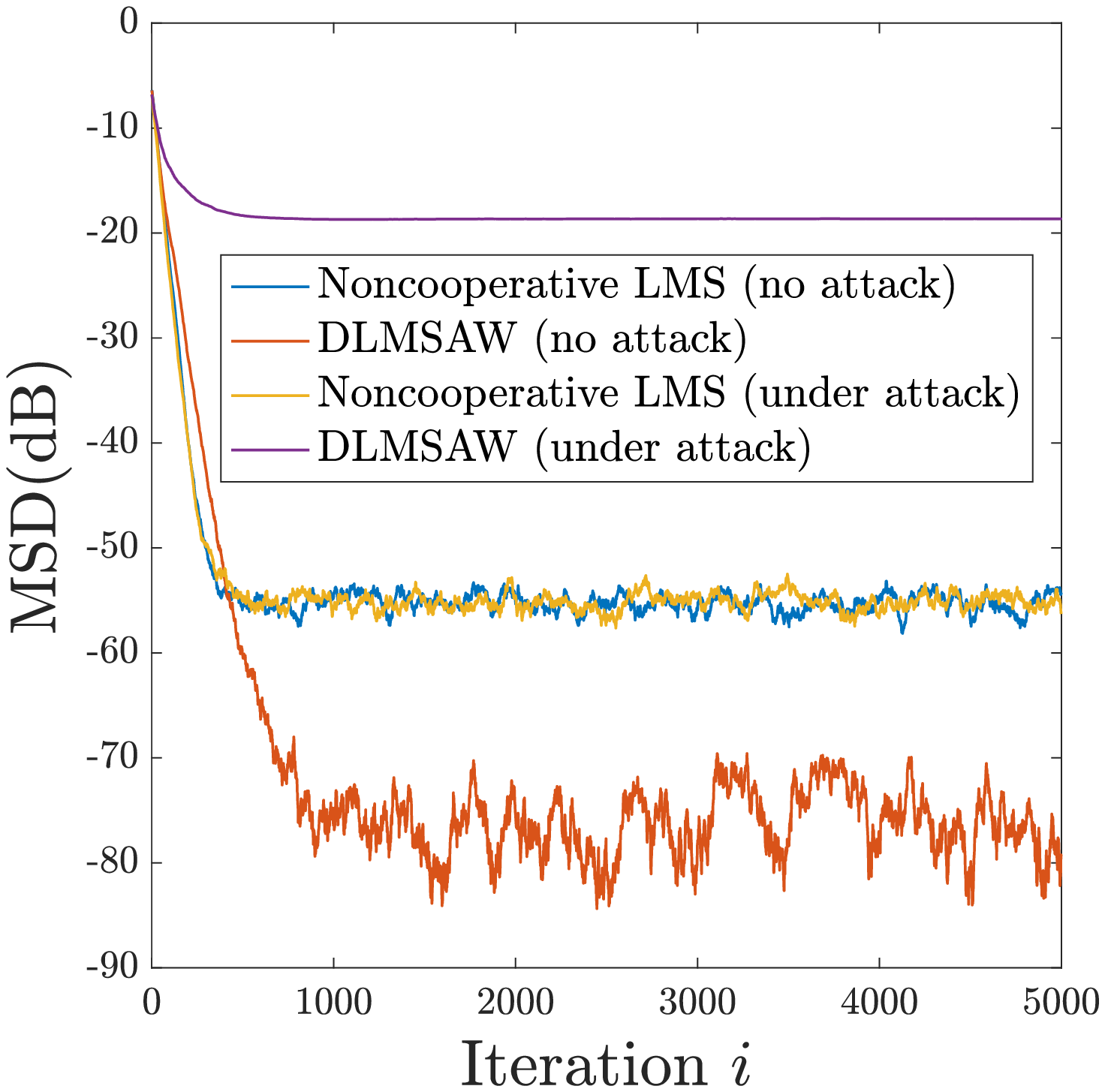}
        \vspace{0.3cm}
        \caption{Stationary targets}
        \label{fig: MSD level for stationary targets}
    \end{subfigure}
    ~ %add desired spacing between images, e. g. ~, \quad, \qquad, \hfill etc. 
    \begin{subfigure}[t]{0.45\textwidth}
           \centering
\includegraphics[width=0.9\textwidth, trim=1.5cm 1.5cm 1.5cm 1.5cm]{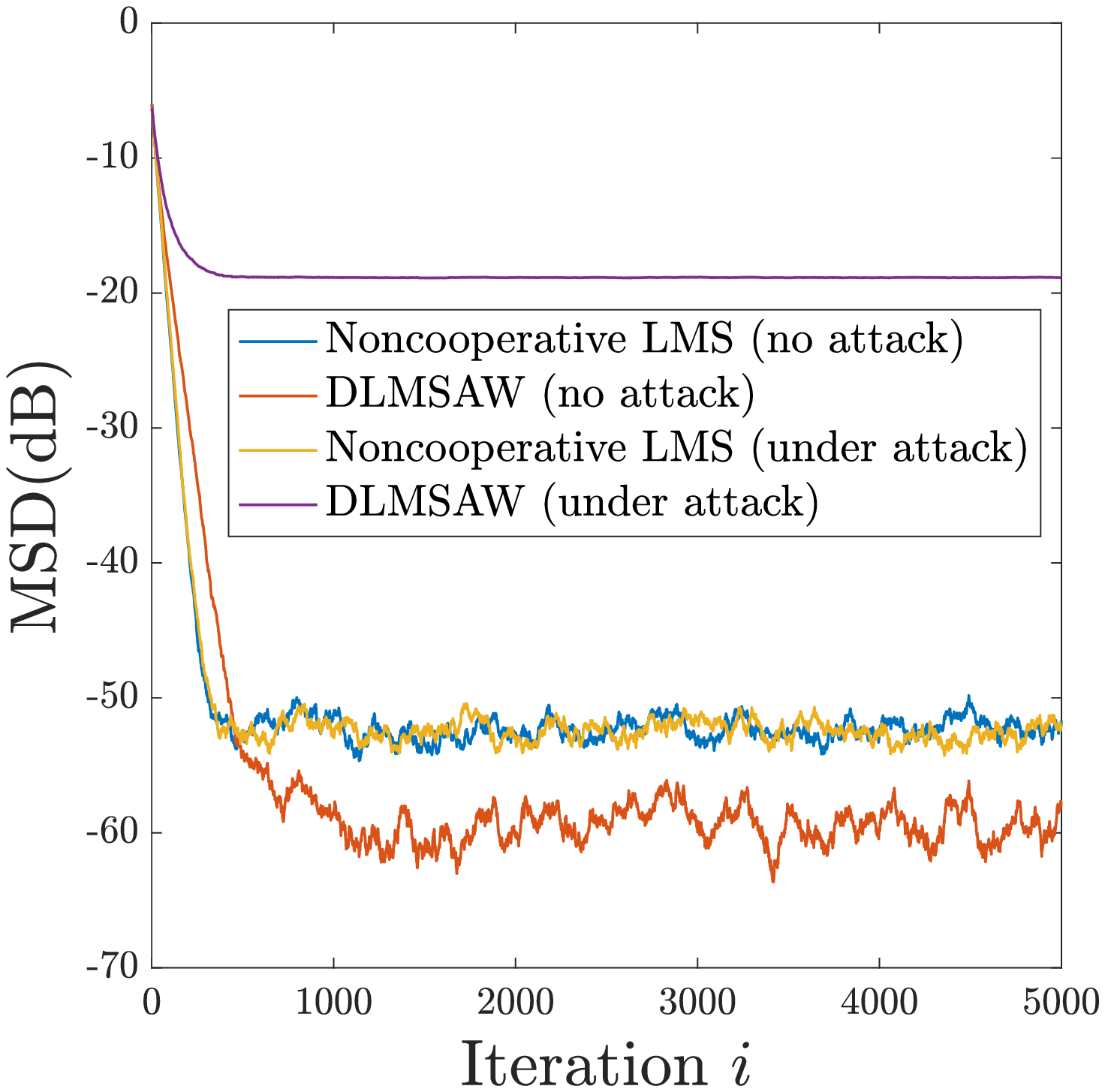}
\vspace{0.3cm}
\caption{Non-stationary targets}\label{fig:non-stationary MSD level}
    \end{subfigure}
    \caption{\textcolor{revision2}{Steady-state} MSD levels in non-cooperative LMS \\and DLMSAW (under strong attack).}\label{fig:MSD}
      \end{minipage}%
\end{figure*}
% ------------- Figure Ends ---------------

\textcolor{revision}{\subsection{Resilient Diffusion for Strong Attacks}}
To evaluate R-DLMSAW, we compute the cost $J_k(\bm{\psi}_{l,i})$ using the streaming data from the latest 100 iterations. 
%Since the algorithm is resilient by selecting $F$ greater than or equal to the quantity of compromised nodes in one normal agent's neighborhood, 
We adopt uniform $F$ for every normal agent but it can be distinct for each  agent. 
R-DLMSAW behaves identical\textcolor{revision2}{ly} to DLMSAW at one extreme, that is when $F = 0$, and on the other extreme it behaves like a non-cooperative LMS algorithm, that is for large $F$. \textcolor{green}{We consider the same initial network as in Figure \ref{fig: initial stationary network topology} and consider an attack consisting of four compromised nodes as previously. Note that there is at most one compromised node in the neighborhood of a normal agent. Figure \ref{fig:R-DLMSAW different F} shows network topologies after executing R-DLMSAW for various values of $F$. Since there is at most one compromised node in the neighborhood of a normal agent, the selection of $F=1$ should be sufficient to guarantee that none of the normal nodes converge to attacker's desired states, which is indeed the case as indicated by the removal of all links between normal and compromised nodes in Figure \ref{fig:network_Flocal}. As we increase $F$, resilience against attack is certainly achieved, but at the same time the network becomes sparser as illustrated in Figures \ref{fig:network_Flocal_F=3} and \ref{fig:network_Flocal_F=5}. In the case of non-stationary state estimation, the resulting network topologies are similar, and hence, are not presented.}

\textcolor{revision}{\figref{fig:state convergence R-DLMSAW} shows the estimation dynamics by R-DLMSAW for the targets' location $\bm{w}_{k,i}(1)$ and $\bm{w}_{k,i}(2)$ for every agent $k$ and iteration $i$ from 0 to $5000$ under attack.
%, where $\bm{w}_{k,i} = [ \bm{w}_{k,i}(1), \bm{w}_{k,i}(2)]^\top$. 
The attacked nodes in the figure refer to the immediate neighbors of the compromised nodes. Since there is at most one compromised node in a normal node's neighborhood, setting $F \geq 1$ will make R-DLMSAW algorithm resilient to attacks, which is demonstrated by the results from the figure. We also observe that by setting a smaller $F$ value, which is sufficient to to make the algorithm resilient, we achieve better estimation performance ($F=1$ has less noise than that of $F=5$).}

%Figure \ref{fig:R-DLMSAW different F} shows the final network topology for various values of $F$. Since each normal agent has at most one compromised  neighbor, by selecting $F \geq 1$, R-DLMSAW is resilient. But as we increase $F$, the network becomes more and more sparse. The resulted network topology for non-stationary state estimation is similar and omitted.

\figref{fig:MSD for three algorithms} shows the \textcolor{revision2}{steady-state} MSD level of the network for the three algorithms, that is, non-cooperative LMS, DLMSAW, and R-DLMSAW. \textcolor{green}{The simulation results validate claims in Section \ref{sec:resilient_diffusion}. We observe that in the presence of compromised nodes, DLMSAW performs the worst and has the highest \textcolor{revision2}{steady-state} MSD. Since there is at most one compromised node in the neighborhood of any normal node, the most appropriate value of $F$ for R-DLMSAW is 1. We note that \textcolor{revision2}{the steady-state} MSD is indeed minimum for $F=1$. As we increase $F$, \textcolor{revision2}{the steady-state} MSD also increases. In fact, for $F=5$, the performance of R-DLMSAW and non-cooperative LMS is almost the same as we expect.}

%------------------ Begin Figure ------------------
\begin{figure*}
    \centering
    \setlength{\abovecaptionskip}{0.1cm}  
    \begin{subfigure}[t]{0.235\textwidth}
\centering
\includegraphics[width=0.9\textwidth, trim=1.7cm 1.7cm 1.7cm 1.7cm]{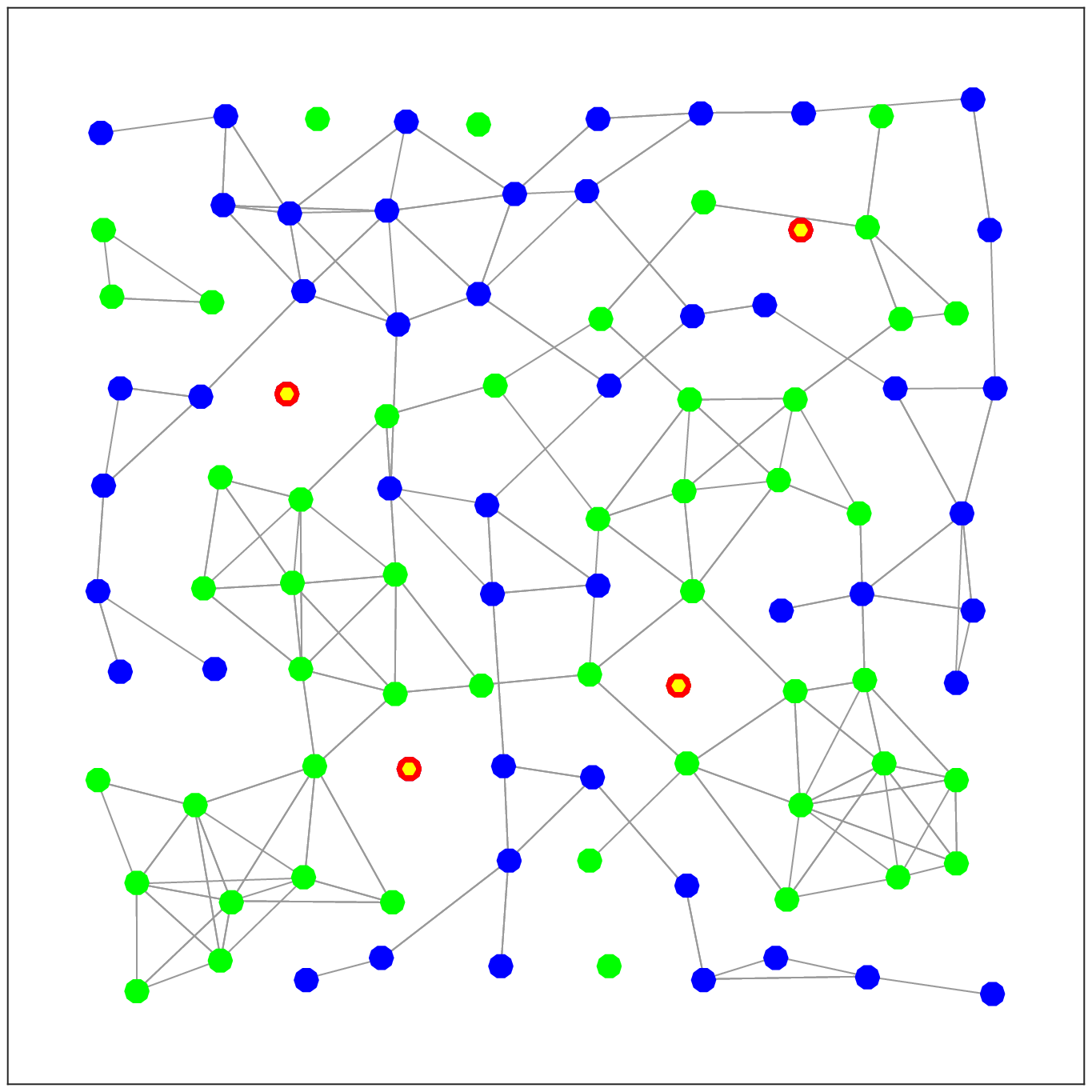}
\caption{$F = 1$}\label{fig:network_Flocal}
    \end{subfigure}
    ~ %add desired spacing between images, e. g. ~, \quad, \qquad, \hfill etc. 
    \begin{subfigure}[t]{0.235\textwidth}
    \centering
\includegraphics[width=0.9\textwidth, trim=1.7cm 1.7cm 1.7cm 1.7cm]{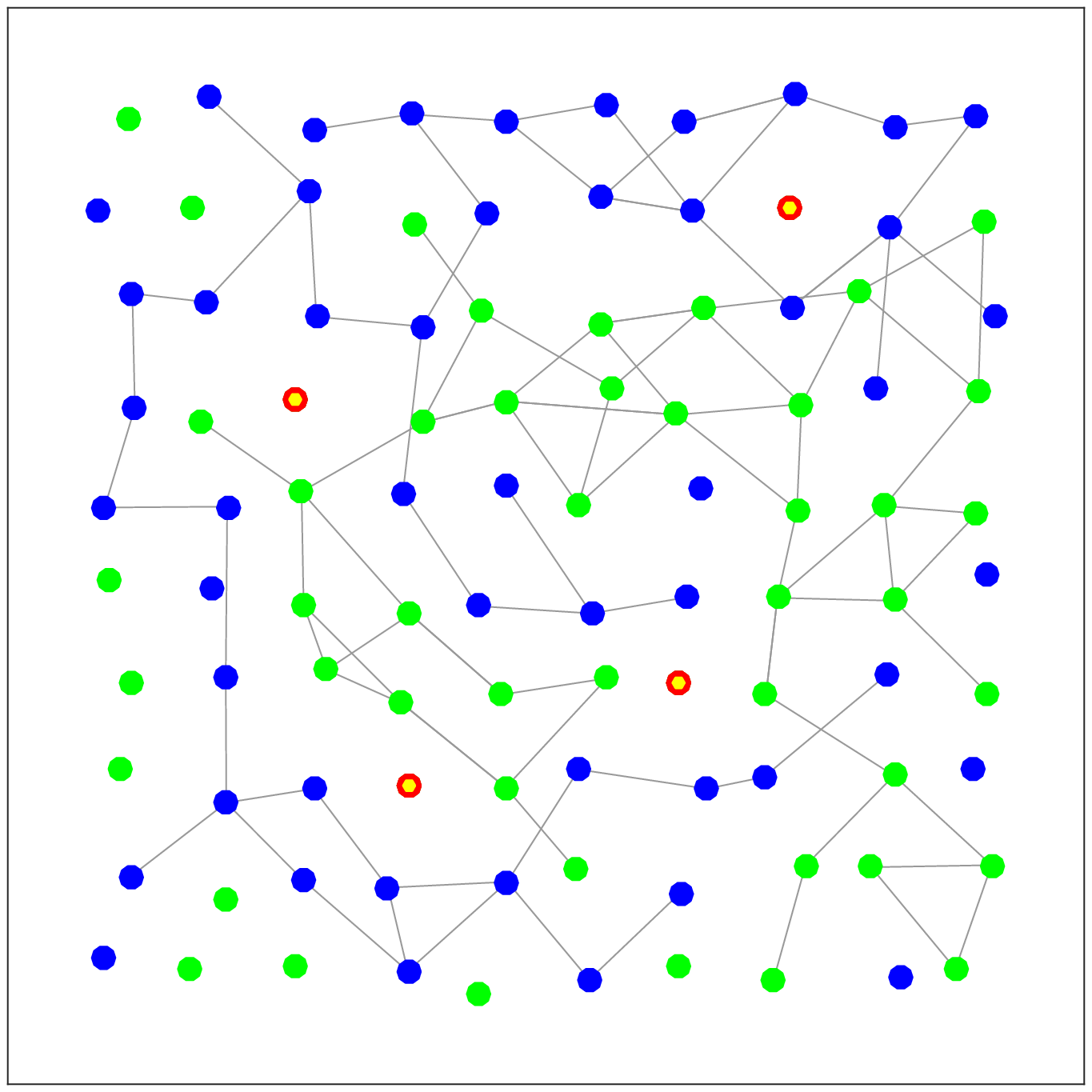}
\caption{$F = 3$}\label{fig:network_Flocal_F=3}
    \end{subfigure}
    ~ %add desired spacing between images, e. g. ~, \quad, \qquad, \hfill etc. 
    \begin{subfigure}[t]{0.235\textwidth}
\centering
\includegraphics[width=0.9\textwidth, trim=1.7cm 1.7cm 1.7cm 1.7cm]{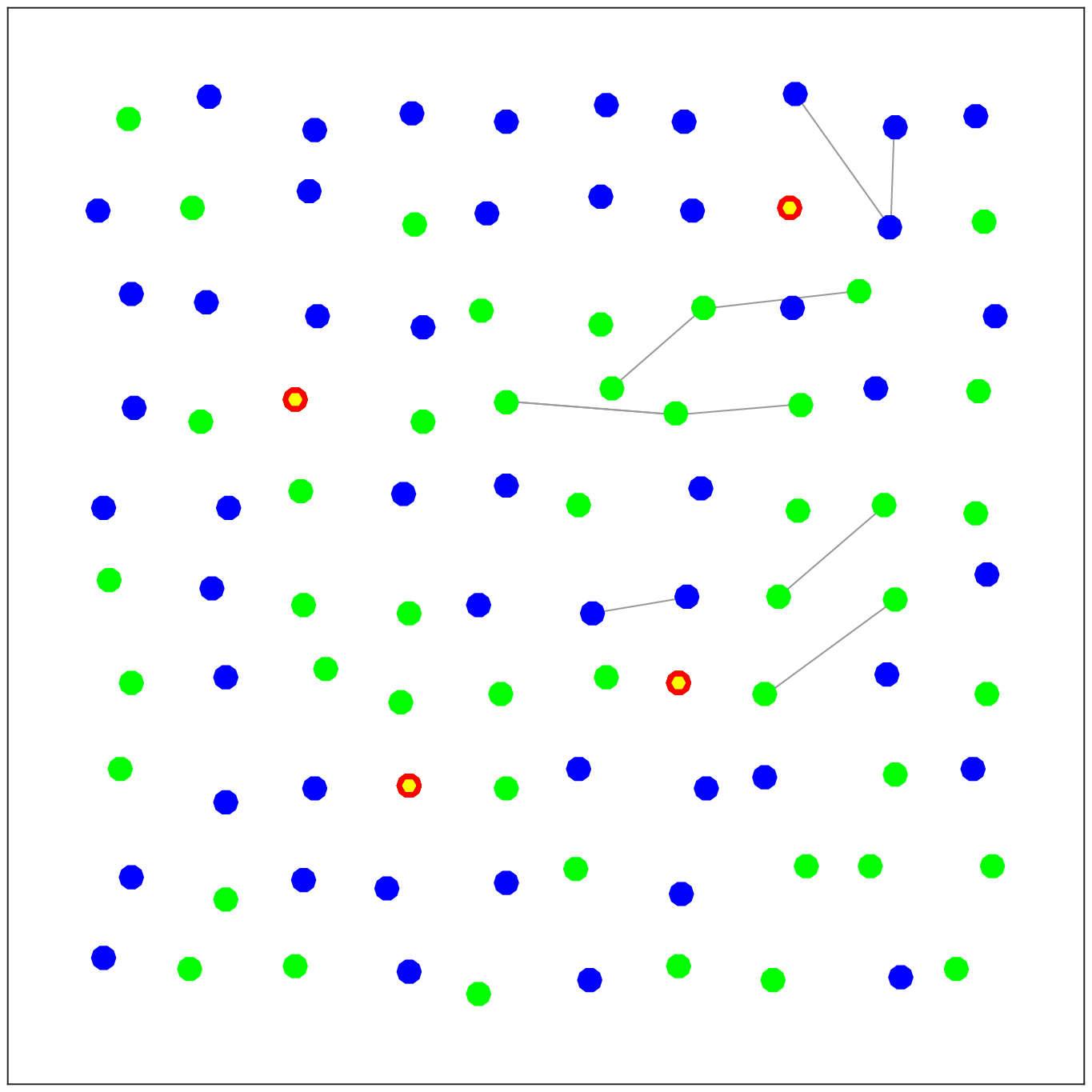}
\caption{$F = 5$}\label{fig:network_Flocal_F=5}
    \end{subfigure}
    \caption{Network topologies at the end of R-DLMSAW under strong attack (stationary targets) for various values of $F$.}\label{fig:R-DLMSAW different F}
\end{figure*}
%---------------- End Figure -----------------------------

%------------ Begin Figure -------------------------
\begin{figure*}
    \centering
    \setlength{\abovecaptionskip}{0.1cm}  
    \begin{subfigure}[t]{0.235\textwidth}
    \centering
        \includegraphics[width=0.9\textwidth, trim=1.5cm 1.5cm 1.5cm 1.5cm]{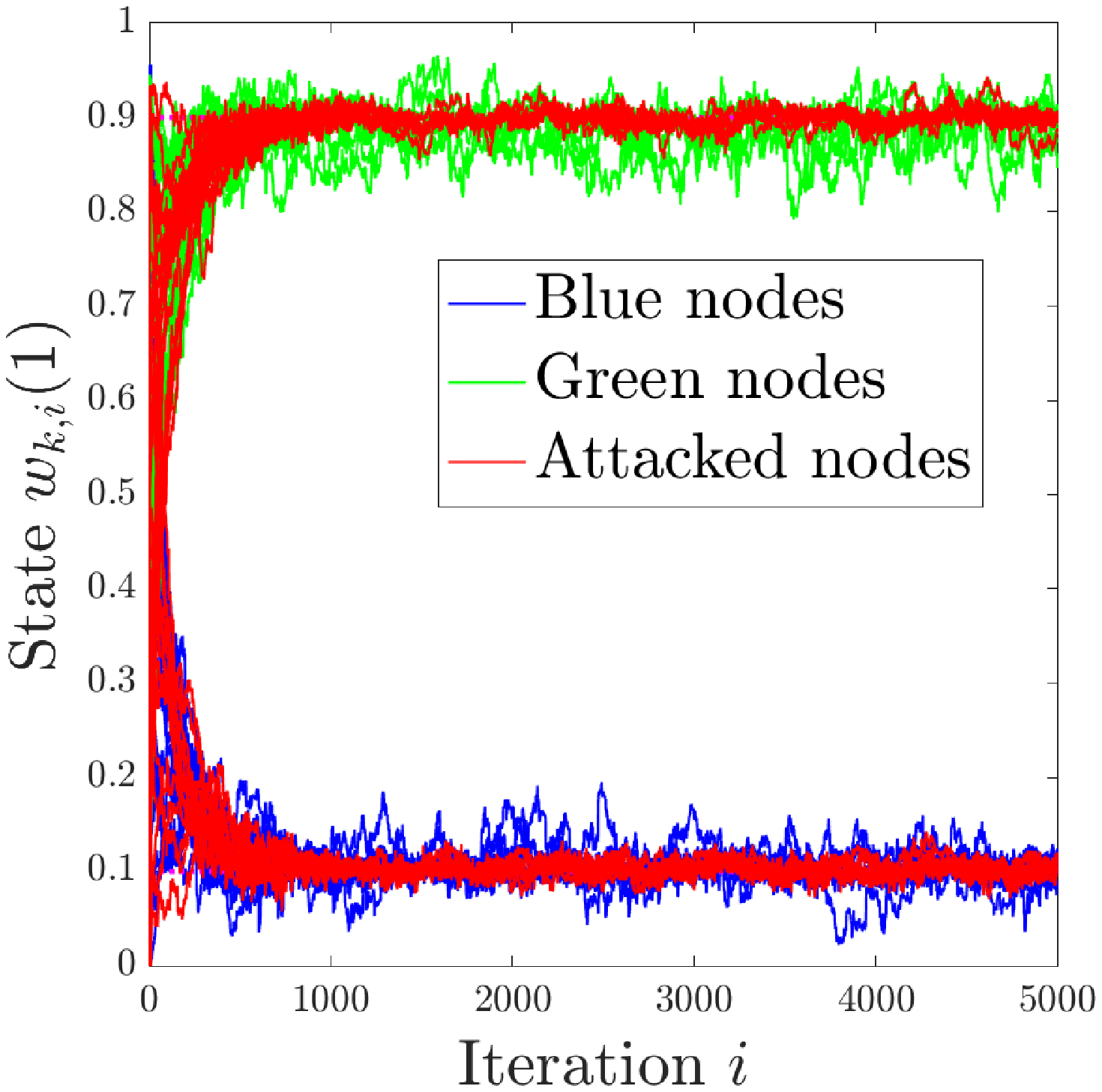}
         \vspace{0.3cm}
        \caption{$\bm{w}_{k,i}(1)$ ($F=1$)}
        \label{fig: state convergence R-DLMSAW 1}
    \end{subfigure}
    ~ 
    \begin{subfigure}[t]{0.235\textwidth}
    \centering
        \includegraphics[width=0.9\textwidth,  trim=1.5cm 1.5cm 1.5cm 1.5cm]{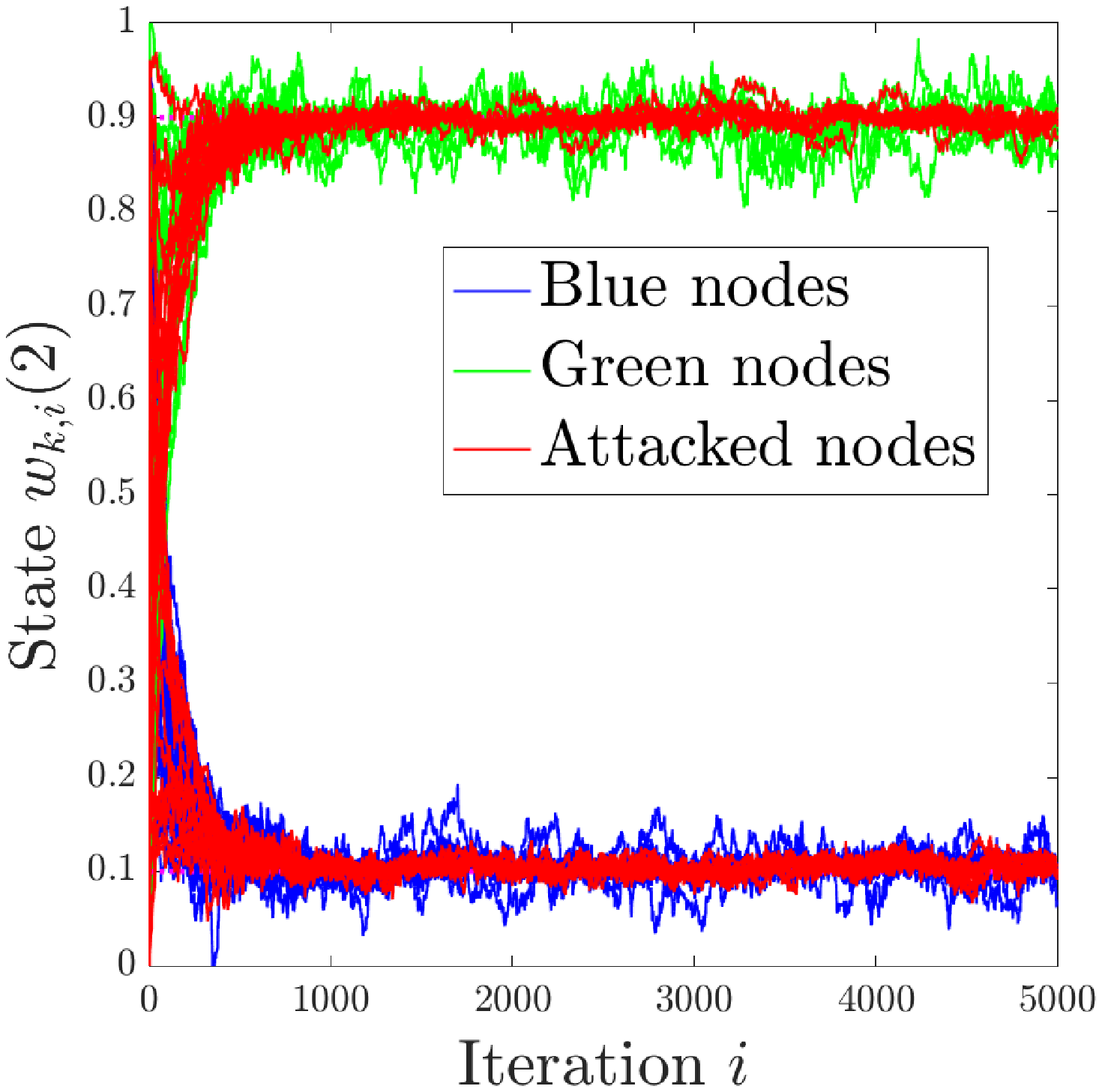}
         \vspace{0.3cm}
        \caption{$\bm{w}_{k,i}(2)$ ($F=1$)}
        \label{fig: state convergence R-DLMSAW 2}
    \end{subfigure}
~
    \begin{subfigure}[t]{0.235\textwidth}
    \centering
        \includegraphics[width=0.9\textwidth, trim=1.5cm 1.5cm 1.5cm 1.5cm]{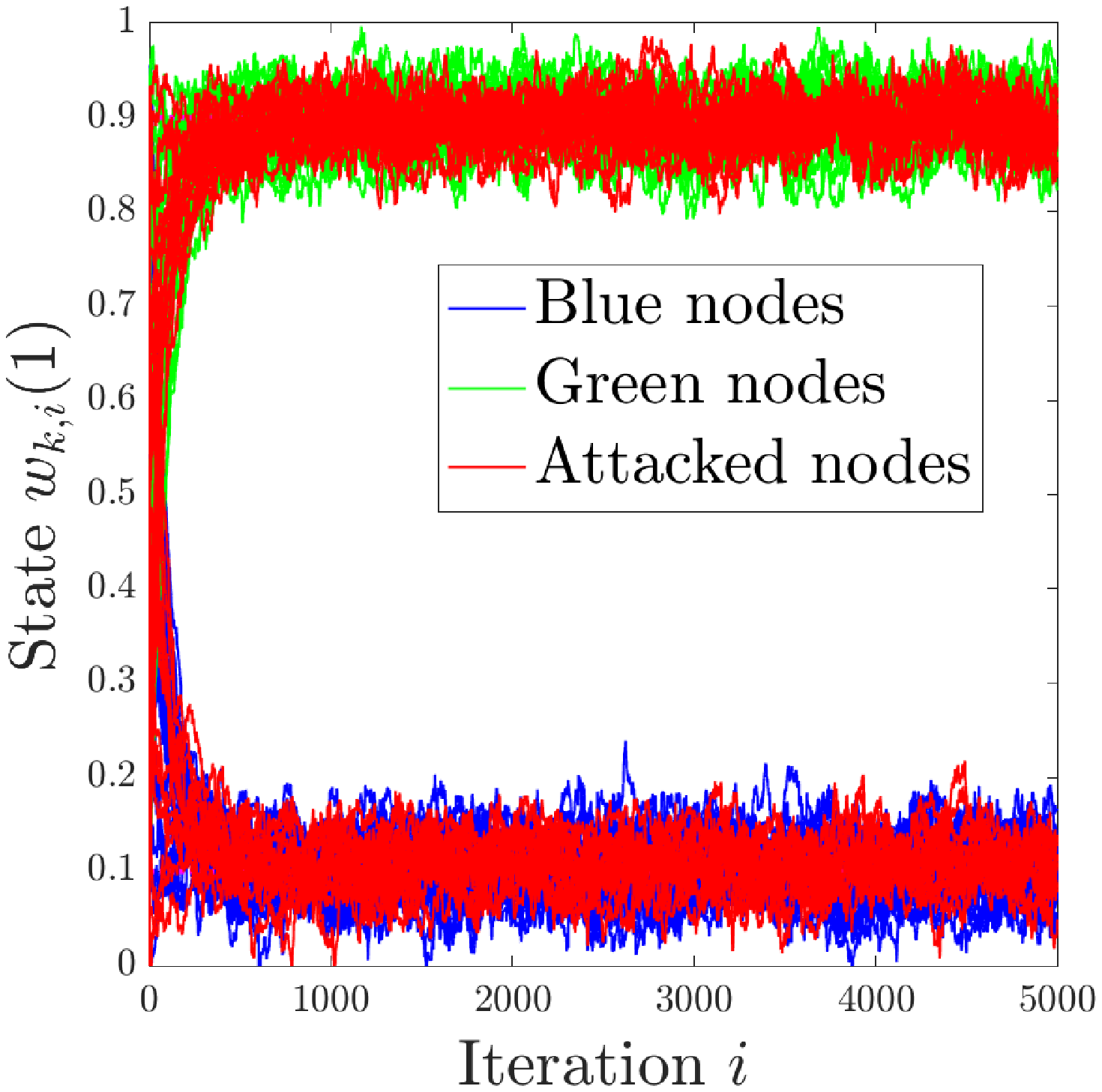}
         \vspace{0.3cm}
        \caption{$\bm{w}_{k,i}(1)$ ($F=5$)}
        \label{fig: state convergence R-DLMSAW 3}
    \end{subfigure}
    ~
        \begin{subfigure}[t]{0.235\textwidth}
    \centering
        \includegraphics[width=0.9\textwidth, trim=1.5cm 1.5cm 1.5cm 1.5cm]{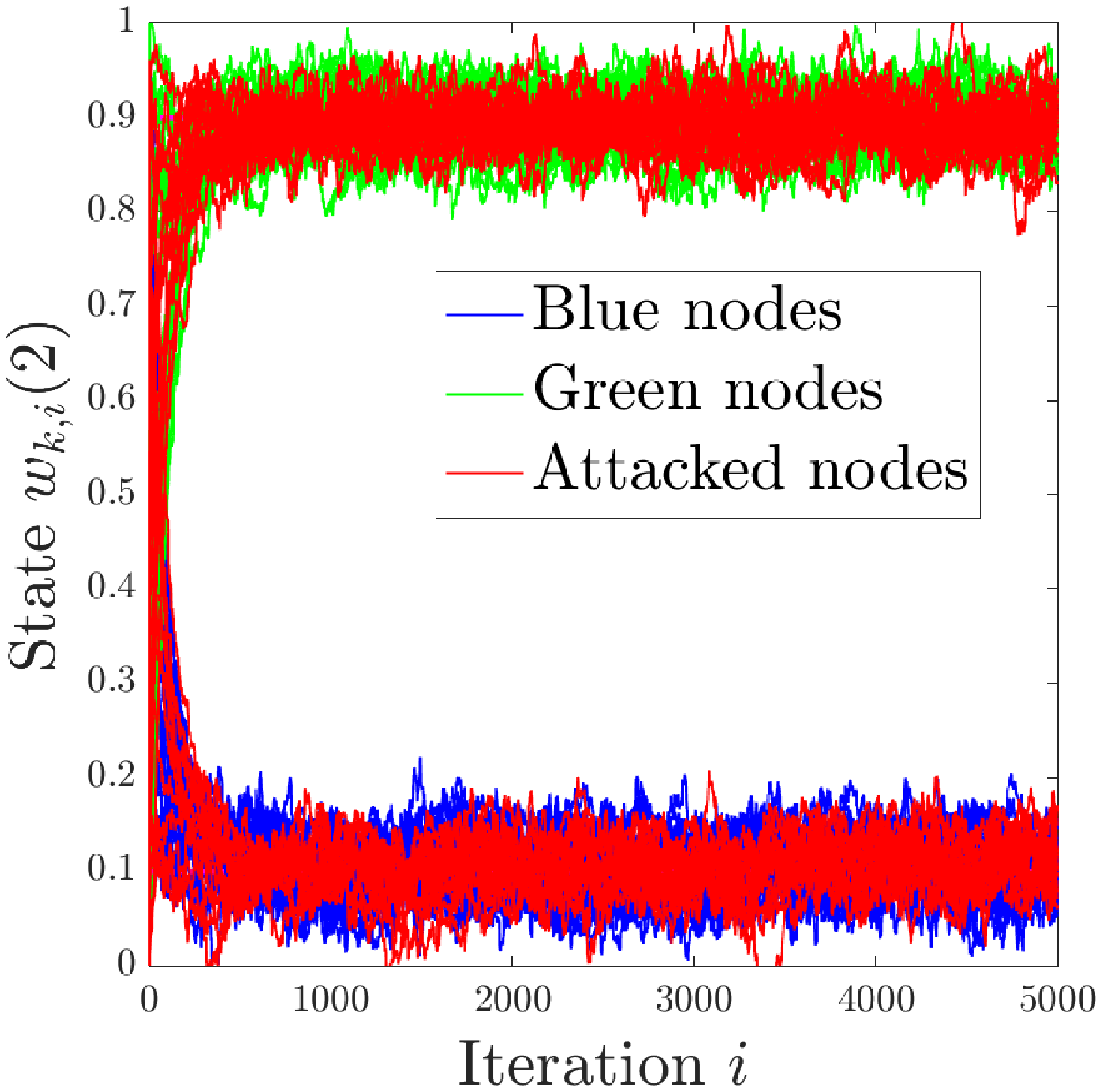}
        \vspace{0.3cm}
        \caption{$\bm{w}_{k,i}(2)$ ($F=5$)}
        \label{fig: state convergence R-DLMSAW 4}
    \end{subfigure}
    \caption{\textcolor{green}{Estimation dynamics for stationary target localization by R-DLMSAW under strong attack.}}\label{fig:state convergence R-DLMSAW}
\end{figure*}
%------------ End Figure -------------------------

%------------------ Begin Figure ------------------
\begin{figure*}
    \centering
    \vspace{0cm} 
\setlength{\abovecaptionskip}{0.1cm}  
    \begin{subfigure}[t]{0.25\textwidth}
        \centering
        \includegraphics[width=0.9\textwidth, trim=1.5cm 1.5cm 1.5cm 1.5cm]{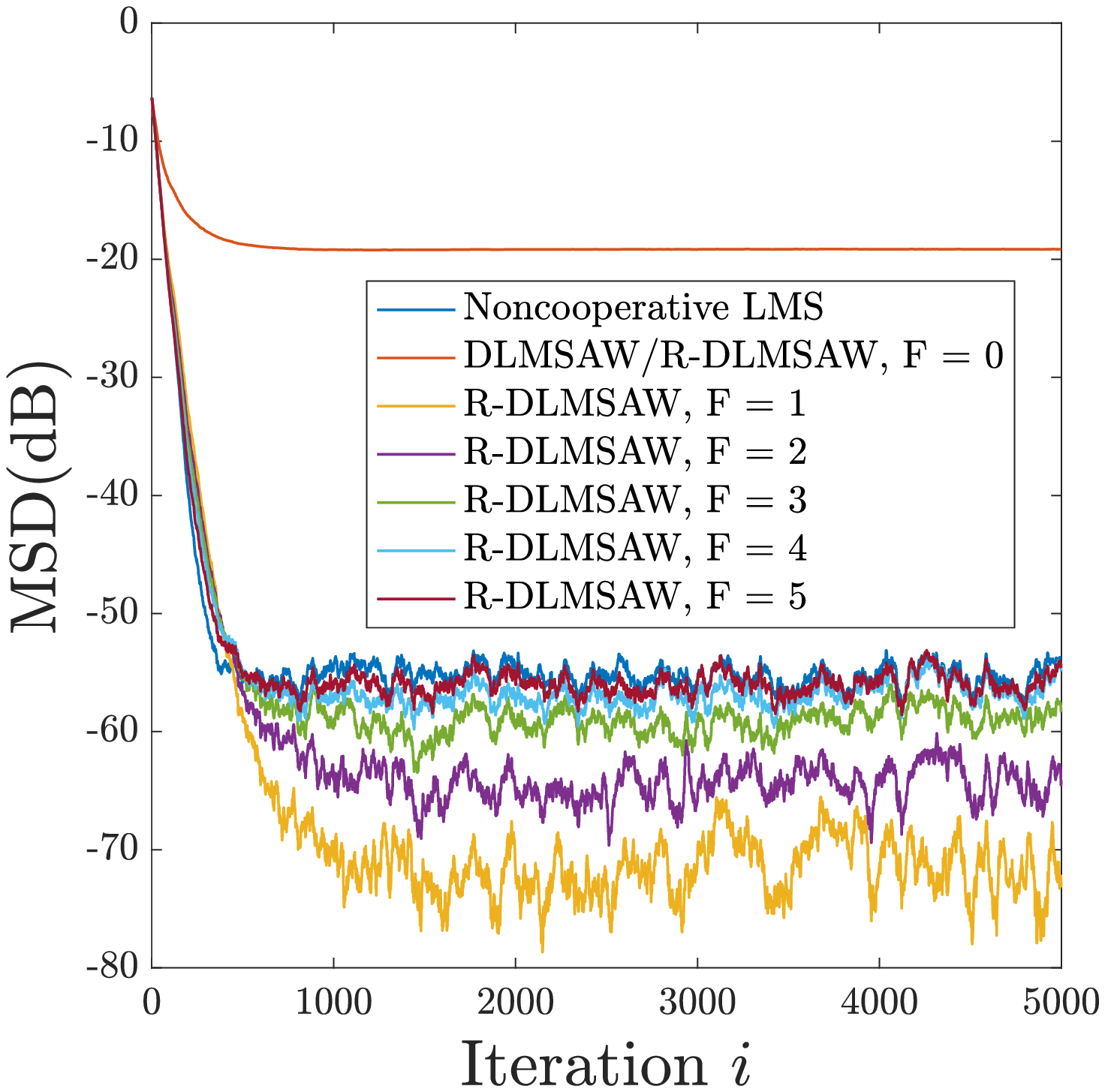}
        \vspace{0.3cm}
\caption{Stationary targets}\label{fig: MSD_Flocal}
    \end{subfigure}
    ~ %add desired spacing between images, e. g. ~, \quad, \qquad, \hfill etc. 
    \begin{subfigure}[t]{0.25\textwidth}
           \centering
\includegraphics[width=0.9\textwidth, trim=1.5cm 1.5cm 1.5cm 1.5cm]{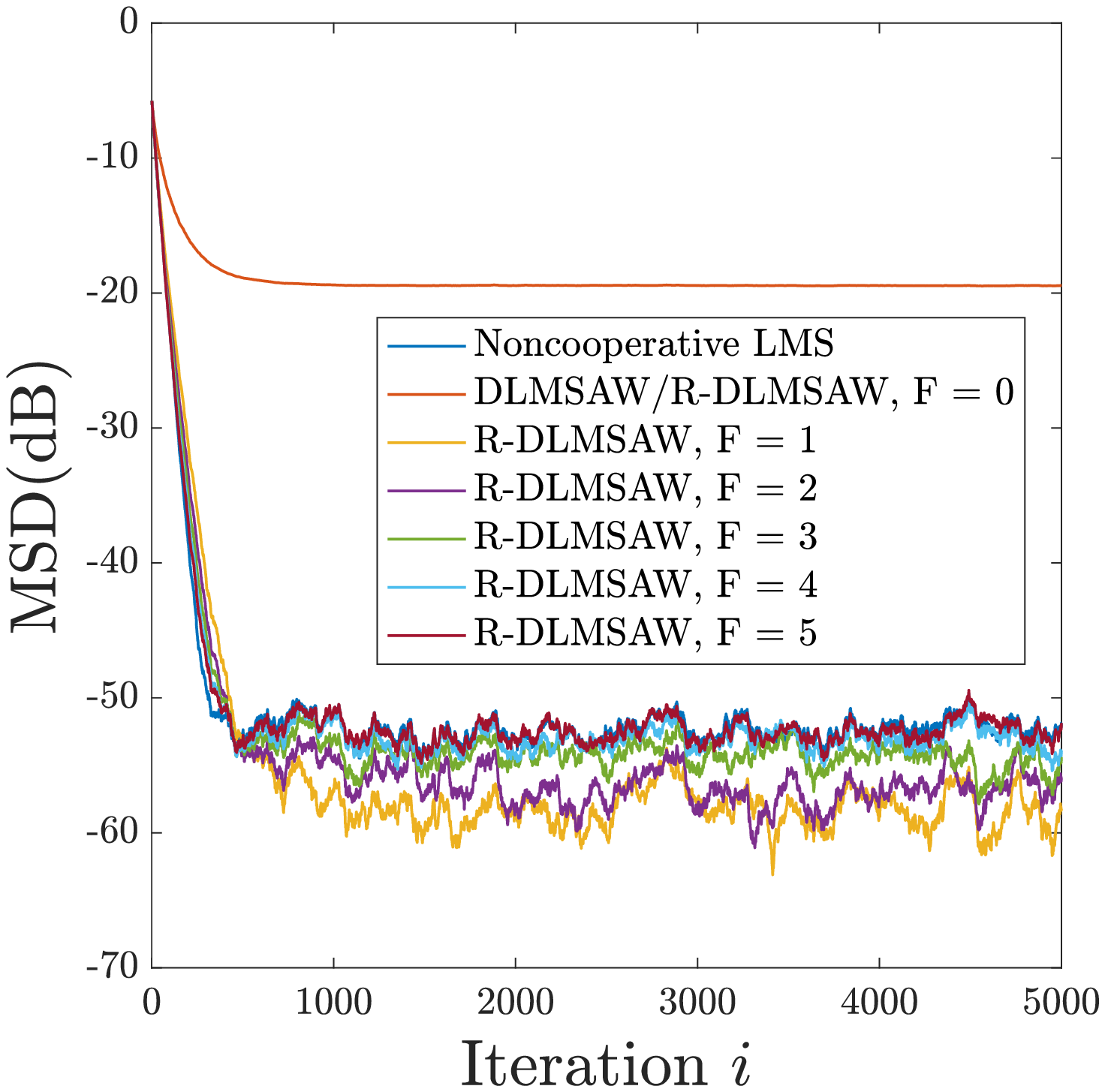}
\vspace{0.3cm}
\caption{Non-stationary targets}\label{fig:non-stationary F-local}
    \end{subfigure}
    \caption{\textcolor{green}{A comparison of MSD performance of non-cooperative LMS, DLMSAW, and R-DLMSAW under strong attack.}}\label{fig:MSD for three algorithms}
\end{figure*}
%---------------- End Figure -----------------------------

\textcolor{revision}{\section{Weak Attacks\label{sec:weak attack}}
Though it is common to assume a strong attacker with complete knowledge when examining the resilience of a distributed system, it is interesting to examine what an attacker can do in practise if all the information is not available. %it seems doubtful whether such a powerful attacker exists in practice which renders the whole system design somewhat trivial.
%To dispel such concerns, 
In this section, we analyze how the attack can still be deployed on a normal agent $k$ without the assumption of a strong knowledge by the attacker (streaming data and parameters used by $k$). We assume that an attacker has access only to the intermediate estimates shared by agents with others in their neighborhood. For instance, if $l\in\mathcal{N}_k$ then agent $k$ receives $\bm{\psi}_{l,i}$ from $l$ and attacker also has an access to it. We show that the other knowledge needed by the attacker can actually be approximated in an alternative way, and the success of the attack relies on how accurate this information can be approximated. We refer to such an attack in which attacker can only gather intermediate estimates and not the other data (including streaming data and agent parameters) as the \emph{weak attack}.
% , and the only knowledge the attacker knows is the intermediate state $\bm{\psi}_{l,i}$ that normal agent $l \in \mathcal{N}_k$ uses to exchange with $k$, which can be collected by establishing communication with $l \in \mathcal{N}_k$.
% We show that the knowledge needed by the attacker can actually be approximated in an alternative way, and the success of the attack relies on how accurate this information can be approximated.
% To differentiate, we refer to the attack with a strong knowledge as the strong attack and the attack without a strong knowledge and with only the communication message as the weak attack.
%To make the attack more practical, we then consider if it is possible to deploy the same attack without assuming such strong knowledge of the system (attacker knows only the intermediate state $\bm{\psi}_{l,i}$ that normal agent $l \in \mathcal{N}_k$ uses to exchange with $k$, which can be achieved by establishing communication with $l \in \mathcal{N}_k$). 
}

\textcolor{revision}{
The strong attack in \eqref{eq: attacker model} relies essentially on the knowledge of $\bm{w}_{k,i-1}$, that is the estimated state of agent $k$ in the last iteration.
If the attacker has complete knowledge, it can compute $\bm{w}_{k,i-1}$ exactly as \textit{Lemma 1} indicates.
However, without such knowledge, $\bm{w}_{k,i-1}$ can only be approximately computed. We note that approximating $\bm{w}_{k,i-1}$ is equivalent to approximating the weight matrix $A_{k}(i) = [a_{lk}(i)], \forall l \in \mathcal{N}_k$. This is true because $\bm{w}_{k,i} =  \sum_{l \in N_k} a_{lk}(i) \bm{\psi}_{l,i}$, and $\bm{\psi}_{l,i}$ is received by the attacker $a$ from $l$. 
% In order to approximate $\bm{w}_{k,i-1}$, it is equivalent to approximate the weight matrix $A_{k}(i) = [a_{lk}(i)], \forall l \in \mathcal{N}_k$ since the relationship holds that $\bm{w}_{k,i} =  \sum_{l \in N_k} a_{lk}(i) \bm{\psi}_{l,i}$ and $\bm{\psi}_{l,i}$ is received by $a$ from $l$ as a priori.
}

%\subsection{Retrieve weight matrix for black-box attack}
%Compromised nodes have to gain large enough weights from normal agents such that they can manipulate the state of normal agents. 
%According to the rule of how normal agents assign weights to their neighbors, compromised nodes need to minimize the difference between its current exchanging state and the normal neighbors' last estimated state so as to maximize the weight being assigned. To achieve this, compromised nodes need the knowledge of normal neighbors' last estimated state, which is assumed to be known in the white-box attack model. 
\textcolor{revision}{
Next, we discuss how to compute the approximated weight matrix $\hat{A}_k(i-1)$ using only the information $\bm{\psi}_{l,i}, \forall l \in \mathcal{N}_k$.
Note that the adaptation step \eqref{eq:adapt} of diffusion can be written as,
\begin{equation*}
    \bm{\psi}_{k,i} = \bm{w}_{k,i-1} + \nabla_{k,i}
    = A_k(i-1) \Psi_{k,i-1} + \nabla_{k,i}.
\end{equation*}
where $\nabla_{k,i} = \mu_k \bm{u}_{k,i}^*(\bm{d}_{k}(i) -\bm{u}_{k,i}\bm{w}_{k,i-1})$, $\Psi_{k,i-1}$ is an $|\mathcal{N}_k| \times M$ matrix $\Psi_{k,i-1} = [\bm{\psi}_{l,i-1}], \forall l \in \mathcal{N}_k$.
Thus,
\begin{equation*}
     \nabla_{k,i} = \bm{\psi}_{k,i}  - A_k(i-1) \Psi_{k,i-1},
\end{equation*}
and therefore, 
\begin{equation*}
    \lim_{i \rightarrow \infty}\mathbb{E}\{\|\nabla_{k,i}\|^2\} = \lim_{i \rightarrow \infty} \mathbb{E}\{\|\bm{\psi}_{k,i} - A_k(i-1) \Psi_{k,i-1}\|^2\}.
\end{equation*}
Since $\lim_{i \rightarrow \infty}\mathbb{E}\{\|\nabla_{k,i}\|^2\} = 0$,
the value of $A_k(i)$ can be approximated by assigning a cost function
%$\lim_{i \rightarrow \infty} A_k(i-1)$ can be obtained by solving the following optimization problem:
%\begin{equation*}
%    \min_{A_k(i-1)} \lim_{i \rightarrow \infty} \mathbb{E} \{\|  \bm{\psi}_{k,i} - A_k(i-1) \Psi_{k,i-1}\|^2\}
%\end{equation*}
%And it can be computed at every iteration by:
%\begin{equation}\label{eq:weight optimization}
%    \min_{A_k(i-1)} \mathbb{E} \{ \|  \bm{\psi}_{k,i} - A_k(i-1) \Psi_{k,i-1}\|^2\}
%\end{equation}
%The value of $A_k(i)$ can be approximated in an adaptive way.
%We denote the cost function of the estimated combination weight $A_k(i)$ of a normal agent $k$ as $\ell(A_k(i))$, such that
\begin{equation*}
    \ell(A_k(i)) \triangleq \mathbb{E} \{ \|  \bm{\psi}_{k,i+1} - A_k(i) \Psi_{k,i}\|^2\},
\end{equation*}
where $A_k(i)$ is the global minimizer of $\ell(A_k(i))$ as $i \rightarrow \infty$.
%According to the character of the second-order filter \eqref{eq: adaptive relative-variance combination rule} on weight matrix, the weight matrix changes over time but with a small changing rate.
Next, we compute the successive estimators of the weight matrix based on stochastic gradient descent method as follows:
\begin{equation}\label{eq: weight gradient descent}
    \begin{aligned}
        \hat{A}_{k}(i) &= \hat{A}_{k}(i-1) - \mu'_A \nabla_A \ell(\hat{A}_k(i-1)) \\
            &= \hat{A}_{k}(i-1) + \mu_A \Psi_{k,i-1} (\bm{\psi}_{k,i} - \hat{A}_k(i-1) \Psi_{k,i-1}),
    \end{aligned}
\end{equation}
where $\mu_A = \frac{1}{2} \mu'_A$.
}

\textcolor{revision}{
Also recall weight matrix $A_{k}(i)$ has to satisfy the condition \eqref{eq: weight constraints}. Thus, to make the adaptive approximation of weight matrix hold condition \eqref{eq: weight constraints}, we introduce two more steps following \eqref{eq: weight gradient descent}, that is the clip step and the normalization step. In the clip step, the negative weights are clipped and are set to zero; and the in the normalization step, weights are divided by their sum. %(so that these normalized weights sum up to one).
%Obviously, the additional operation can guarantee the condition in  \eqref{eq: weight constraints} is hold.
The operation for approximating weight matrix of a normal agent $k$ is summarized in \textit{Algorithm 2}.}

\begin{algorithm}
%\SetAlgoLined
\small
\DontPrintSemicolon
\textcolor{revision}{
\KwIn{ $l \in N_k$, randomized $a_{lk}(0)$ satisfying \eqref{eq: weight constraints}, $\mu_A$, $\bm{\psi}_{k,i}$, $\Psi_{k,i-1}$ }
             \For{$i > 0$}
             {
             $A_{k}(i) = A_{k}(i-1) + \mu_A \Psi_{k,i-1} (\bm{\psi}_{k,i} - A_k(i-1) \Psi_{k,i-1})$\;
             \For{$l \in \mathcal{N}_k$} {
               $a_{lk}(i) = \max(a_{lk}(i), 0)$
             }
            $A_{k}(i) = \frac{A_{k}(i)}{\sum a_{lk}(i)}$ \;
            \textbf{return} $ A_k(i) $ 
}
\caption{Approximate weight matrix for agent $k$}
}
\label{Algorithm 1}
\end{algorithm}

\textcolor{revision}{
We then approximate normal agent $k$'s estimated state by
\begin{equation*}
    \hat{\bm{w}}_{k,i} = \hat{A}_k(i) \Psi_{k,i},
\end{equation*}
and use $\hat{\bm{w}}_{k,i}$ instead of $\bm{w}_{k,i}$. The attack model in \eqref{eq: attacker model} then becomes
\begin{equation}\label{eq: attack using approximated weight}
\bm{\psi}_{a,i} = \hat{\bm{w}}_{k,i-1} + r_{k}^a (x_i - \hat{\bm{w}}_{k,i-1}).
\end{equation}
Note that the sufficient condition listed in \textit{Proposition 1} guarantees the convergence of the attack objective. However, without an exact knowledge of $\bm{w}_{k,i-1}$ it is not guaranteed the sufficient condition can be satisfied.
In other words, the success of the attack relies highly on how accurate the state $\hat{\bm{w}}_{k,i}$ can be approximated.
In the following, we provide evaluation results for such an attack. 
}

\textcolor{revision}{
\subsection{Evaluation}
%As we analyze above, the success of weak attack relies highly on how accurate the approximated state can be.
%Therefore, the attack may not necessary succeed in any case.
We adopt the same evaluation set-up as we used in section \Rmnum{7}. Initial network topology is the same as in \ref{fig: initial stationary network topology}.
Parameters we select are: $\sigma_{u,k}^2 \in [0.75, 0.85]$, $\sigma_{k}^2 \in [0.75, 0.85]$ for each agent $k$ and $\mu_A = 0.002$, while all the other settings are the same as in section \Rmnum{7}. 
}

\textcolor{revision}{
At the end of DLMSAW under weak attack, we reach the network topology as shown in \figref{fig: black box network topology DLMSAW}.
From the plots, we find some of the agents maintain connection with the compromised nodes, while others do not,  which is not the case with a strong attack, where all the neighboring agents of a compromised node end up cooperating only with the compromised node.
The main reason for this is that the weak attack may not have an accurate approximation of normal agents' state. Without an accurate approximation, compromised nodes may not be able to collect large weights from their neighbors and may not keep influencing the state\textcolor{revision2}{s} of their neighbors.}

\textcolor{revision}{
\figref{fig: estimation precision} illustrates the estimation precision ($\|\hat{\bm{w}}_{k,i} - \bm{w}_{k,i}\|$) by the attacker.
It shows that the attacker has different levels of accuracy to estimate the state\textcolor{revision2}{s} of its neighboring agents.
For some agents, the attacker has accurate approximation along the simulation iterations. As a result, the attacker is more likely to make its attack successful on those agents. However, for other agents, the attacker does not have very good approximation accuracy and therefore, it is hard for the attacker to successfully attack such agents.
\figref{fig: estimation dynamics black box attack DLMSAW} shows the state estimation dynamics of normal agents (wherein attacked nodes refer to the neighboring nodes of the compromised nodes).
We find the attacker can only drive a few of its neighbors to its desired state, whereas most of the normal neighbors converge to their true goal state, which is consistent with the results of \figref{fig: estimation precision}.
The \textcolor{revision2}{steady-state} MSD performance for the weak attack is shown in the yellow line in \figref{fig: MSD comprison black box attack}. We find that such an attack still worsens the network \textcolor{revision2}{steady-state} MSD as compared to the non-cooperative LMS (the blue line) and DLMSAW without attack (the red line).
}

\textcolor{revision}{
Next, we evaluate the proposed resilient diffusion algorithm R-DLMSAW against the weak attack.
The network topology at the end of simulation is shown in \figref{fig: black box network topology R-DLMSAW}.
Most normal agents have cut the link with the compromised nodes. Yet some links are maintained because these compromised nodes behave in a benign way as to send message with a smaller cost than a normal neighbor of the targeted node. In other words, these compromised nodes exchange a state message similar to normal nodes in order to maintain communication with them. Therefore, such links need not to be cut down to achieve the network resilience.
\figref{fig: estimation dynamics black box attack R-DLMSAW} shows the estimation dynamics of normal nodes by R-DLMSAW. We find none of the attacked nodes are driven to the attacker's selected state. All the nodes successfully converge to their true goal states.
The purple line in \figref{fig: MSD comprison black box attack} shows the \textcolor{revision2}{steady-state} MSD performance of R-DLMSAW with $F=1$. We observe that this line lies between the noncooperative LMS and DLMSAW (without attack), and has a much smaller \textcolor{revision2}{steady-state} MSD than DLMSAW under such attack.
This illustrates the effectiveness of the proposed resilient diffusion algorithm by showing that the algorithm is resilient to not only strong but also to weak attacks, as well as other data falsification attacks.
}

% ------------- Figure Begins ---------------
\begin{figure*}
 \centering
\begin{minipage}[c]{0.49\linewidth}
    \centering
    \vspace{0cm} 
\setlength{\abovecaptionskip}{0.1cm}  
    \begin{subfigure}[t]{0.45\textwidth}
       \centering
\includegraphics[width=0.9\textwidth, trim=2cm 2cm 2cm 2cm]{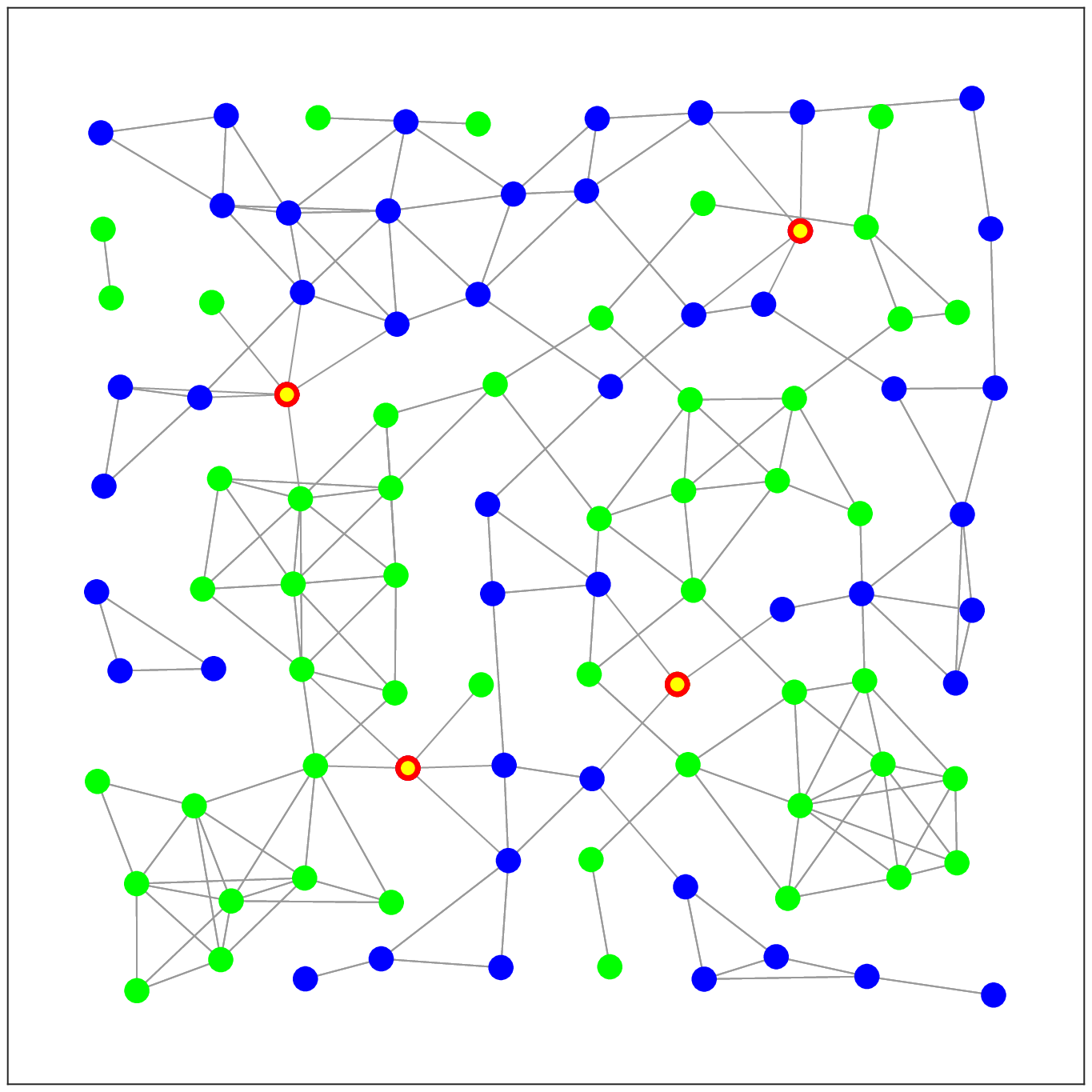}
\vspace{0.1cm}
\caption{DLMSAW}\label{fig: black box network topology DLMSAW}
    \end{subfigure}
    ~  
    \begin{subfigure}[t]{0.45\textwidth}
            \centering
        \includegraphics[width=0.9\textwidth, trim=2cm 2cm 2cm 2cm]{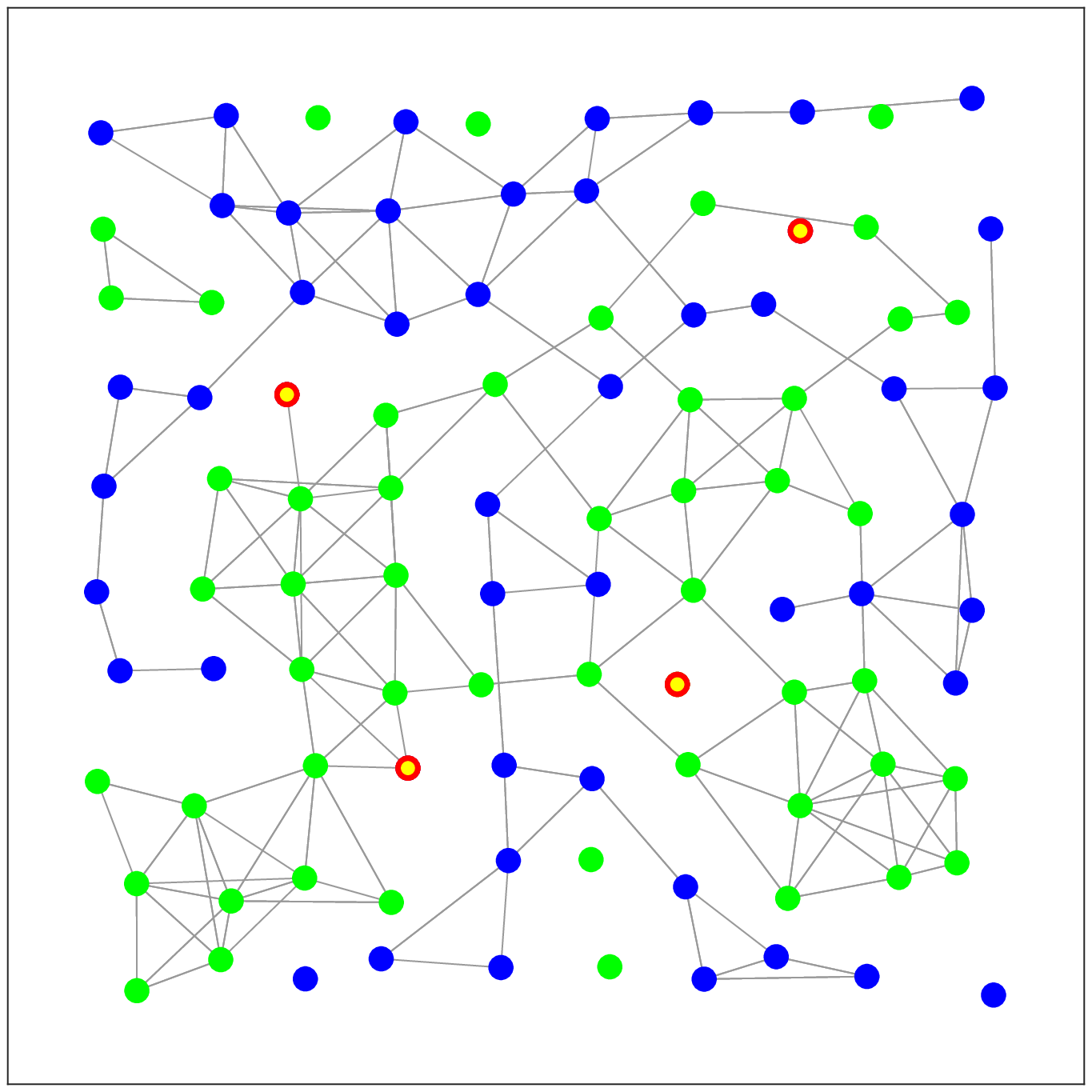}
        \vspace{0.1cm}
\caption{R-DLMSAW}\label{fig: black box network topology R-DLMSAW}
    \end{subfigure}
    \caption{Network topologies at the end of simulation \\under weak attack.}\label{fig: network topologies black box attack}
     \end{minipage}%
\begin{minipage}[c]{0.25\linewidth}
\centering
        \includegraphics[width=0.87\textwidth, trim=1cm 1cm 1cm 1cm]{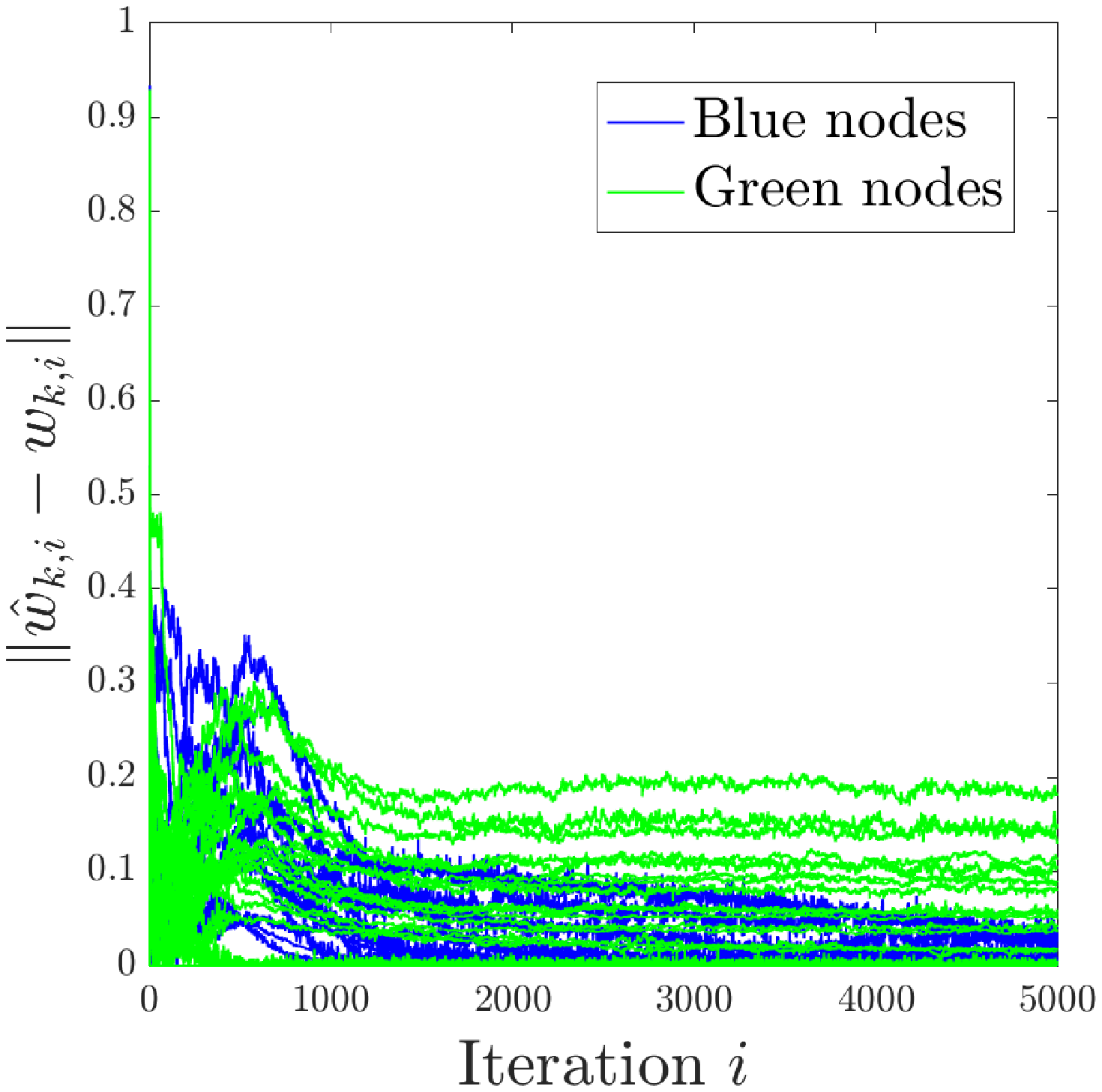}
       \parbox{5cm}{\vspace{0.5cm}\caption{Sate estimation \\precision.} \label{fig: estimation precision}}
    \end{minipage}
\begin{minipage}[c]{0.25\linewidth}
\centering
\includegraphics[width=0.87\textwidth, trim=1cm 1cm 1cm 1cm]{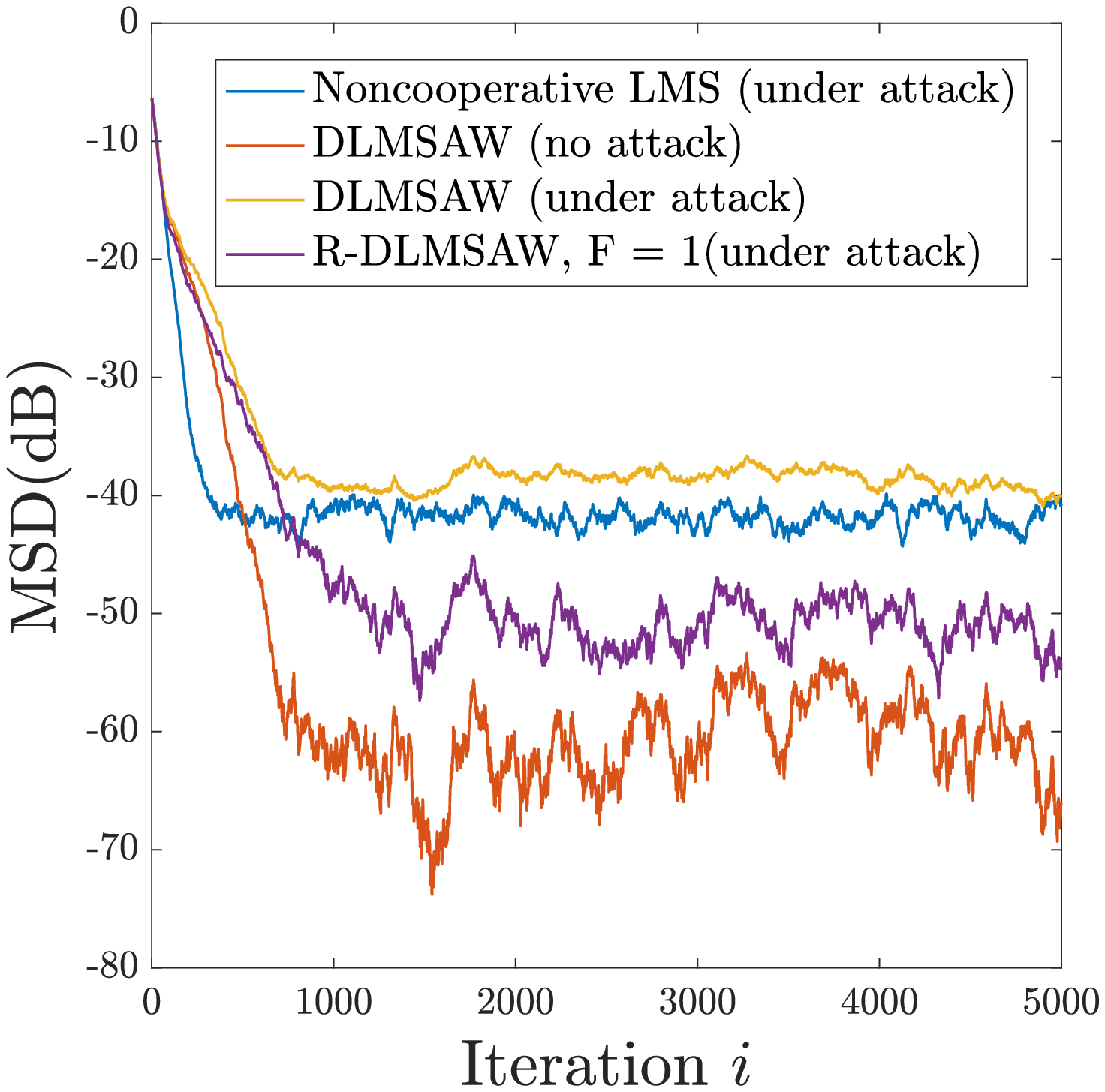}
\parbox{5cm}{\vspace{0.5cm}\caption{\textcolor{revision2}{Steady-state} MSD\\ comparison under weak attack.}\label{fig: MSD comprison black box attack}}
      \end{minipage}%
      %\hspace{1cm}
\end{figure*}
% ------------- Figure Ends ---------------

% ------------- Figure Begins ---------------
\begin{figure*}[!htbp]
\begin{minipage}[c]{0.5\linewidth}
    \centering
    \vspace{0cm} 
\setlength{\abovecaptionskip}{0.1cm}  
    \begin{subfigure}[t]{0.45\textwidth}
        \centering
        \includegraphics[width=0.9\textwidth, trim=1.5cm 1.5cm 1.5cm 1.5cm]{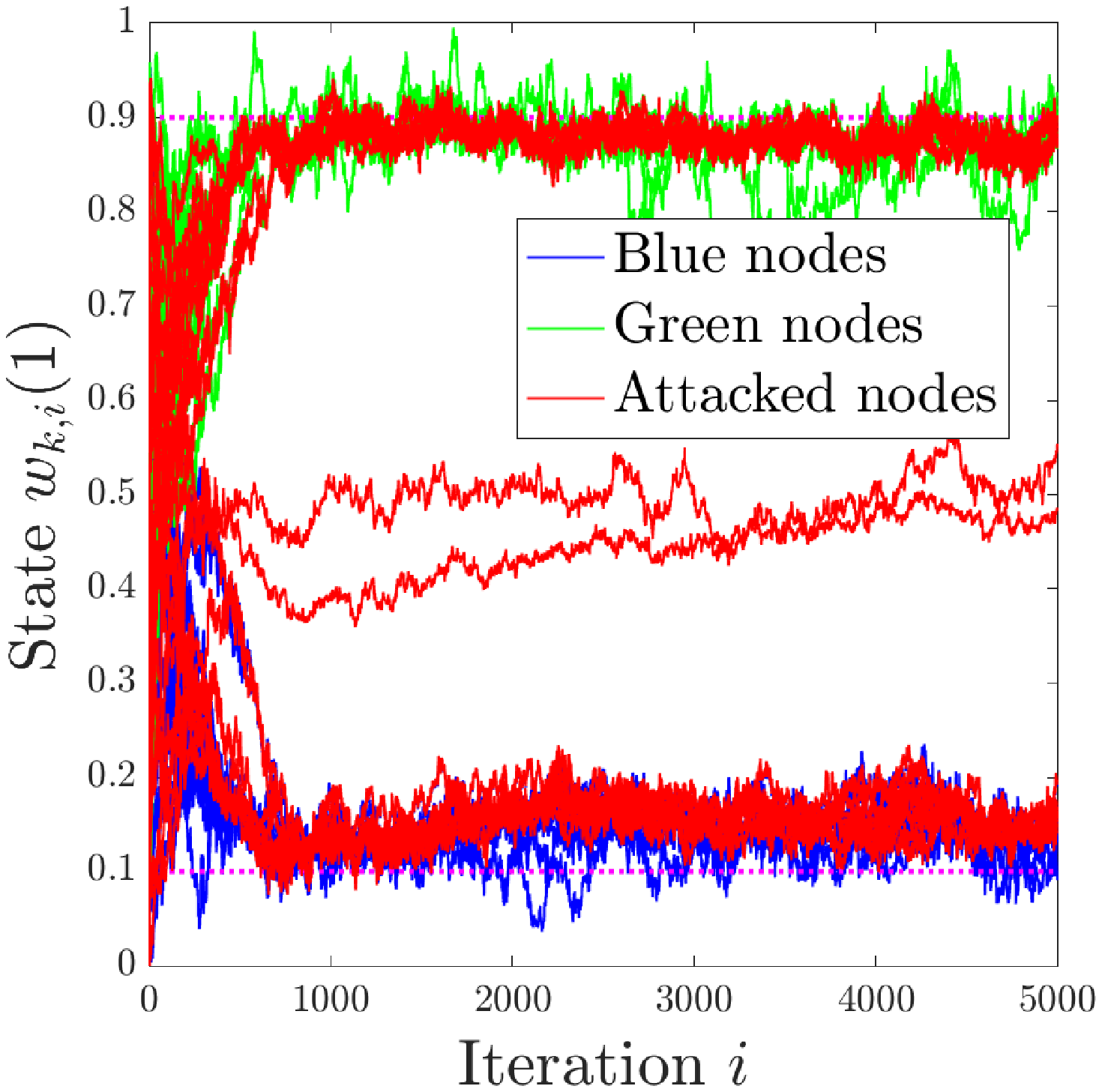}
        \vspace{0.3cm}
        \caption{$\bm{w}_{k,i}(1)$}
        \label{fig: dynamic w1 black box attack DLMSAW}
    \end{subfigure}
    ~  
    \begin{subfigure}[t]{0.45\textwidth}
           \centering
\includegraphics[width=0.9\textwidth, trim=1.5cm 1.5cm 1.5cm 1.5cm]{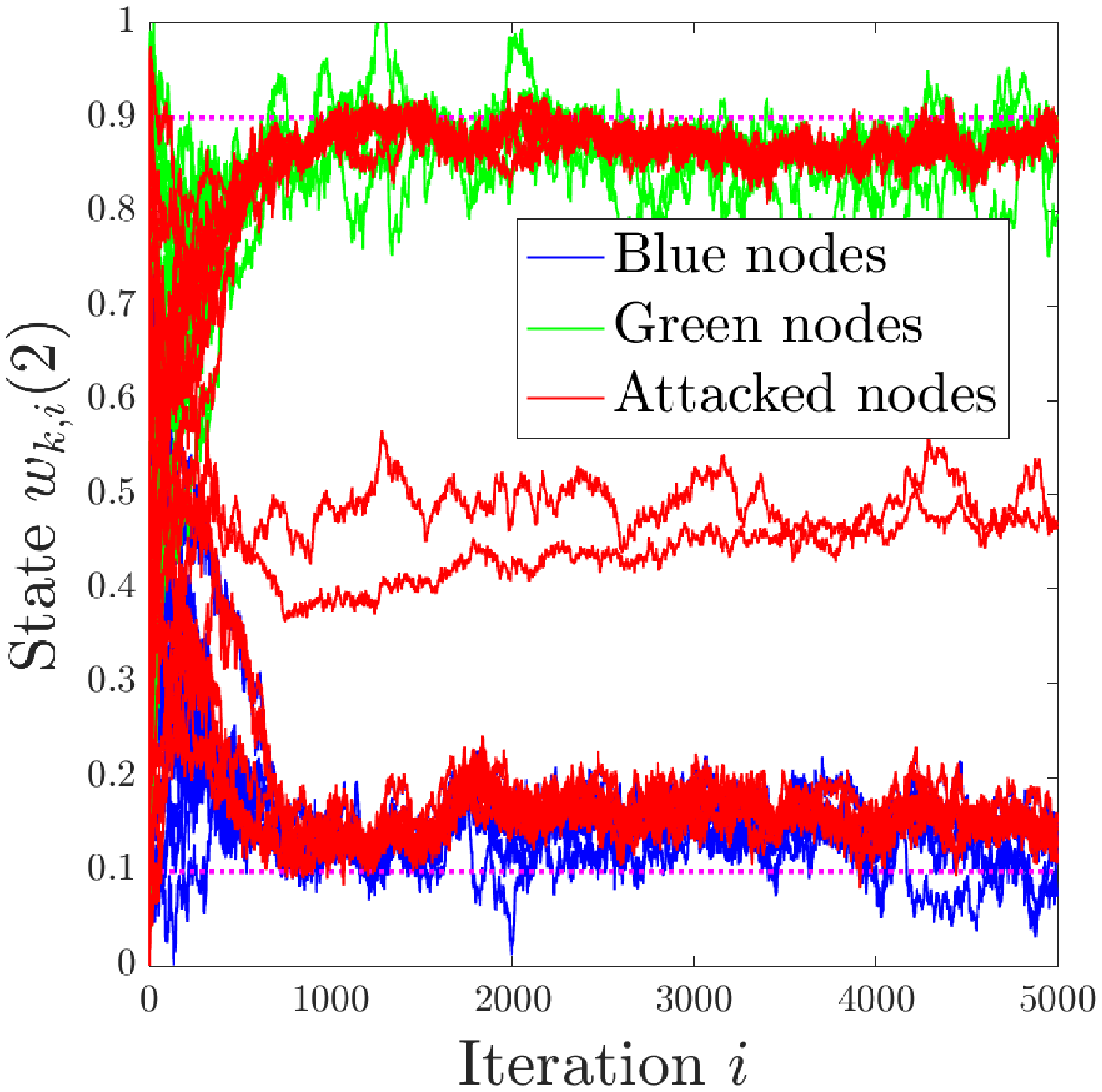}
\vspace{0.3cm}
\caption{$\bm{w}_{k,i}(2)$}\label{fig:dynamic w2 black box attack DLMSAW}
    \end{subfigure}
    \caption{Estimation dynamics for stationary target localization \\by DLMSAW under weak attack.}\label{fig: estimation dynamics black box attack DLMSAW}
      \end{minipage}%
      \begin{minipage}[c]{0.5\linewidth}
    \centering
    \vspace{0cm} 
\setlength{\abovecaptionskip}{0.1cm}  
    \begin{subfigure}[t]{0.45\textwidth}
       \centering
\includegraphics[width=0.9\textwidth, trim=1.5cm 1.5cm 1.5cm 1.5cm]{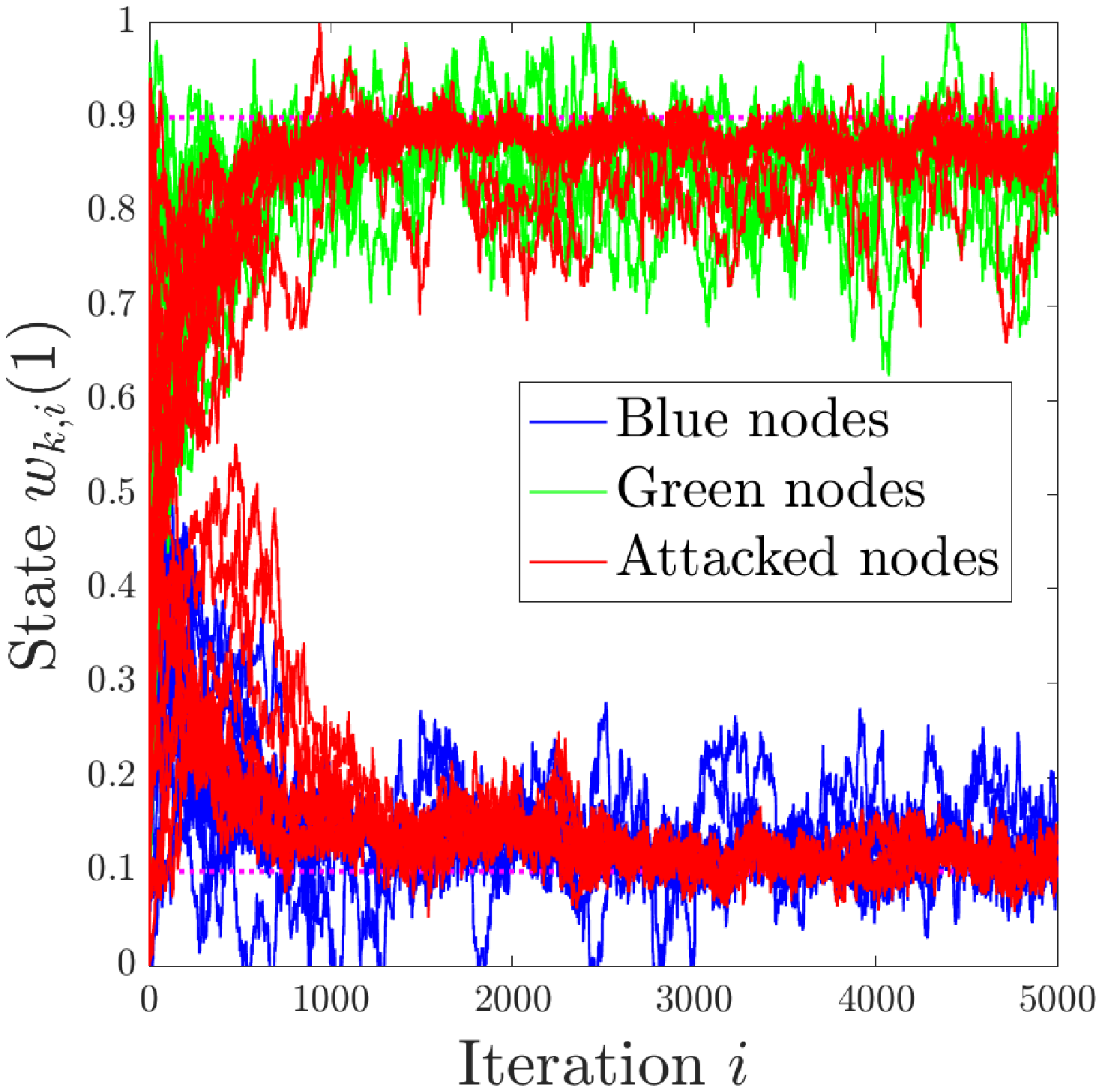}
\vspace{0.3cm}
\caption{$\bm{w}_{k,i}(1) (F=1)$}\label{fig: state w1 estimation dynamics R-DLMSAW}
    \end{subfigure}
    ~  
    \begin{subfigure}[t]{0.45\textwidth}
            \centering
        \includegraphics[width=0.9\textwidth, trim=1.5cm 1.5cm 1.5cm 1.5cm]{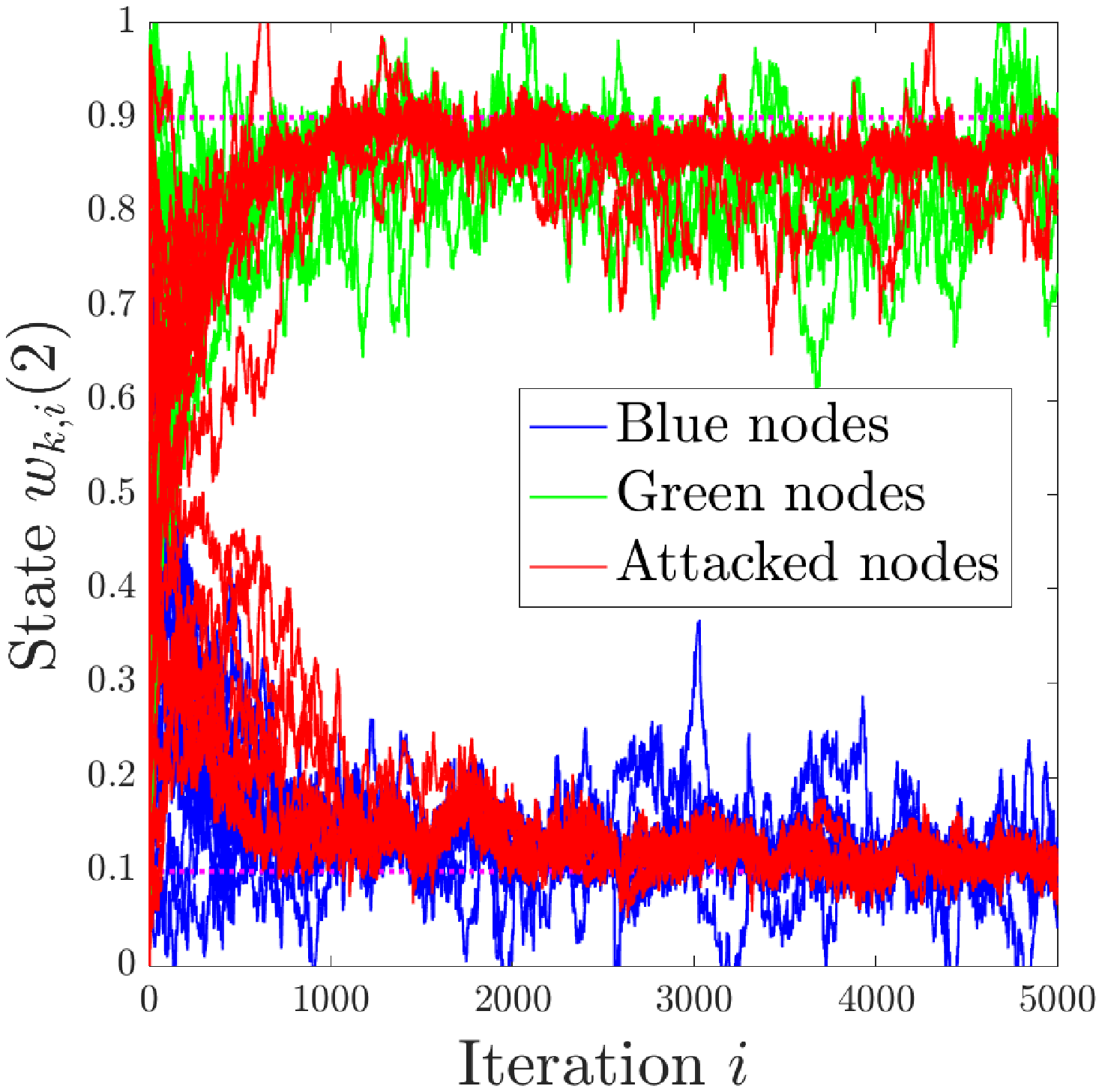}
        \vspace{0.3cm}
\caption{$\bm{w}_{k,i}(2) (F=1)$}\label{fig: state w2 estimation dynamics R-DLMSAW}
    \end{subfigure}
     \caption{Estimation dynamics for stationary target localization \\by R-DLMSAW under weak attack.}\label{fig: estimation dynamics black box attack R-DLMSAW}
     \end{minipage}%
\end{figure*}
% ------------- Figure Ends ---------------

%===================== Section: Related Work ========================
\section{Related Work}\label{sec:related_work}
% distributed algorithms - consensus/diffusion 
Many distributed algorithms are vulnerable to cyber attacks. The existence of an adversarial agent may prevent the algorithm from performing the desired task. Distributed consensus and diffusion based strategies are often employed to resolve distributed estimation and optimization problems, for instance see  \cite{xiao2007distributed,nedic2009distributed,khan2010higher,matei2012consensus,chen2012diffusion,journals/spm/SayedTCZT13}.
%Two main strategies to address distributed estimation/optimization problems are based either on consensus or on diffusion. 
Resilience of consensus-based distributed algorithms in the presence of malicious nodes has received considerable attention in recent years. % XK: Can we add one or two references?   
In particular, the approaches presented in \cite{DBLP:journals/tac/PasqualettiBB12, zhao2018resilient,DBLP:journals/jsac/LeBlancZKS13, DBLP:conf/hicons/LeBlancH14} consider the consensus problem for scalar parameters in the presence of attackers, and resilience is achieved by leveraging high connectivity. Resilient consensus in the case of special network structures, such as triangular networks for distributed robotic applications \cite{TriangularNetworks}, has also been studied. %Resilience has been studied also for triangular networksfor distributed robotic applications \cite{TriangularNetworks}.  
To achieve resilience in sparse networks, \cite{abbas2018improving} presents the idea of employing few trusted nodes, which are hardened nodes that cannot be attacked.
\textcolor{revision}{Resilience for concensus+innovation problems have also been studied by \cite{2017arXiv170906156C, DBLP:conf/cdc/ChenKM18, DBLP:conf/amcc/ChenKM18} in a fully-distributed way via agents' local observations and high network connectivity.}
%\textcolor{revision}{Furthermore, authors in  \cite{mitra2016secure} and \cite{mitra2018distributed} present resilient consensus-based distributed state estimation algorithms in networks with Byzantine adversaries.}
\textcolor{revision}{Resilience can also be achieved via fault detection and isolation (FDI).
For instance, \cite{5779706} studied the FDI problem for linear consensus networks via high connectivity networks and global knowledge of the network structure by each agent.
\cite{SHAMES20112757} considered a similar FDI problem for second-order systems. 
Authors in \cite{yuanchen/AdversaryDetection} presented distributed detection method for consensus+innovation algorithms via local observations of agents only. } 
% attack model and objective
For attacks, typical approaches usually consider Byzantine adversaries \textcolor{revision}{with fixed target different than the true value \cite{DBLP:journals/jsac/LeBlancZKS13} or with updates without time-dependent intention \cite{yuanchen/AdversaryDetection, DBLP:conf/cdc/ChenKM18}} and assume that the goal of the attacker is to disrupt the convergence (stability) of the distributed algorithm. 
In contrast, this work focuses on attacks that do not disrupt convergence but drive normal agents to converge to states selected by the attacker.
\textcolor{revision}{Moreover, in our attack model, the attacker continuously changes its values over time as compared to the fixed value attacks considered previously.}
%\textcolor{revision}{Since the attack objective is  time-dependent, our attack model defers than the previous work in the way that it needs to be designed to effect over time.}

Resilience of diffusion-based distributed algorithms has been studied in \textcolor{revision}{\cite{journals/spm/SayedTCZT13, 6232902,5948418}. Similar to the resilient consensus problems, fixed-value attacks are usually considered, and the main approach has been to use adaptive combination rules to counteract malicious values}. This is an effective measure and has been applied to multi-task networks and distributed clustering problems \cite{6232902}. Several variants focusing on adaptive weights applied to multi-task networks can be found in
\textcolor{revision2}{\cite{7362874,  7060710,  6845334, 7065284}}.
%\textcolor{revision}{\cite{7362874, 7362373,  7060710,  6845334, 7065284}}.
\textcolor{revision}{Note that the essence of adaptive weights is similar to distributed detection. In contrast, it turns the detection method from a binary classification problem to a regression problem. Detection approach has also been applied in \cite{7362874} for clustering over diffusion networks.}
Although adaptive weights provide some degree of resilience to \textcolor{revision}{byzantine adversaries with fixed values}, we have shown in this work that adaptive weights may introduce vulnerabilities that allow \textcolor{revision}{time-dependent} deception attacks. 

%example:target localization
Finally, there has been considerable work on applications of diffusion algorithms that include spectrum sensing in cognitive networks \cite{7086338}, target localization \cite{targetLocalization}, distributed clustering \cite{6232902},  biologically inspired designs \cite{mobileAdaptiveNetworks}. Although our approach can be used for resilience of various applications, we have focused on multi-target localization
\textcolor{revision2}{\cite{6845334}}.%\cite{DBLP:conf/icassp/ChenRS14}. 

% =============== Section : Conclusions ===========
\section{Conclusions}\label{sec:con}
\textcolor{green}{In this paper, we studied distributed diffusion for multi-task networks and investigated vulnerabilities  introduced by adaptive weights. Cooperative diffusion is a powerful strategy to perform optimization and estimation tasks, however, its performance and accuracy can deteriorate significantly in the presence of adversarial nodes. In fact, cooperative diffusion performs significantly better (in terms of \textcolor{revision2}{steady-state} MSD) as compared to non-cooperative diffusion if there are no adversarial nodes. However, with adversaries, cooperative diffusion could be even worse than the non-cooperative diffusion.  To illustrate this, we proposed attack models that can drive normal agents---implementing distributed diffusion (DLMSAW)---to any state selected by the attacker, for both stationary  and non-stationary estimation. We then proposed a resilient distributed diffusion algorithm (R-DLMSAW) to counteract adversaries' effect. \textcolor{revision}{The proposed algorithm always performs at least as good as the non-cooperative diffusion, but if an input parameter $F$ in the algorithm is selected appropriately, it performs significantly better than the non-cooperative diffusion in the presence of adversaries.} % \textcolor{revision}{if the number of compromised neighbors is bounded by $F$, and if $F$ is} selected appropriately, 
% %but if parameters (such as $F$) are selected appropriately, 
% does significantly better than the non-cooperative diffusion in the presence of adversarial nodes.
We also analyzed how the performance of R-DLMSAW changes with the selection of parameter $F$ by the nodes. We evaluated our approach by applying it to stationary and non-stationary multi-target localization. In future, \textcolor{revision}{we are interested in generalizing our model to other types of distributed diffusion algorithms and  with the missing data. \textcolor{revision2}{It is also worth investigating the relationship between the underlying network connectivity and the steady-state performance of such algorithms.} % (probably with missing data) and consensus+innovation problems other than LMS and different adaptive combination rules. 
}}

%To illustrate that, we proposed attack models that can drive normal agents to  any state selected by the attacker, for both stationary  and non-stationary estimation. We then developed a resilient distributed diffusion algorithm to counteract message falsification byzantine attack aiming at making normal agents converge to selected states. Finally, we evaluated our approach by applying it to stationary and non-stationary multi-target localization.

%\input{Acknowledgments}

%=================== Bibliography ======================
\bibliographystyle{IEEEtran}
\bibliography{sigproc_new}

%\newpage

\

%==================== Appendix =======================
%\balancecolumns
%\newpage
%\appendix

\begin{appendices}
\centerline{{\sc Appendix}}
\centerline{{\sc Proof of Lemma 1}}
%\section{Proof of Lemma 1}
The message received by $a$ from $k \in \mathcal{N}_a$ is $\bm{\psi}_{k,i}$. Agent $a$ can compute $\bm{w}_{k,i-1}$ from $\bm{\psi}_{k,i}$ using
\begin{equation*}
\bm{w}_{k,i-1} = \bm{\psi}_{k,i} - \mu_k \bm{u}_{k,i}^* (\bm{d}_k(i) - \bm{u}_{k,i} \bm{w}_{k,i-1}),
\end{equation*}
from which it can compute $\bm{w}_{k,i-1}$ as:
\begin{equation*}
\bm{w}_{k,i-1} = \frac{\bm{\psi}_{k,i} - \mu_k \bm{u}_{k,i}^* \bm{d}_k(i)}{1 - \mu_k \bm{u}_{k,i}^* \bm{u}_{k,i}}.
\end{equation*}
Given the knowledge of $\mu_k$, $\bm{d}_k(i)$, and $\bm{u}_{k,i}$, 
the value $\bm{w}_{k,i-1}$ can be computed exactly.
\\[14pt]
\centerline{{\sc Proof of Lemma 2}}
%\section{Proof of Lemma 2}
We use $\delta_{a,k,i}$ to denote $\|\bm{\psi}_{a,i} - \bm{w}_{k,i-1}\|$, and $\delta_{l,k,i}$ to denote $\|\bm{\psi}_{l,i} - \bm{w}_{k,i-1}\|$, for $l \in \mathcal{N}_k, l \neq a$. 
Since
$$\gamma_{lk}^2(i) = (1-\nu_k)\gamma_{lk}^2(i-1)+\nu_k \| \bm{\psi}_{l,i}-\bm{w}_{k,i-1}\| ^2, l \in \mathcal{N}_k,$$
Suppose the attack starts at $i_a$, then at iteration $(i_a + n)$,
\begin{equation*}
\begin{split}
& \gamma^2_{ak}(i_a + n) \\
=& (1 - \nu_k) \gamma^2_{ak}(i_a + n - 1) + \nu_k \delta_{a,k,i_a+n}^2 \\
=& (1 - \nu_k)((1 - \nu_k) \gamma^2_{ak}(i_a + n - 2) + \nu_k \delta_{a,k,i_a+n-1}^2) \\
&+ \nu_k \delta_{a,k,i_a+n}^2 \\
=& (1 - \nu_k)^{n+1} \gamma^2_{ak}(i_a - 1) \\
&+ \nu_k [(1 - \nu_k)^{n} \delta_{a,k,i_a}^2 + (1 - \nu_k)^{n-1} \delta_{a,k,i_a+1}^2 \\
&+ \ldots + (1 - \nu_k) \delta_{a,k,i_a+n-1}^2 +\delta_{a,k,i_a+n}^2],
\end{split}
\end{equation*}
\begin{equation*}
\begin{split}
\gamma^2_{lk}(i_a + n) =& (1 - \nu_k)^{n+1} \gamma^2_{lk}(i_a - 1) \\
&+ \nu_k [(1 - \nu_k)^{n} \delta_{l,k,i_a}^2 + (1 - \nu_k)^{n-1} \delta_{l,k,i_a+1}^2 \\
&+ \ldots + (1 - \nu_k) \delta_{l,k,i_a+n-1}^2 +\delta_{l,k,i_a+n}^2].
\end{split}
\end{equation*}
For large enough $n$, $(1-\nu_k)^{n+1} \rightarrow 0$. Since we assume $\|\bm{\psi}_{a,i} - \bm{w}_{k,i-1}\| \ll \| \bm{\psi}_{l,i} - \bm{w}_{k,i-1}\|$, i.e., $\delta_{a,k,i} \ll \delta_{l,k,i}$, for $i \geq i_a + n$, 
$\gamma^2_{ak}(i) \ll \gamma^2_{lk}(i)$ holds.
Thus,
\begin{equation}\label{eq:propto weight}
    \frac{a_{lk}(i)}{a_{ak}(i)}\propto \frac{\gamma^{-2}_{lk}(i)}{\gamma^{-2}_{ak}(i)} \rightarrow 0.
\end{equation}
Given the property of weights, \eqref{eq: weight condition} is true.
\\[14pt]
\centerline{{\sc Proof of Lemma 3}}
%\section{Proof of Lemma 3}
We use $\mathcal{A}$ to denote the set of compromised nodes targeting at the same normal node $k$.
The proposed attack strategy results in the following condition holding as proved in \textit{Lemma 2}:
\begin{equation*}
\begin{split}
    &\frac{a_{lk}(i)}{a_{ak}(i)} \rightarrow 0, \qquad  l \in \mathcal{N}_{k} \backslash \mathcal{A}, a \in \mathcal{A},  \\
    &(i \geq i_a + n, \text{ subject to } (1-\nu_k)^{n+1} = 0).
\end{split}
\end{equation*}
Given that $\sum_{l \in \mathcal{N}_k}a_{lk} = 1$, we have
\begin{equation*}
    a_{lk}(i) = 0, a_{ak}(i) = \frac{1}{|\mathcal{A}|}, \qquad  l \in \mathcal{N}_{k} \backslash \mathcal{A}, a \in \mathcal{A}, 
\end{equation*}
\textcolor{green}{where $|\mathcal{A}|$ denotes the number of nodes in $\mathcal{A}$}. Since every compromised node $a \in \mathcal{A}$ sends the same message and is assigned the same weight that sums up to 1, it is equivalent to only one compromised node attacking the target node and being assigned a weight of 1.
Therefore, there is no need for multiple compromised nodes attacking a single normal node.
\\[14pt]
\centerline{{\sc  Proof of Proposition 1}}
%\section{Proof of Proposition 1}
The constraint of $r_k^a$ is consistent with the condition of Lemma \ref{lem:lemma_2}.
Thus, from some point $i$, the state of node $k$ will be attacked as to be:
\begin{equation}\label{eq:attack w}
\begin{split}
\bm{w}_{k,i} &= \bm{w}_{k,i-1} - r_k^a(\bm{w}_{k,i-1} - x_i) \\
& =r_k^a x_{i} + (1 - r_k^a) \bm{w}_{k,i-1}, \\
&(i \geq i_a + n, \text{ subject to } (1-\nu_k)^{n+1} = 0).
\end{split}
\end{equation}

Let $X_{i}$ be $\bm{w}_{k,i}$, $X_{i-1}$ be $\bm{w}_{k,i-1}$, $A_{i}$ be $r_k^a x_{i}$, and $B$ be $(1-r_k^a)$.
Equation \eqref{eq:attack w} turns to:
\begin{equation} \label{convergent proof 1}
X_{i} = A_{i} + B X_{i-1}.
\end{equation}
Assume $\lim_{i \rightarrow \infty} X_{i-1} = X_{i-1}^0$ and $\lim_{i \rightarrow \infty} X_{i} = X_{i}^0$, then for $i \rightarrow \infty$ we get:
\begin{equation} \label{convergent proof 2}
X_{i}^0 = A_{i} + B X_{i-1}^0.
\end{equation}
Subtract \eqref{convergent proof 2} from \eqref{convergent proof 1}, we get
$
X_{i} - X_{i}^0  = B (X_{i-1} - X_{i-1}^0).
$
Let $\varepsilon_i = X_{i} - X_i^0$, for $i = 0, 1, 2, \ldots$, then $\varepsilon_{i} = B \varepsilon_{i-1} = B^2 \varepsilon_{i-2} = \ldots = B^{i} \varepsilon_0$. 
The necessary and sufficient requirement for convergence is
$\lim_{i \to \infty} \varepsilon_{i} = 0$
or, $\lim_{i \to \infty} B^{i} \varepsilon_{0} = 0$,
that is, 
\begin{equation}\label{eq: convergence condition}
    \lim_{i \to \infty} B^{i} = 0.
\end{equation}
Therefore, we get a necessary and sufficient requirement for convergence as $|B| < 1$.
Since $B = 1 - r_{k}^a$, and $r_{k}^a \in (0, 1)$, we get $B \in (0, 1)$. Therefore, $\lim_{i \rightarrow \infty} (X_i - X_i^0) = 0$. The assumption $\lim_{i \rightarrow \infty} X_i = X_i^0$ holds, and therefore, $X_i$ is convergent to $X_i^0$.

To get the value of $X_i^0$, we need to analyze the following two scenarios: stationary state estimation and non-stationary state estimation, separately.
\subsubsection{Stationary state estimation}
In stationary scenarios, the convergence state is independent of time, that is, $X_{i}^0 = X_{i-1}^0 = X^0$. Therefore, equation \eqref{convergent proof 2} turns to:
\begin{equation*}
X^0 = A_i + B X^0.
\end{equation*}
Thus,  $(1-B) X^0 = A_i$, $X^0 = \frac{A_i}{1-B}$. The convergent point is:
\begin{equation*}
\bm{w}_{k,i} = \frac{r_k^a x_{i+1}}{1-(1-r_k^a)} = \frac{r_k^a w_k^a}{1-(1-r_k^a)} = w_k^a = \bm{w}_{k,i}^a, \quad i \rightarrow \infty
\end{equation*}
which realizes the attacker's objective \eqref{eq: objective function}.

\subsubsection{Non-stationary state estimation}
In non-stationary scenarios, we first assume $x_i = w_{k}^a + \theta_{k,i-1}^a$ and later we will show how $\theta_{k,i-1}^a$ turns to  $\theta_{k,i-1}^a + \frac{\Delta \theta_{k,i-1}^a}{r_k^a}$. 

Assume the convergence point $X_i^0$ is a combination of a time-independent value and a time-dependent value, such that $X_i^0 = X^0 + \rho_i$. After taking original values into \eqref{convergent proof 2}, we get
\begin{equation}\label{combination of w/wo}
X^0 + \rho_{i} = r_k^a (w_k^a + \theta_{k,i-1}^a) + (1-r_k^a)(X_0 + \rho_{i-1}).
\end{equation}
\textcolor{green}{Next, we divide \eqref{combination of w/wo} into the time-independent and time-dependent components to get}
\begin{equation*}
X^0 = w_k^a, \quad 
\rho_{i} - \rho_{i-1} = r_k^a (\theta_{k,i-1}^a - \rho_{i-1}).
\end{equation*}
Let $\Delta \rho_{i-1} = \rho_{i} - \rho_{i-1}$, we get:
\begin{equation}\label{eq:rho}
\rho_{i-1} = \theta_{k,i-1}^a - \frac{\Delta \rho_{i-1}}{r_k^a}
\quad \text{ and } \quad
\rho_{i} = \theta_{k,i}^a - \frac{\Delta \rho_{i}}{r_k^a}.
\end{equation}
Thus,
$
\Delta \rho_{i-1} = \rho_{i} - \rho_{i-1} = \theta_{k,i}^a - \theta_{k,i-1}^a - \frac{1}{r_k^a} (\Delta \rho_{i} - \Delta \rho_{i-1})
$.
Let $\Delta \theta_{k,i-1}^a = \theta_{k,i}^a - \theta_{k,i-1}^a$ and $\Delta^2 \rho_{i-1} = \Delta \rho_{i} - \Delta \rho_{i-1}$, then 
$
\Delta \rho_{i-1} = \Delta \theta_{k,i-1}^a - \frac{\Delta^2 \rho_{i-1}}{r_k^a}
\quad \text{ or } \quad 
\Delta \rho_{i} = \Delta \theta_{k,i}^a - \frac{\Delta^2 \rho_{i}}{r_k^a}
$.
If we assume $\frac{\Delta^2 \rho_{i}}{r_k^a} \ll \Delta \theta_{k,i}^a$, then we have $\Delta \rho_{i} = \Delta \theta_{k,i}^a$. 
Therefore, \eqref{eq:rho} can be written as
$
\rho_{i} = \theta_{k,i}^a - \frac{\Delta \theta_{k,i}^a}{r_k^a}.
$
Thus, the dynamic convergence point for $k$ is
\begin{equation*}
\bm{w}_{k,i} = w_k^a + \theta_{k,i}^a - \frac{\Delta \theta_{k,i}^a}{r_k^a}, \qquad i \rightarrow \infty.
\end{equation*}
This means when sending ${\bm{\psi}}_{a,i} = \bm{w}_{k,i-1} + r_k^a (w_k^a + \theta_{k,i-1}^a - \bm{w}_{k,i-1})$ as the communication message, the compromised node $a$ can make $k$ converge to $w_k^a + \theta_{k,i}^a - \frac{\Delta \theta_{k,i}^a}{r_k^a}$. To make agent $k$ converge to a desired state $w_k^a + \Omega_{k,i}^a$, we assume the message sent is
\begin{equation*}
{\bm{\psi}}_{a,i} = \bm{w}_{k,i-1} + r_k^a (w_k^a + m_{i-1} - \bm{w}_{k,i-1}).
\end{equation*}
The corresponding convergence point will be 
$ w_k^a + m_{i} - \frac{\Delta m_{i}}{r_k^a}$. We want the following equation to hold, 
\begin{equation}\label{solve convergence point}
w_k^a + m_{i} - \frac{\Delta m_{i}}{r_k^a} = w_k^a + \Omega_{k,i}^a.
\end{equation}
Assuming $\Delta^2 m_i \rightarrow 0$, the solution of \eqref{solve convergence point} is: $m_i = \Omega_{k,i}^a + \frac{\Delta \Omega_{k,i}^a}{r_k^a}$, meaning to make $k$ converge to a desired state $w_k^a + \Omega_{k,i}^a$, the compromised node $a$ should send communication message:
\begin{equation*}
{\bm{\psi}}_{a,i} = \bm{w}_{k,i-1} + r_k^a (w_k^a + \Omega_{k,i-1}^a + \frac{\Delta \Omega_{k,i-1}^a}{r_k^a} - \bm{w}_{k,i-1}).
\end{equation*}

Thus, to make $k$ converge to $w_k^a + \theta_{k,i}^a$, the compromised node $a$ should send communication message:
\begin{equation*}
{\bm{\psi}}_{a,i} = \bm{w}_{k,i-1} + r_k^a (w_k^a + \theta_{k,i-1}^a + \frac{\Delta \theta_{k,i-1}^a}{r_k^a} - \bm{w}_{k,i-1}).
\end{equation*}
The convergence point is:
\begin{equation*}
\bm{w}_{k,i} = w_k^a + \theta_{k,i}^a = \bm{w}_{k,i}^a, \qquad i \rightarrow \infty,
\end{equation*}
which realizes the attacker's objective \eqref{eq: objective function}.

We can verify the convergence point by putting $x_i = w_k^a + \theta_{k,i-1}^a + \frac{\Delta \theta_{k,i-1}^a}{r_k^a}, \bm{w}_{k,i} = w_k^a + \theta_{k,i}^a, \bm{w}_{k,i-1} = w_k^a + \theta_{k,i-1}^a$ back into equation \eqref{eq:attack w}, we get:
\begin{equation*}
\begin{split}
w_k^a + \theta_{k,i}^a &= r_k^a (w_k^a + \theta_{k,i-1}^a + \frac{\Delta \theta_{k,i-1}^a}{r_k^a}) + (1 - r_k^a)  (w_k^a + \theta_{k,i-1}^a)\\
\theta_{k,i}^a &= r_k^a (\theta_{k,i-1}^a + \frac{\Delta \theta_{k,i-1}^a}{r_k^a}) + (1 - r_k^a)  \theta_{k,i-1}^a\\
\theta_{k,i}^a &= \theta_{k,i-1}^a + \Delta \theta_{k,i-1}^a.
\end{split}
\end{equation*}
The resulting equation holds, illustrating the validity of the convergence state.

\end{appendices}

\end{document}